\documentclass[12pt,a4paper,oneside, IEEEtran]{amsart}

\usepackage{amsmath}
\usepackage{amssymb}
\usepackage{amsthm}
\usepackage{graphicx}
\usepackage{xcolor}
\usepackage{esint}
\usepackage{enumitem}
\usepackage{epstopdf}
\usepackage[utf8]{inputenc}
\usepackage[ruled,vlined]{algorithm2e}
\usepackage{geometry}
\usepackage{float}
\usepackage[T1]{fontenc}
\usepackage{xfrac}
\usepackage{bbm}
\usepackage{cancel}
\usepackage{parskip}
\usepackage{relsize}
\usepackage[normalem]{ulem}
\usepackage[english]{babel}
\usepackage{textcomp}
\usepackage{mathrsfs}
\usepackage{upgreek}
\usepackage{stackengine}
\usepackage{verbatim}
\usepackage[font=small,labelfont=bf]{caption}
\usepackage{subcaption}
\usepackage{hyperref}
\hypersetup{
	colorlinks,
	citecolor=black,
	filecolor=black,
	linkcolor=black,
	urlcolor=black
}
\usepackage{indentfirst}
\usepackage{booktabs}
\usepackage[T2A,T1]{fontenc}
\usepackage[foot]{amsaddr}
\usepackage{setspace}
\usepackage{array}
\newcolumntype{M}[1]{>{\centering\arraybackslash}m{#1}}

\makeatletter
\renewcommand{\@seccntformat}[1]{\csname the#1\endcsname.\quad}
\makeatother

\DeclareMathOperator*{\argmin}{arg\,min}

\DeclareMathOperator*{\arginf}{arg\,inf}
\newcommand{\norm}[1]{\left\lVert#1\right\rVert}

\newcommand{\ex}{\mathbb{E}}

\newcommand{\real}{\mathbb{R}}

\newcommand{\mesp}{\mathbb{P}}

\mathchardef\mhyphen="2D

\newtheorem{theorem}{Theorem}[section]
\newtheorem{proposition}[theorem]{Proposition}

\theoremstyle{remark}
\newtheorem{remark}[theorem]{Remark}

\theoremstyle{definition}
\newtheorem{definition}[theorem]{Definition}
\newtheorem{problem}[theorem]{Problem}

\allowdisplaybreaks

\begin{document}
	
	\title{Clustering Market Regimes using the Wasserstein Distance}
	\author{B. Horvath$^{1,3,4}$, Z. Issa$^1$, and A. Muguruza$^2$}
	\address{%
		$^1$King's College, London\\%
		$\mathtt{zacharia.issa@kcl.ac.uk}$.
		$^2$Imperial College London and Kaiju Capital Management\\%
		$\mathtt{aitor.muguruza\mhyphen gonzalez@kcm.vg}$.
		$^3$TU Munich\\%
		$\mathtt{blanka.horvath@tum.de}$.
		$^4$Alan Turing Institute.
	}
	\date{\today}

	\maketitle
	\begin{abstract}
		
		The problem of rapid and automated detection of distinct market regimes is a topic of great interest to financial mathematicians and practitioners alike. In this paper, we outline an unsupervised learning algorithm for clustering financial time-series into a suitable number of temporal segments (market regimes). 
		As a special case of the above, we develop a robust algorithm that automates the process of  classifying market regimes. The method is robust in the sense that it does not depend on modelling assumptions of the underlying time series as our experiments with real datasets show. This method -- dubbed the Wasserstein $k$-means algorithm -- frames such a problem as one on the space of probability measures with finite $p^\text{th}$ moment, in terms of the $p$-Wasserstein distance between (empirical) distributions. We compare our WK-means approach with a more traditional 
		clustering algorithms by studying the so-called maximum mean discrepancy scores between, and within clusters. In both cases it is shown that the WK-means algorithm vastly outperforms all considered competitor approaches. We demonstrate the performance of all approaches both in a controlled environment on synthetic data, and on real data.
	\end{abstract}
	\begin{spacing}{0.5}
	\tableofcontents
	\end{spacing}
	\singlespacing
	
	\section{Introduction}\label{sec:intro}
	
	Time series data derived from asset returns are known to exhibit certain properties, termed \emph{stylised facts}, that are ubiquitous across asset classes. For example, it is well-understood that return series are non-stationary in the strong sense, and exhibit volatility clustering (for a full recount, see \cite{cont2001empirical}). In particular, understanding the heteroskedastic nature of financial time series data  is relevant for market practitioners. An observed sequence of asset returns (or, for multiple assets, a tuple of sequences) exhibits periods of similar behaviour, followed by potentially distinct periods that indicate a significantly different underlying distribution. Such periods are often referred to as \emph{market regimes}. Our interest in swiftly and accurately detecting changes in market regimes is motivated by a multitude of financial applications (both classical and modern): Most naturally, an accurate detection of shifts in market behaviour is tantamount for making optimised investment decisions or trading strategies. But also within the arena of recent, deep learning-based methods for pricing hedging and market generation \cite{horvath2019deep, buehler2020data}, the detection of significant shifts in market behaviour is a central tool for their model governance, since it serves as an indicator for the need to retrain the ML-model.
	Henceforth, we call the task of finding an effective way of grouping different regimes the \emph{market regime clustering problem (MRCP)}.
	
	In this paper, we propose a methodology to classify segments of the historical evolution of market returns into distinct regimes. We do so by devising a modified, versatile version of the classical $k$-means clustering algorithm to group distributions of asset returns into regimes, which exhibit a higher degree of homogeneity. The way modify the classical algorithm is twofold: Firstly, by a shift of perspective, we consider the clustering problem as one on the space of distributions with finite $p^{\text{th}}$ moment, as opposed to one on Euclidean space. Secondly, our choice of metric on this space is the $p^\text{th}$ Wasserstein distance, and we aggregate nearest neighbours using the associated Wasserstein barycenter. We motivate why the Wasserstein distance is the natural choice for this particular problem in Section \ref{subsec:wassmotivation}. Accordingly in later sections we also present different numerical setups to demonstrate how to navigate how reactive/robust the algorithm is on different datasets. The latter will depend on the precise application at hand and the  modellers appetite for more swift or more robust indicators.
	
	We benchmark our results with two alternative approaches. The first applies the classical $k$-means algorithm to the first $p$ moments associated to each segment of market returns. The second is a more classical approach: We implement a version of the hidden Markov model (HMM) using appropriate modifications to make results comparable to ours. We test each algorithm both on real and synthetic data. The success of the unsupervised learning algorithms in the real data setting is evaluated using a marginal maximum mean discrepancy (MMD) metric, a metric arising from a powerful two-sample test which has experienced a notable rise in popularity in recent machine learning literature. Overall, we show that our data-driven and non-parametric methodology accurately partitions financial return series into distinct clusters which are both distinct from each other whilst remaining internally \emph{self-similar}, relative to a more naive approach based on moments or HMMs, which we verify on synthetic parametric data as well as on historical time series, where our algorithm correctly identifies all (now historically known) periods of unusual market activity as they arise.
	
	\subsection{The market regime clustering problem}\label{subsec:marketregimes}
	Let us start by giving an overview of related literature and existing insights on the MRCP. Recall that the MRCP is defined as the task of classifying segments of return series $(r_i)_{i\ge 0}$, where
	\begin{equation*}
		r_i = (r^1_i, \dots, r^n_i) \text{ for } n \in \mathbb{N}.
	\end{equation*}
	Any vector $r_i \in \mathbb{R}^n$ can be associated to an empirical measure $\delta_{r_i} = \frac{1}{n}\sum_{j=1}^n \delta_{r^j_i}$ for $i\ge 0$ with $n$ atoms. Thus, the problem of classifying market regimes is equivalent to assigning labels to probability measures $\mu \in \mathcal{P}_p(\real)$, where $\mathcal{P}_p(\mathbb{R})$ is the set of probability measures on $\mathbb{R}$ with finite $p^{\text{th}}$ moment. 
	
	Historically, approaches to solving the MRCP vary depending on their applications: for instance, a modeller may be interested in regime classification over different time scales: from microseconds, to months or years. Further variations may come according to the modeller's choice of framing the problem; for example, one may wish to segment regimes via change points detection methods, which are place (as the name sugests) emphasis on change-points rather than on the regimes themselves (see \cite{niu2016multiple} for an overview). Another example is the more general outlier detection problem \cite{cochrane2020anomaly, kondratyev2020data}, which is a special one-class case of the MRCP. Those studying such a problem are often more interested in identifying anomalous datum, as opposed to characterising the distribution $\mu \in \mathcal{P}(\real)$ such datum are generated from. 
	
	Some of the early attempts at analysing financial return series to extract regime switching signals fall under the umbrella of \emph{technical analysis}, where signals such as temporal moving average crossovers and support level breaches are used as indicators for a regime change for a particular financial asset or a collection of assets, see Achelis \cite{achelis2001technical} for a comprehensive guide. More rigorous statistical analyses of financial time series have also been employed to detect regime changes. A classical approach is performing dynamic PCA on inter-asset covariance matrices, see Pelletier \cite{PELLETIER2006445} for an example on high-dimensional synthetic Markovian asset price paths.
	
	Another classical approach to the MRCP is via Markovian switching models and HMMs.  To the author's knowledge, this technique was introduced in the work of Hamilton \cite{hamilton1989new}, in which the state term signifying the regime is given by an AR$(1)$ process. For a more detailed history of the Markov switching model, we refer the reader to \cite{kuan2002lecture} and to \cite{lange2009introduction}, \cite{guidolin2011markov} for their use in empirical finance. The hidden Markov model approach is not altogether model-free as it makes two main assumptions: first, the latent state variable specifying the current regime is Markovian, and secondly, that the likelihood of observing a return given the latent state variable is given by some parametric distribution, often Gaussian. Other approaches include agent-based models as shown in Lux and Marchesi \cite{lux2000volatility}, Bayesian approaches as seen in Maheu, Thomas and Song \cite{bullbearregimes}, or more data-driven approaches as seen in Lahmiri \cite{lahmiri2016clustering}.
	
	Previous work on the problem of clustering families of distributions via non-parametric unsupervised learning approaches has been found in Nielson, Nock and Amari \cite{nielsen2014clustering}, where a modification of the classical kmeans algorithm is used to cluster histograms via mixed $\alpha$-divergences. An approach to empirical distribution clustering via $k$-means is also given in Henderson, Gallagher and Eliassi-Rad \cite{henderson2015ep} in a non-financial context. Other works have utilised the Wasserstein distance for clustering problems, see for instance \cite{li2008real, ye2017fast}, where in the latter distributions are represented as weight-mass pairs, and clustering is considered in the context of images and documents, or \cite{mi2018variational} for an approach using variational optimal transport. Such approaches are similar to the work in this paper as they often employ classic unsupervised learning algorithms with some modification that allows them to handle distributional datum. Our approach seeks to meld the financial regime-switching and clustering worlds together in an attempt to provide a robust, non-parametric approach to \emph{a posteriori} market regime clustering.
	
	\subsection{Motivation for using $\mathcal{W}_p$}\label{subsec:wassmotivation}
	In this section, we motivate our perspective on the clustering problem and the choice of the Wasserstein metric (\ref{eqn:wasserstein}) for this purpose, which has recently seen a swift rise in artificial intelligence and machine learning, see for instance \cite{ni2020conditional}.
	
	Given that our clustering problem is defined over the space of probability measures $\mathcal{P}_p(\real)$, there exist many candidate  one could employ: classical choices include the Kolmogorov-Smirnov statistic or the Kullback-Leibler (KL)-divergence. We argue that either of these choices are inadequate for the given problem statement. It is well known that the Kolmogorov-Smirnov statistic lacks the sensitivity required to distinguish between elements of the clustering set (see \cite{mason1983modified}). The KL-divergence has been employed in the literature as a distance function within a clustering algorithm see \cite{ackermann} for a $k$-medians implementation with generalised Bregman divergences, and \cite{nielsen2014clustering} for general $\alpha$-divergences. We note that our clustering problem is over empirical measures and thus requires an estimation of the probability density function associated to each measure if the KL-divergence is to be employed. Barycenters with respect to the symmetrized and non-symmetrized versions of the KL-divergence have been shown to exist (\cite{veldhuis2002centroid}), and \cite{nielsen2009sided} for more general Bregman divergences), however these rely on algorithmic derivations. We argue that our approach is much simpler, elegant and scalable when compared to the alternatives.

	The Wasserstein distance is a natural choice for comparing distributions of points on a metric space $(X, d)$. The distance function $d$ appears in the expression (\ref{eqn:wasserstein}) characterising the distance and furthermore, the Wasserstein distance metrizes weak convergence\footnote{A sequence $(\mu_n)_{n\ge 1} \subset \mathcal{P}_p(X)$ converges weakly to $\mu \in \mathcal{P}_p(X)$ iff $\mathcal{W}_p(\mu_n, \mu) \to 0$ as $n\to \infty$. Moreover,  if $(X, d)$ is a Polish space, then $(\mathcal{P}_p(X), \mathcal{W}_p)$ is also Polish (see \cite{ambrosio2013user}, Theorem 2.6 for the case $p=2$).}.  Thus, measures that are close in the Wasserstein sense are also close in the classical narrow sense as well. 
	
	In our examples we focus on the univariate case $d=1$, since computing the Wasserstein distance in between empirical measures is particularly tractable, though a multivariate characterisation of our results is also possible (see next paragraph). From Proposition \ref{prop:wassrep}, the algorithm to compute the $p$-Wasserstein distance between two empirical measures with $N$ atoms is $\mathcal{O}(N \log N)$. It is important to note that other comparative distances share this property. However, the Wasserstein distance also has a natural aggregator candidate in the Wasserstein barycenter, which is fast to calculate in the case where $d=1$. Other comparative distances either do not have a natural candidate for aggregation (Kolmogorov-Smirnov) or a canonical aggregator which is tractable to calculate (KL-divergence).
	
	In the case where $d > 1$, the Wasserstein distance is also viable as a choice of metric on $\mathcal{P}_p(\mathbb{R}^d)$. Although $\mathcal{W}_p$ has shown to be an effective tool to tackle the curse of dimensionality associated to sequences on $\real$, it becomes computationally too demanding to solve when extending to higher-dimensional data (\cite{rabin2011wasserstein}, Section 2.2). Recently, the sliced Wasserstein distance \sout{$\mathcal{W}_p$} as seen in \cite{rabin2011wasserstein}, \cite{bonneel2015sliced} and \cite{kolouri2019generalized} has been employed to extend the $\mathcal{W}_p$ to $\mathbb{R}^d$. It does this by projecting distributions $\mu, \nu \in \mathcal{P}_p(\mathbb{R}^d)$ onto $\mathbb{R}$ via points on the unit sphere in $\mathbb{R}^d$, and returns the expected one-dimensional Wasserstein distance of such projections. This turns the complex problem of calculating $\mathcal{W}_p(\mu, \nu)$ into an $\mathcal{O}(MN \log N)$ operation, where $M$ is the number of points taken. Recently, Bayraktar and Guo \cite{bayraktar2021strong} showed that the max sliced Wasserstein metric $\overline{\mathcal{W}}_p$ is strongly equivalent to the classical Wasserstein distance in the cases $p=1, 2$.\footnote{The authors would like to thank Claude Martini and Frédéric Patras for informing us of this.}
		
	\subsection{Problem setting and notation}\label{subsec:problemsetting}
	We begin by giving an overview of the problem setting. First, we introduce the notion of a data stream.
	\begin{definition}[Set of data streams, \cite{DBLP:journals/corr/abs-1905-08494}, Definition 2.1]\label{def:streamofdata}
		Let $\mathcal{X}$ be a non-empty set. The set of \emph{streams of data} $\mathcal{S}$ over $\mathcal{X}$ is given by 
		\begin{equation}\label{eqn:setofstreams}
			\mathcal{S}(\mathcal{X}) = \{\mathsf{x} = (x_1, \dots, x_n) : x_i \in \mathcal{X}, n \in \mathbb{N} \}.
		\end{equation}
	\end{definition}
	In this paper, we take $\mathcal{X} = \mathbb{R}$ and fix $N \in \mathbb{N}$. In the context of the MRCP, elements $S = (s_0,\dots, s_N) \in \mathcal{S}(\real)$ will be price paths associated to a financial asset.
	
	Given $S \in \mathcal{S}(\real)$, we define the vector of log-returns $r^S$ associated to $S$ by 
	\begin{equation}\label{eqn:logreturns}
		r^S_i = \log(s_{i+1}) - \log(s_{i}) \qquad \text{for }0\le i \le N-1,
	\end{equation}
	so $r^S \in \mathcal{S}(\real)$. We use the following expression to highlight that we may to wish partition the original stream of data (\ref{eqn:logreturns}) into potentially overlapping segments equal length. 
	\begin{definition}[Stream lift, \cite{DBLP:journals/corr/abs-1905-08494}, Section 3.3]\label{def:streamlift}
		Let $\mathcal{S}(\mathcal{X})$ be a space of streams over a non-empty set $\mathcal{X}$. Let $\mathcal{V}$ be another non-empty set, and let $v \ge 1$. We call a function 
		\begin{equation*}
			\ell = (\ell^1, \dots, \ell^v): \mathcal{S}(\mathcal{X}) \to \mathcal{S}(\mathcal{S}(\mathcal{V}))
		\end{equation*}
		a \emph{lift} from the space of streams to the space of streams of segments over $\mathcal{V}$. 
	\end{definition}
	Thus, for $\mathsf{x} \in \mathcal{S}(\real)$ and $h_1, h_2 \in \mathbb{N}$ with $h_1 > h_2$, we define a lift $\ell:= \ell_{h_1, h_2}$ from $\mathcal{S}(\real)$ to $\mathcal{S}(\mathcal{S}(\real))$ via
	\begin{equation}\label{eqn:streamlift}
		\ell^i(\mathsf{x}) = (x_{1 + h_2(i-1)}, \dots, x_{1 + h_1 + h_2(i-1)}) \qquad \text{for }i=1,\dots, M,
	\end{equation} 
	where $M := \lfloor \tfrac{N}{h_1-h_2}\rfloor$ is the maximum number of partitions with length $h_1$ that can be extracted from $\mathsf{x} \in \mathcal{S}(\real)$ with sliding window offset parameter $h_2$. We obtain the stream of segments by applying $\ell$ to $r^S$. 
	
	Finally, as stated in the introduction, the main idea of this paper is to lift the regime clustering problem from one on Euclidean space to one on the space of probability distributions with finite $p^\text{th}$ moment with $p>1$. This requires defining a family of measures via the map $\ell(r^S)$. 
	\begin{definition}[Empirical measure, \cite{10.5555/3086952}, Section 3.9.5]\label{def:empiricalmeasures}
		Let $\mathsf{x} \in \mathcal{S}(\real)$ such that $\mathsf{x} = (x_1, \dots, x_N)$ for $N \in \mathbb{N}$. Furthermore, let
		\begin{equation*}
			Q^j: \mathcal{S}(\real) \to \real
		\end{equation*}
		be the function which extracts the $j^\text{th}$ order statistic of $\mathsf{x}$, for $j=1,\dots, N$. Then, the cumulative distribution function of the \emph{empirical measure} $\mu \in \mathcal{P}_p(\real)$ associated to $\mathsf{x}$ is defined as
		\begin{equation}\label{eqn:empiricalmeasure}
			\mu^{\mathsf{x}}((-\infty, x]) = \frac{1}{N}\sum_{i=1}^N \chi_{\{Q^i(\mathsf{x}) \le x\}}(x),
		\end{equation}
		where $\chi: \real \to [0, 1]$ is the indicator function.
	\end{definition}
	
	Thus, we can associate to each segment of data $r_i \in \ell(r^S)$ the empirical measure $\mu_{i}$ for $i=1,\dots, M$. This gives us a family of measures
	\begin{equation}\label{eqn:clusteringset}
		\mathcal{K} = \left\{(\mu_1, \dots, \mu_M) : \mu_i \in \mathcal{P}_p(\mathbb{R}) \text{ for }i=1,\dots, M\right\}.
	\end{equation}
	It is this family $\mathcal{K}$ which will be the subject of our clustering algorithm.

	\subsection{The $k$-means algorithm}\label{subsec:kmeans}
	
	Suppose $X = \{(\mathsf{x}_1, \dots, \mathsf{x}_N): \mathsf{x}_i \in V\} \in \mathcal{S}(V)$ is a stream of data over a normed vector space $(V, \norm{\cdot}_V)$. We further assume that each $\mathsf{x}_i = (x^i_1, \dots, x^i_d)$ has been standardised coordinate-wise, that is,  
	\begin{equation}\label{eqn:standardisation}
		\ex[\{x^i_j\}_{1\le i\le N}] = 0 \text{ and } \text{Var}(\{x^i_j\}_{1\le i\le N}) = 1\qquad \text{for }j=1,\dots, d.
	\end{equation}
	The \emph{$k$-means clustering algorithm} is a unsupervised vector quantization method which assigns elements of $X$ to $k$ distinct clusters. Each of these clusters are defined by central elements $\overline{\mathsf{x}} := \{\overline{\mathsf{x}}_j\}_{j=1,\dots, k}$ called \emph{centroids}, which are initially sampled from $X$. 
	
	At each step $n \in \mathbb{N}$ of the algorithm, one first calculates the \emph{nearest neighbours} 
	\begin{equation}\label{eqn:kmeansupdate}
		\mathcal{C}^n_l := \left\{\mathsf{x}_i \in X: \argmin_{j=1,\dots, k}d(\mathsf{x}_i, \overline{\mathsf{x}}^{n-1}_j)= l\right\}
	\end{equation} 
	associated to each $\overline{\mathsf{x}}^{n-1}_j$ for $j=1,\dots,k$, where $d: V \times V \to [0, +\infty)$ is the metric induced by the norm on $V$. 
	\begin{remark}
		Classically, one chooses $(V, \norm{\cdot}_V) = (\real^d, \norm{\cdot}_{\real^d})$, but we note here that any normed vector space could be chosen.
	\end{remark}

	Each set $\mathcal{C}^n_l$ is then aggregated into a new centroid $\mathsf{x}^n_l$ for $l=1,\dots, k$ via a function $\alpha: 2^V \to V$, so 
	\begin{equation*}
		\overline{\mathsf{x}}^n_l := \alpha(\mathcal{C}^n_l) \qquad \text{for }l=1,\dots,k.
	\end{equation*}

	For a given a tolerance level $\varepsilon > 0$ and a loss function $l: V^k \times V^k \to [0, +\infty)$, 	the $k$-means algorithm terminates at step $n \in \mathbb{N}$ if the stopping condition 
	\begin{equation*}
		l(\overline{\mathsf{x}}^n, \overline{\mathsf{x}}^{n-1}) < \varepsilon
	\end{equation*}
	is satisfied. The algorithm outputs the final clusters $\mathcal{C}^* = \{\mathcal{C}^n_l\}_{l=1,\dots,k}$ and the $k$ quantizations $\overline{\mathsf{x}}^n = \{\overline{\mathsf{x}}^n_l\}_{l=1,\dots, k}$. We conclude this section with the assumptions associated to the $k$-means algorithm that, if satisfied, will result in uniform and isotropic clustering of a data set $X$. 

	\begin{proposition}[\cite{kanungo2002efficient}]\label{prop:kmeansassumptions}
		Given data $X$, the $k$-means algorithm produces $k$ suitable clusters if the following is true:
		\begin{enumerate}[label={\arabic*.}]
			\item There exist $k$ natural clusters in the data $X$.
			\item Each cluster within $X$ is of roughly equal size.
			\item Within-cluster variation (cf. Definition \ref{def:withinclustervariation}) is uniform. That is, for $\delta_2 >0$ small we have that 
			\begin{equation*}
				\big|\mathrm{WC}(\mathcal{C}_i) - \mathrm{WC}(\mathcal{C}_j)| < \delta_2 \qquad \text{for }i,j=1,\dots, k\text{ and }i\ne j,
			\end{equation*}
			\item Clusters are spherical in shape, so we expect the nearest neighbours $\mathcal{C}_j$ to the $j^\text{th}$ centroid $\overline{\mathsf{x}}_j$ to be contained within a ball $B(\overline{\mathsf{x}}_j, \delta)$ where $\delta > 0$ is uniform across all clusters $j = 1,\dots, k$.
		\end{enumerate}
		If conditions (1)-(4) are satisfied, then optimal clusterings $\mathcal{C}^*$ will be suitable.
	\end{proposition}

	Counterexamples to suitability include forcing $k$ clusters on data with fewer than $k$ natural clusters available. Another classical example is the problem of clustering concentric data $X \subset \real^2$, which violates assumption (4) in Proposition \ref{prop:kmeansassumptions}. 	We note that do exist other clustering algorithms which do not share the drawbacks of $k$-means, in that they do not make assumptions regarding the number or shape of the clusters (hierarchical clustering), nor do they enforce that data points $x_i$ belong to one individual cluster (fuzzy c-means clustering, see for instance \cite{cannon1986efficient}). Exploring other clustering algorithms for the MRCP is a topic for future research. 	
	
	\subsection{The maximum mean discrepancy}\label{subsec:MMD}
	
	Evaluating derived $k$-means clusters on a stream of data $\mathcal{S}(\mathcal{X})$ is typically done by evaluating the final total cluster variation $\mathrm{TC}(\mathcal{C}^*)$ or \emph{inertia} (cf. Definition \ref{def:totalclustervariation}). Here, $\mathcal{C}^*$ are the final clusters obtained from a given run of the $k$-means algorithm. Naturally the value of $\mathrm{TC}(\mathcal{C}^*)$ is dependent on the normed vector space $(V_1, \norm{\cdot}_{V_1})$ one decides to cluster the steam of data $\mathcal{X}$ on (assuming it is feasible, a natural choice may be $V_1=\mathcal{X}$). Of course one could make a different choice $(V_2, \norm{\cdot}_{V_2})$ by transforming sets of datum $A \subset \mathcal{S}(\mathcal{X})$, see Section \ref{subsec:benchmark}. Since the total cluster variation depends on $V$, one cannot use it in evaluation between clusterings on different choices of $V$. 
	
	In this section, we outline an integrable probability metric on the space of distributions called the \emph{maximum mean discrepancy (MMD)}, which will be used as part of a robust methodology for confirming goodness-of-fit of clustered market regimes. The MMD has been shown to be a robust estimator under both dependence and presence of outliers \cite{cheriefabdellatif2021finite} and has been employed frequently in the quantitative finance and machine learning literature, see \cite{buehler2020data}, \cite{alquier2020estimation}, \cite{briol2019statistical}.
	
	We provide a brief introduction here and refer the reader to the literature \cite{gretton2012kernel}, \cite{NIPS2006_3110}, \cite{NIPS2009_3738} for further details. We begin by introducing the following.
	
	\begin{problem}[Two-sample test, \cite{gretton2012kernel}, Problem 1]\label{pro:twosampletest}
		Let $(\mathcal{X}, d)$ be a metric space. Suppose $X$ and $Y$ are independent random variables defined on $\mathcal{X}$. Suppose that $X_{\#}\mesp = \mu$ and $Y_{\#}\mesp = \nu$, where $\mu, \nu \in \mathcal{P}(\mathcal{X})$ are Borel. If we draw samples $x = (x_1, \dots, x_n)$ and $y = (y_1, \dots, y_m)$ where $x_i \sim \mu$ for $i=1,\dots,n$ and $y_j \sim \nu$ for $j =1,\dots, m$, when can we determine if $\mu \ne \nu$? That is, we wish to implement a test for the \emph{two-sample problem}
		\begin{equation}\label{twosampleproblem}
			H_0 : \mu = \nu \text{ against } H_1: \mu \ne \nu.
		\end{equation}
	\end{problem}

	We introduce the following test statistic associated to Problem \ref{pro:twosampletest}. 

	\begin{definition}[Maximum mean discrepancy, \cite{gretton2012kernel}, Definition 2]\label{def:maximummeandiscrepancy}
		Let $\mathcal{F}$ be a class of functions $f: \mathcal{X} \to \real$ and let $\mu, \nu$ be defined as in Problem \ref{pro:twosampletest}. Then, the \emph{maximum mean discrepancy} (MMD) between $\mu$ and $\nu$ is defined as 
		\begin{equation}\label{eqn:mmd_general}
			\mathrm{MMD}[\mathcal{F}, \mu, \nu] := \sup_{f \in \mathcal{F}}\Bigg(\ex_\mu[f(x)] - \ex_{\nu}[f(y)] \Bigg).
		\end{equation}
	 	If $x = (x_1, \dots, x_n)$ and $y = (y_1, \dots, y_m)$ are samples where $x_i \sim \mu$ and $y_j \sim \nu$, then a \emph{biased empirical estimate} of the MMD is given by
	 	\begin{equation}\label{eqn:biased-mmd-general}
	 		\mathrm{MMD}_b[\mathcal{F}, x, y] := \sup_{f \in \mathcal{F}}\Bigg[\frac{1}{n}\sum_{i=1}^n f(x_i) - \frac{1}{m}\sum_{j=1}^m f(y_j) \Bigg]
	 	\end{equation}
	\end{definition}

	Clearly, the value of the MMD between two measures is determined by the function class $\mathcal{F}$ one decides to calculate the supremum in (\ref{eqn:mmd_general}) over; in particular, it is not even guaranteed to be a metric. Often the MMD is employed in the context of studying mean differences between datum in a typically higher-dimensional feature space. This motivates the use of kernel methods to define $\mathcal{F}$, which is often chosen to be the unit ball in a reproducing kernel Hilbert space (RKHS) $(\mathcal{H}, \kappa)$, where $\kappa: \mathcal{X} \times \mathcal{X} \to \mathbb{R}$ is the associated reproducing kernel. That the MMD is a metric depends on the associated kernel $\kappa$ being \emph{universal} (cf. Definition \ref{def:universalkernel}) if $\mathcal{X}$ is compact. If $\mathcal{X}$ is non-compact, the following property is enough to guarantee that this is the case.
	\begin{definition}[Characteristic kernel, \cite{NIPS2008_d07e70ef}, Section 2]\label{def:characteristic}
		Let $\mathcal{X}$ be a non-empty set. A kernel $\kappa$ on $\mathcal{X}$ is called \emph{characteristic} if the mean mapping
		\begin{equation}
			\mu \mapsto \ex_{X\sim \mu}[\kappa(\cdot, X)]
		\end{equation}	
		is injective.
	\end{definition}
	The \emph{Gaussian kernel}
	\begin{equation}\label{eqn:gaussiankernelmain}
		\kappa_G: \mathbb{R}^d \times \mathbb{R}^d \to [0, +\infty), \qquad \kappa_G(x,y) = \exp(-\norm{x-y}^2_{\mathbb{R}^d}/2\sigma^2)	
	\end{equation}
	is characteristic to the set of Borel measures on $\mathcal{X}$ and indeed makes the MMD a metric on $\mathcal{P}(\mathcal{X})$. For more details we refer the reader to Theorem 2 in \cite{fukumizu2007kernel} or to Appendix \ref{appendix:B}, Theorem \ref{theroem:mmdmetric} for more details.
	
	We will use the MMD with $\mathcal{F} = (\mathcal{H}, \kappa_G)$ to validate how effective a given clustering algorithm is in the case that we are unable to infer true regime labels, i.e., when we are working with real data. A key notion to define how similar a collection of samples are to each other is the following.
	
	\begin{definition}[Within-cluster self-similarity (homogeneity)]\label{def:selfsimilarity}
		Let $X \in \mathcal{S}(\mathcal{X})$ be a stream of data with $N$ observations. Let $\mathcal{F}$ be the unit ball in a universal RKHS $\mathcal{H}$. For $n, m \in \mathbb{N}$, we define the \emph{self-similarity score} associated to $\mathcal{X}$ to be
		\begin{equation}\label{eqn:selfsimilarity}
			\mathrm{Sim}(X) = \mathrm{Median}\left((\mathrm{MMD}^2_b[\mathcal{F}, x_i, y_i])_{1\le i\le n} \right),
		\end{equation}
		where $x_i = (x^i_1,\dots, x^i_m)$ and $y_i=(y^i_1, \dots, y^i_m)$ are samples drawn pairwise from $\mathcal{X}$ for $i=1,\dots,n$.
	\end{definition}
	\begin{remark}
		The \emph{true self-similarity score} is obtained by calculating the biased MMD (\ref{eqn:biasedmmd}) for all unique combinations of pairwise samples. For computational reasons, we often calculate (\ref{eqn:selfsimilarity}) from $n << {N \choose 2}$ iterations.
	\end{remark}

	\section{$k$-means on the space of distributions}\label{sec:kmeansdist}
	
	In this section, we outline our modification to the $k$-means algorithm which allows us to cluster the set (\ref{eqn:clusteringset}) directly on the space of probability measures with finite $p^\text{th}$ moment. Central to this paper is the following distance metric on $\mathcal{P}_p(\mathbb{R})$.
	\begin{definition}[$p$-Wasserstein distance, \cite{ambrosio2005gradient}]\label{def:wasserstein}
		Suppose $(X,d)$ is a separable Radon space. The \textit{$p$-th Wasserstein distance} between measures $\mu ,\nu \in \mathcal{P}_p(X)$ is defined by
		\begin{equation}\label{eqn:wasserstein}
			\mathcal{W}^p_p(\mu,\nu):= \min_{\mesp \in \Pi(\mu,\nu)}\left\{\int_{X \times X} d(x,y)^p\,\mesp(dx,dy) \right\},
		\end{equation}
		where
		\begin{equation*}
			\Pi(\mu,\nu) := \{\mesp \in \mathcal{P}(X \times X): \mesp(A\times X) = \mu(A), \ \ \mesp(X\times B) = \nu(B) \}
		\end{equation*}
		is the set of \emph{transport plans} between $\mu$ and $\nu$.
	\end{definition}

	The $p$-Wasserstein distance (\ref{eqn:wasserstein}) is the solution to the Kantorovich-type optimal transportation problem between measures $\mu$ and $\nu$ for the cost function $c(x,y) = d(x,y)^p$. For our applications, existence of an optimal plan $\mesp^* \in \Pi(\mu, \nu)$ realising the Wasserstein distance between measures $\mu, \nu \in \mathcal{P}(\real)$ is guaranteed by continuity of the metric $d(x,y) = |x-y|^p$ and the fact that $\mu, \nu$ will be empirical measures and thus posses compact support. We refer the reader to \cite{santambrogio} for further details.
	
	\begin{remark}[Relationship to the $\mathrm{MMD}$]
		The Wasserstein distance (\ref{eqn:wasserstein}), via its equivalent dual formulation, is a special case of an integral probability metric (see, for instance, \cite{wang2021twosample}, Definition 1). In the case where $p=1$, the dual representation is given by
		\begin{equation}\label{eqn:wassersteindual}
			\mathcal{W}_1(\mu, \nu) = \sup_{f\in \text{Lip}_1(X)}\left\{\int_{X}f\,d(\mu-\nu)\right\},
		\end{equation}
		where $\text{Lip}_1(X)$ denotes the space of continuous $\real$-functions over $X$ with Lipschitz constant $L \le 1$. Since $\mu, \nu$ are probability measures, we can write (\ref{eqn:wassersteindual}) as 
		\begin{equation*}
			\mathcal{W}_1(\mu, \nu) = \sup_{\norm{f}_{\text{Lip}_1} \le 1}\left\{\ex_\mu[f(x)] - \ex_\nu[f(y)]\right\}.
		\end{equation*}
		Thus $\mathcal{W}_1(\mu, \nu)$ is an integrable probability metric over the function class $\mathcal{F}$, which is given by the unit ball in the space of functions
		\begin{equation*}
			\mathrm{Lip}(X) = \{f: X \to \real: f \text{ continuous,} \norm{f}_{\mathrm{Lip}} < + \infty\},
		\end{equation*}
		where
		\begin{equation*}
			\norm{f}_\mathrm{Lip} = \sup_{x\ne y}\frac{|f(x) - f(y)|}{d(x,y)}.
		\end{equation*}
	\end{remark}
	
	In what follows, we choose the $p$-Wasserstein distance to be our metric on $\mathcal{P}_p(\real)$. As explained in Section \ref{subsec:wassmotivation}, this distance is natural and tractable to use on the space of probability measures.
	
	Our next decision we need to make is how we aggregate nearest neighbours $\mathcal{C}_l$ into central elements $\overline{\mu}_l \in \mathcal{P}_p(\real)$ for $l=1,\dots, k$. One of the advantages of choosing the Wasserstein distance $\mathcal{W}_p$ to be the metric we apply on our clustering space is the existence of the following, which gives a natural way to ``average'' a family of measures under $\mathcal{W}_p$.
	\begin{definition}[Wasserstein barycenter]\label{def:wassbary}
		Suppose $(X, d)$ is a separable Radon space and let $\mathcal{K} = \{\mu_i\}_{i\ge 1}\subset \mathcal{P}(X)$ be a family of Radon measures. Define the \emph{$p$-Wasserstein barycenter} $\overline{\mu}$ of $\mathcal{K}$ to be 
		\begin{equation}\label{eqn:wassbary}
			\overline{\mu} = \argmin_{\nu \in \mathcal{P}(X)}\sum_{\mu_i \in \mathcal{K}}\mathcal{W}_p(\mu_i, \nu).
		\end{equation}
	\end{definition}
	\begin{remark}
		If $\{\mu_i\}_{i\ge 1}$ are a family of measures associated to a cluster $\mathcal{C}_l, l=1,\dots, k$, then the Wasserstein barycenter (\ref{eqn:wassbary}) is the measure $\overline{\mu} \in \mathcal{P}_p(\real)$ which minimises the within-cluster variation $\mathrm{WC}(\mathcal{C}_l)$ from Definition \ref{def:withinclustervariation}.
	\end{remark}

	Since the $k$-means algorithm requires repeated evaluations of elements on clustering space under the given metric, tractability of the non-linear optimisation (\ref{eqn:wasserstein}) becomes relevant. In the case where measures $\mu, \nu \in \mathcal{P}(\real)$ are absolutely continuous, there exist a closed-form solution to (\ref{eqn:wasserstein}).
	\begin{proposition}[\cite{kolouri2019generalized}, Equation (3)]\label{prop:wassrep}
		Suppose $\mu, \nu \in \mathcal{P}_p(\real^d)$ and let $d=1$. Moreover, suppose that $\mu, \nu$ are absolutely continuous with respect to the Lebesgue measure on $\mathbb{R}$. Then, the $p$-Wasserstein distance $\mathcal{W}_1(\mu, \nu)$ is given by
		
		\begin{equation}\label{eqn:closedformwasserstein}
			\mathcal{W}_p(\mu, \nu) = \left( \int_0^1 |F_\mu^{-1}(z) - F_\nu^{-1}(z)|^p\,dz\right)^{1/p},
		\end{equation}
		where the quantile function $F_\mu^{-1}: [0, 1) \to \real$ is defined as
		\begin{equation}\label{eqn:quantilefunction}
			F_\mu^{-1}(z) = \inf\{x: F_\mu(x) > z\}.
		\end{equation}
	\end{proposition}
	\begin{proof}
		A consequence of the fact that the (unique) optimal transport map pushing $\mu$ onto $\nu$ is given by $T(x) = (F^{-1}_\nu \circ F_\mu)(x)$, and applying a change of variables.
	\end{proof}

	In what follows, we assume that $\mu, \nu$ are empirical measures with equal numbers of atoms $N \in \mathbb{N}$ (this will be the case in our experimental setup). Recalling Definition \ref{def:empiricalmeasures}, we may write them as 
	\begin{equation}\label{eqn:empiricalmeasuresexamples}
		\mu((-\infty, x]) = \frac{1}{N}\sum_{i=1}^N \chi_{\alpha_i\le x}(x), \qquad \nu((-\infty, x]) = \frac{1}{N}\sum_{i=1}^N \chi_{\beta_i\le x}(x)
	\end{equation}
	where $(\alpha_i)_{1\le i\le N}$ and $(\beta_i)_{1\le i\le N}$ are increasing sequences corresponding to the atoms of $\mu$ and $\nu$. We wish to use (\ref{eqn:closedformwasserstein}) to obtain a closed-form expression for the Wasserstein distance between the two measures. Thus, we must consider the well-posedness of (\ref{eqn:quantilefunction}): every empirical measure on $\mathbb{R}$ is Radon, and thus one can associate to $\mu$ ($\nu$) a right-continuous function of finite variation $A_t: \mathbb{R} \to [0, 1]$ given by $A_t = \mu\left((-\infty, t)\right)$ (\cite{revuz2004continuous}, Theorem 4.3). The function $A_t$ possesses a right-continuous inverse which is nothing but the quantile function from (\ref{eqn:quantilefunction}), which (in the case of $\mu$) can be written as
	
	\begin{equation}\label{eqn:empiricalquantilefunction}
		F^{-1}_\mu(z) = \alpha_i \qquad \text{for all }z \in \left[\frac{i-1}{N}, \frac{i}{N} \right), \qquad i=1, \dots, N.
	\end{equation}

	Moreover $F^{-1}_\mu(z) = 0$ for all $z < \alpha_1$. Applying (\ref{eqn:empiricalquantilefunction}) to (\ref{eqn:closedformwasserstein}), the Wasserstein distance between the empirical measures $\mu$ and $\nu$ is given by
	\begin{align}
		\mathcal{W}_p(\mu, \nu)^p &= \sum_{i=1}^N \int_{\tfrac{i-1}{N}}^{\tfrac{i}{N}}|F^{-1}_\mu(z) - F^{-1}_\nu(z)|^p\,dz\nonumber \\
			&= \frac{1}{N}\sum_{i=1}^N |\alpha_i - \beta_i|^p \label{eqn:wassdistatoms}.
	\end{align}
	Thus, calculating the Wasserstein distance between two empirical measures can be done in linear time, assuming the atoms of each measure are already sorted ascending. If not, calculating (\ref{eqn:wassdistatoms}) is an $\mathcal{O}(N \log N)$ operation, where $N$ is the number of atoms. This representation also makes calculating the Wasserstein barycenter (\ref{eqn:wassbary}) simple in the case where $p=1$ and some assumptions are made on the number of atoms present in each measure.
	\begin{proposition}[Wasserstein barycenter, empirical measures]
		Suppose that $\{\mu_i\}_{1\le i\le M}$ are a family of empirical probability measures, each with $N$ atoms $(\alpha^i_j)_{1\le j\le N} \subset \real^N$. Let 
		\begin{equation*}
				a_j = \mathrm{Median}(\alpha^1_j, \dots, \alpha^M_j) \qquad \text{for } j=1,\dots, N.
		\end{equation*}
		Then, the cumulative distribution function of the Wasserstein barycenter $\overline{\mu} \in \mathcal{P}_1(\real)$ over $\{\mu_i\}_{1\le i\le M}$ with respect to the 1-Wasserstein distance is given by
		\begin{equation}
			\overline{\mu}\left((-\infty, x]\right) = \frac{1}{N}\sum_{i=1}^N \chi_{a_i \le x}(x).
		\end{equation}
		Moreover, $\overline{\mu}$ is not necessarily unique.
	\end{proposition}
	\begin{proof}
		See Appendix \ref{prop:wassersteinbarycenter}.
	\end{proof}
	
	The last specification we need to make is regarding the loss function. We do this in the natural way by replacing the squared Euclidean distance in standard $k$-means by the $p$-Wasserstein distance. Let $\overline{\mu}^{n} = (\overline{\mu}^{n}_i)_{1\le i\le k}$ be the centroids obtained after step $n$ of the Wasserstein $k$-means algorithm. Therefore, our loss function $l : \mathcal{P}_p(\real)^k \times \mathcal{P}_p(\real)^k \to [0, +\infty)$ is given by
	\begin{equation}\label{eqn:wasskmeansloss}
		l(\overline{\mu}^{n-1}, \overline{\mu}^n) = \sum_{i=1}^k \mathcal{W}_p(\overline{\mu}^{n-1}_i, \overline{\mu}^n_i).
	\end{equation}
	Our stopping rule is unchanged; that is, for a given $\varepsilon > 0$ we terminate the algorithm at step $n$ if $l(\overline{\mu}^{n-1}, \overline{\mu}^n) < \varepsilon$. We give a full statement of the algorithm with the following. 
	\begin{definition}[WK-means algorithm]\label{def:wassersteinkmeans}
		Let $\mathcal{K} \subset \mathcal{P}_p(\real)$ be a family of measures with finite $p^\text{th}$ moment. We refer to the $k$-means clustering algorithm on $(\mathcal{P}_p(\real), \mathcal{W}_p)$, with aggregation method given by the Wasserstein barycenter from Definition \ref{def:wassbary} and loss function given by (\ref{eqn:wasskmeansloss}) as the \emph{Wasserstein $k$-means algorithm}, or \emph{WK-means}.
 	\end{definition} 
	We summarize with Algorithm \ref{wassersteinkmeansalgo}.
	\begin{algorithm}[h]
		\SetAlgoLined
		\KwResult{$k$ centroids}
		\textbf{calculate} $\ell(r^S)$ given $S$\;
		\textbf{define} family of empirical distributions $\mathcal{K} = \{\mu_j\}_{1\le j\le M}$\;
		\textbf{initialise} centroids $\overline{\mu}_i, i=1,\dots,k$ by sampling $k$ times from $\mathcal{K}$\;
		\While{loss\_function $>$ tolerance}{
			\ForEach{$\mu_j$}{
				\textbf{assign} closest centroid wrt $\mathcal{W}_p$ to cluster $\mathcal{C}_l$, $l=1,\dots,k$;
			}
			\textbf{update} centroid $i$ as the Wasserstein barycenter relative to $\mathcal{C}_l$\;
			\textbf{calculate} loss\_function\; 
		}
		\caption{WK-means algorithm}
		\label{wassersteinkmeansalgo}
	\end{algorithm}	

	\section{Methodology and numerical results}\label{sec:methodology}
	
	In this section, we cover the methods used to test the WK-means algorithm on stock data. Initially, we test both algorithms on real data. Validation of each clustering algorithm was conducted using the MMD test statistic (\ref{eqn:biased-mmd-general}). Finally, we tested both algorithms on synthetic data generated via two different models: one where the associated log-returns were distributed normally, and another where they were not.
	
	\subsection{Alternative clustering algorithms as benchmarks}\label{subsec:benchmark}
I this section we seek to benchmark our approach via two alternative algorithms. In this section, we briefly introduce these methods. 
	
	\subsubsection{$k$-means with statistical moments}
	
	A natural and more classical approach to clustering regimes may involve studying the first $p \in \mathbb{N}$ raw moments associated to each measure $\mu \in \mathcal{K}$. With this in mind, consider the image of $\mathcal{K}$ from (\ref{eqn:clusteringset}) under the function
	
	\begin{equation}\label{eqn:momentfunc}
		\varphi^p(\mu) = \left(\frac{1}{n!}\int_\real x^n\,\mu(dx) \right)_{1\le n\le p},
	\end{equation}
	which is the \emph{truncated} unstandardised $p^{\text{\emph{th}}}$-\emph{moment map}. As each $\mu \in \mathcal{K}$ is a sum of Dirac masses, each element of $\varphi^p(\mu)$ is finite. Thus, for a given $p>1$ we obtain
	\begin{equation}\label{eqn:momentclusteringset}
		\varphi^p(\mathcal{K}) = \left\{\left(\varphi^p(\mu_1), \dots, \varphi^p(\mu_M)\right): \varphi^p(\mu_i) \in \mathbb{R}^p \text{ for } i=1,\dots,M\right\}.
	\end{equation}

	After standardising each element of $\varphi^p(\mathcal{K})$ component-wise (cf. Remark \ref{rmk:momentmagnitude}), we obtain a clustering set on $\mathbb{R}^p$, which we can apply the standard k-means algorithm to. This motivates the following definition. 
	
	\begin{definition}[Moment $k$-means]
		Let $\mathcal{K} \subset \mathcal{P}_p(\real)$ be a family of measures. For $p\ge 1$, associate to each $\mu_i \in \mathcal{K}$ the $\real^p$-vector $\varphi^p(\mu_i)$ for $i=1,\dots, M$, where $\varphi^p : \mathcal{P}_p(\real) \to \real^p$ is the $p$-moment map from  (\ref{eqn:momentfunc}).
		
		Then, \emph{moment $k$-means algorithm}, or \emph{MK-means}, is given by applying Algorithm \ref{standardkmeansalgo} to the stream of data $\varphi^p(\mathcal{K})$ from (\ref{eqn:momentclusteringset}). See Appendix \ref{appendix:A} for more details.
	\end{definition}
	\begin{remark}[Magnitude of moments]\label{rmk:momentmagnitude}
		The function $\varphi^p$ defined in (\ref{eqn:momentfunc}) outputs the first $p$ raw moments associated to a measure $\mu \in \mathcal{P}_p(\real)$. Often, moments that appear earlier in the sequence $(\varphi^p(\mu)_i)_{1\le i\le p}$ will be of significantly larger magnitude than those that appear later. In order for the $k$-means algorithm to not place undue emphasis on these moments, it is critical that each slice $\{\varphi_p(\mu_i)_j\}_{1\le i\le M}$ is standardised according to equation (\ref{eqn:standardisation}) for $1 \le j \le p$.
	\end{remark}
	
	\subsubsection{Hidden Markov model}
	
	As mentioned in Section \ref{subsec:marketregimes}, a more classical approach to market regime clustering involves fitting a \emph{hidden Markov model (HMM)} to observed time series data $\mathsf{x} \in \mathcal{S}(\mathbb{R})$,  that there exist $k \in \mathbb{N}$ hidden latent states $\{1, \dots, k\}$ which govern the dynamics of $\mathsf{x}$. The transition between the latent states is assumed Markovian, and although they are not directly observable they are represented by a transition density $f(x\,|\,z_l, \theta_l)$ where $z_l$ is the given latent state and $\theta_l$ are parameters associated to the state, for $l=1,\dots,k$, and the most common choice of likelihood is a Gaussian one. We refer the reader to \cite{dias2015clustering} for more details.
	
	As another point of comparison to our approach, we fit a Gaussian HMM in both our real and synthetic data experiments. A main point of difference here is that the HMM does not cluster sets of returns: instead, it associates returns at time $t$ to a given latent state. We thus can derive accuracy statistics in the case where we run the HMM over synthetic data, but for real data our validation method using the MMD is not possible.

	\subsection{Validation on real data}\label{subsec:realdata}
	
	In this section, we give results from each algorithm on real data.
	
	\subsubsection{Data and hyperparameters}
	We begin by testing both algorithms on market data. In particular, we use one-hourly log-returns $r^S \in \mathcal{S}(\real)$ associated to the SPY index from 2005-01-03 to 2020-12-31. Recalling Definition \ref{def:streamlift} and (\ref{eqn:streamlift}), we set the hyper-parameters $(h_1, h_2) = (35, 28)$. This roughly partitions the time-series into weeks, with adjacent partitions within one day of each other. We defer discussions regarding choices of hyperparameters to Section \ref{subsec:hyperparameters}. 
	
	Regarding the number of clusters, we set $k=2$. Primarily, this is for simplicity as it makes comparisons between each methods simpler. Moreover, it also reflects a stylised fact regarding financial markets, being that market returns can be roughly apportioned into bull and bear cycles. We note that other values of $k$ do have financial interpretations - for instance, see Maheu, McCurdy and Song \cite{bullbearregimes}. 
	
	We ran the two algorithms over the lifted steam of data $\ell(r^S)$. For each algorithm, we obtained centroids $\{\overline{\mu}_i\}_{i=1,2}$ and the nearest neighbours $\{\mathcal{C}_l\}_{l=1,2}$ with $\mathcal{C}_l \subset \mathcal{P}_p(\real)$ being the nearest neighbours corresponding to the $l^\text{th}$ centroid. In order to display our results, we primarily use two plots. The first is the projection of each distribution $\mu \in \mathcal{K}$ onto $\real^2$ via the map 
	\begin{gather*}
		f_p: \mathcal{P}_p(\real) \to \real^2, \\ 
		\mu \mapsto \left(\sqrt{\mathrm{Var}(\mu)}, \ex[\mu]\right),
	\end{gather*}
	that is, a scatter plot of each measure in mean-variance space. We colour these points according to their centroid membership. Points coloured green correspond to the cluster with lower variance than those coloured red. The second plot we make use of is a time-series plot of stock values $S$, where each partition has been coloured corresponding to its centroid membership. Given that the empirical distributions overlap, a given timestamp may be classified into multiple clusters. Thus, we colour these points according to their average centroid memberships. In the case of $k=2$ with the hyperparameters $(h_1, h_2) = (35, 28)$, a single return $r^S_i$ can potentially belong to 5 different empirical measures. Thus, there are 6 total potential centroid membership combinations that it can have, and we colour these sections of the price path accordingly.
	
	\subsubsection{Algorithm results}\label{subsubsec:momentwk}
	
	We first present the results of running either algorithm over hourly SPY data. Figure \ref{fig:spymeanvar} gives the scatter plots of each empirical measure $\mu \in \mathcal{K}$ coloured according to its cluster membership. The centroids are marked by crosses and coloured accordingly. 
	
	\begin{figure}[h]
		\centering
		\begin{subfigure}{0.5\linewidth}
			\centering
			\includegraphics[width=\textwidth]{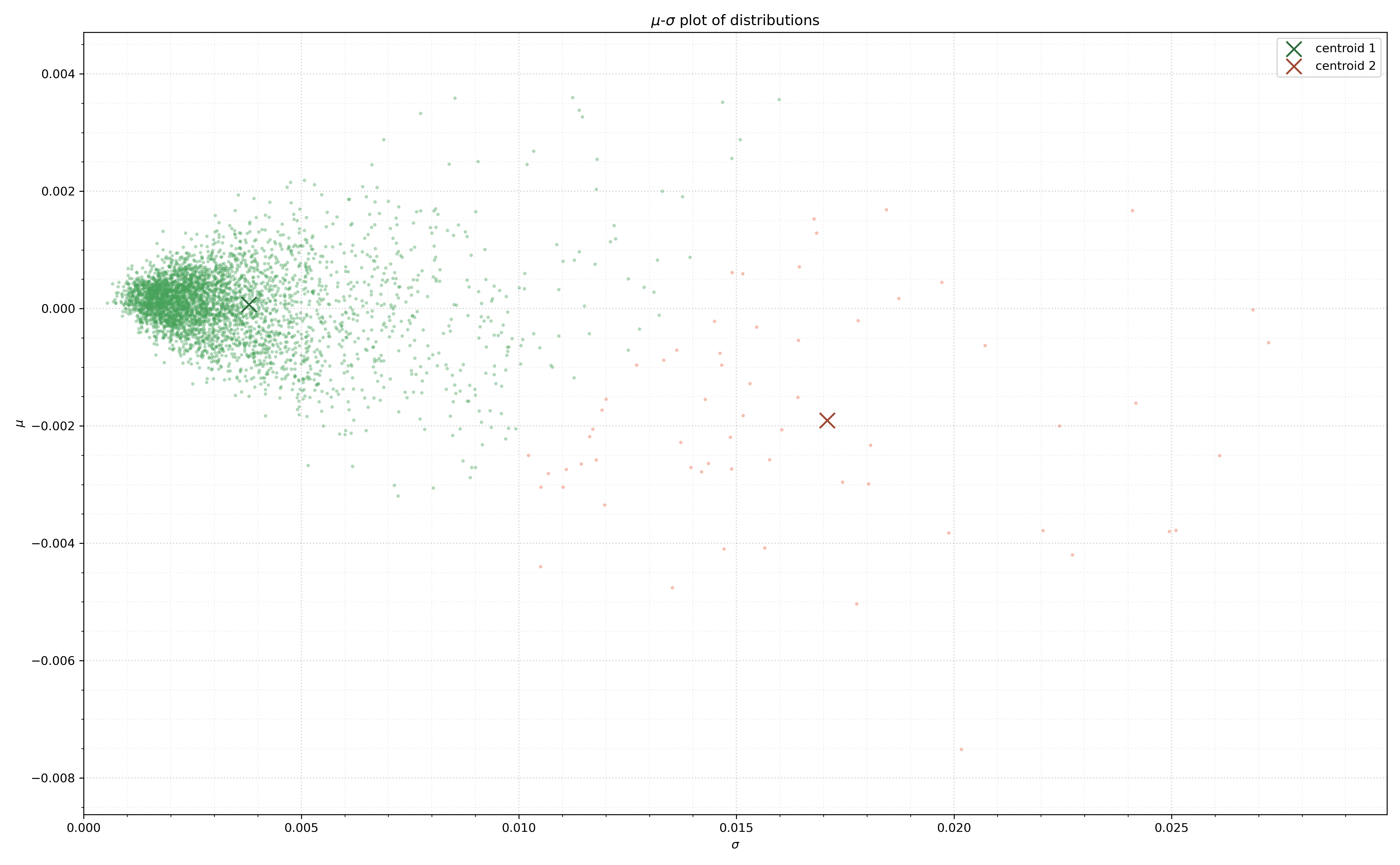}
			\caption{MK-means.}
			\label{fig:momentsmeanvar}
		\end{subfigure}%
		\begin{subfigure}{0.5\linewidth}
			\centering
			\includegraphics[width=\textwidth]{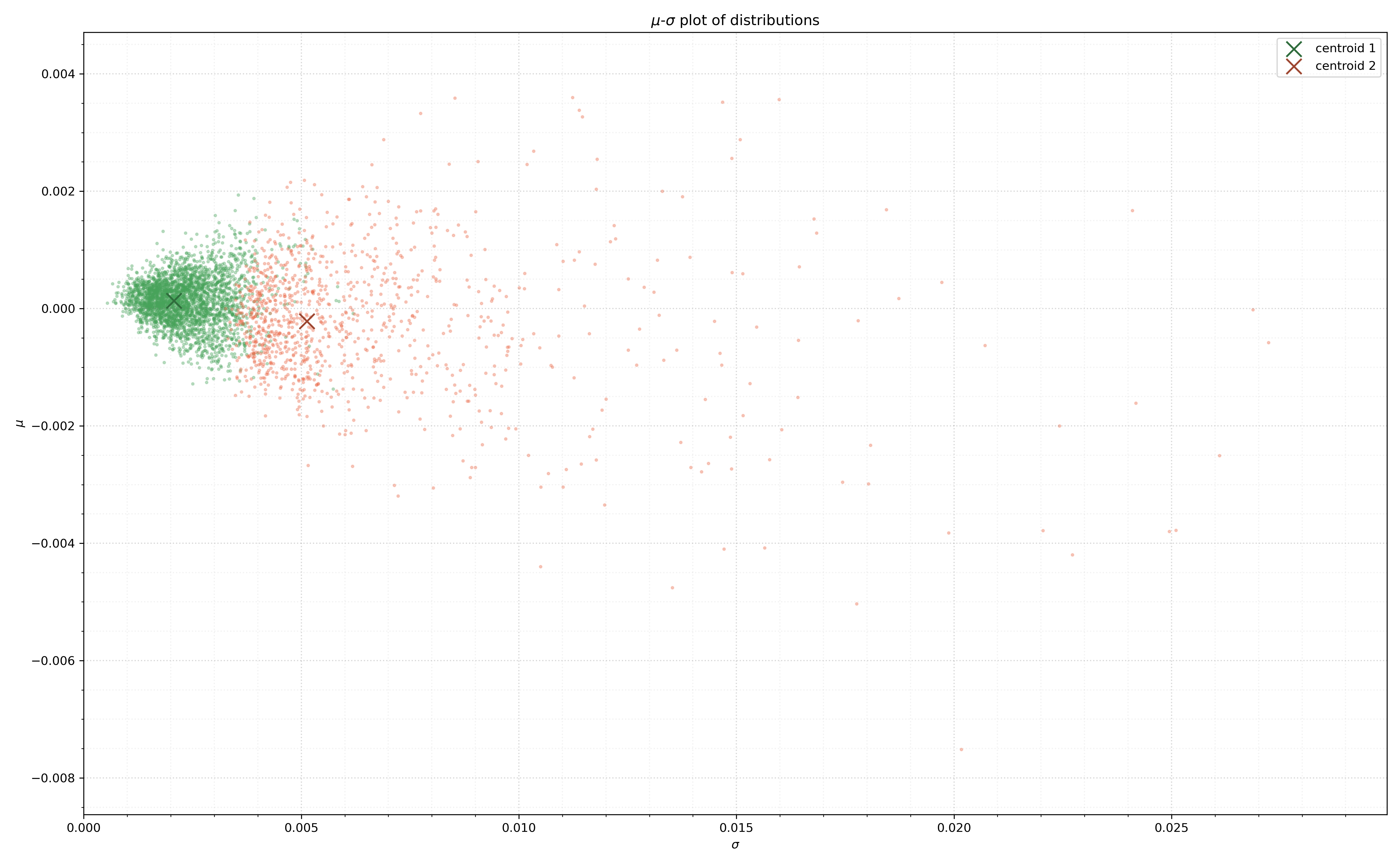}
			\caption{WK-means.}
			\label{fig:wassersteinmeanvar}
		\end{subfigure}
		\caption{Plots of MK-means and WK-means clusters in mean-variance space.}
		\label{fig:spymeanvar}
	\end{figure}
	
	From Figure \ref{fig:wassersteinmeanvar}, we see that the WK-means classifications are much less susceptible to outlier distributions than the MK-means algorithm. We also note that the Wasserstein approach demarcates distributions $\mu \in \mathcal{K}$ by variance, which one naturally expects in a financial market setting. Although Figure \ref{fig:momentsmeanvar} appears to do the same, it is hard to state this definitively as the clustering algorithm seems to primarily group outlier distributions.
	
	This is made apparent in the graphs presented in Figure \ref{fig:spyhistorical}, where we have associated to each $\mu \in \mathcal{K}$ the partition of $S \in \mathcal{S}(\real)$ it is generated from. 
	
	\begin{figure}[h]
		\centering
		\begin{subfigure}{0.5\linewidth}
			\centering
			\includegraphics[width=\textwidth]{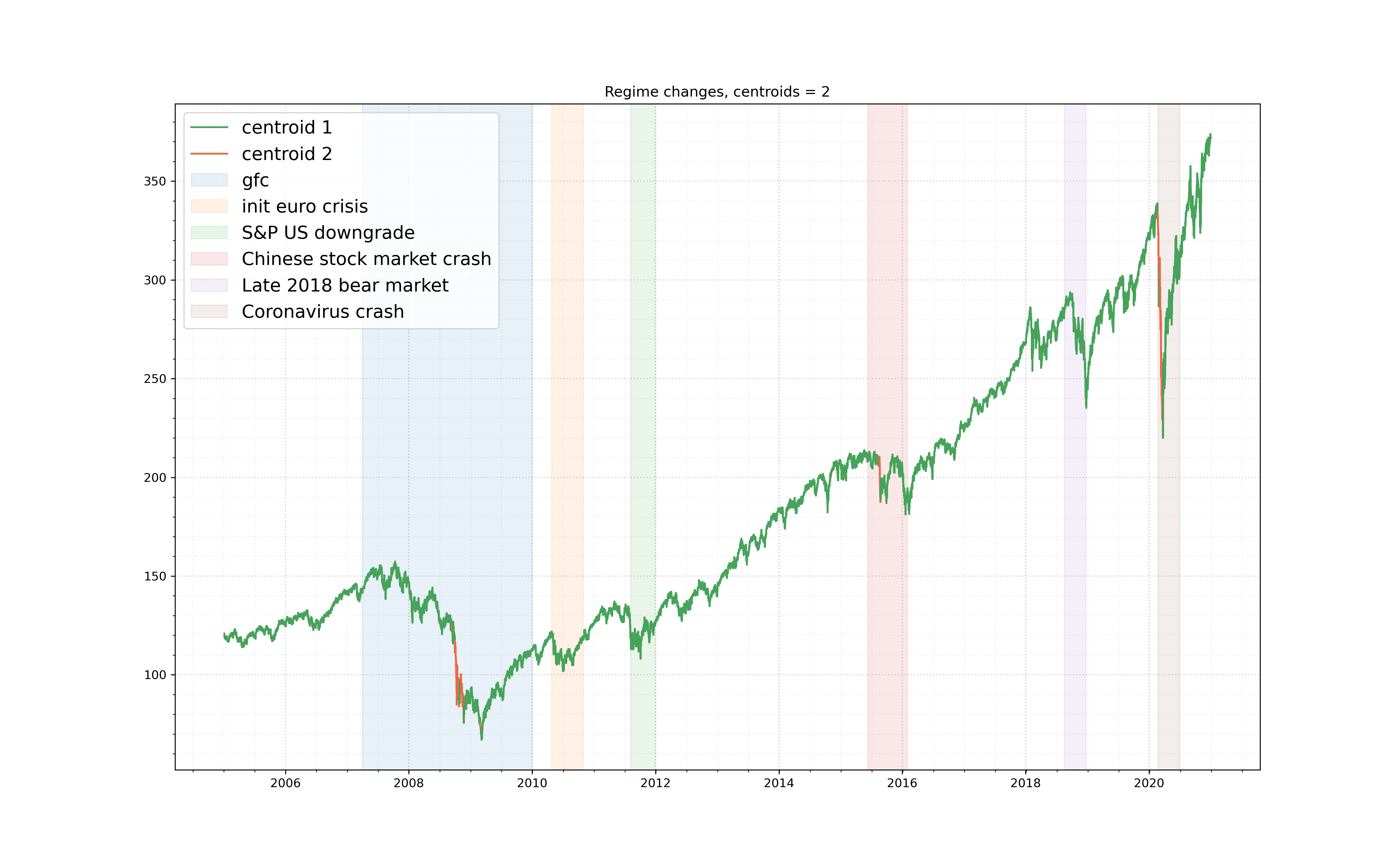}
			\caption{MK-means.}
			\label{fig:momentshistorical}
		\end{subfigure}%
		\begin{subfigure}{0.5\linewidth}
			\centering
			\includegraphics[width=\textwidth]{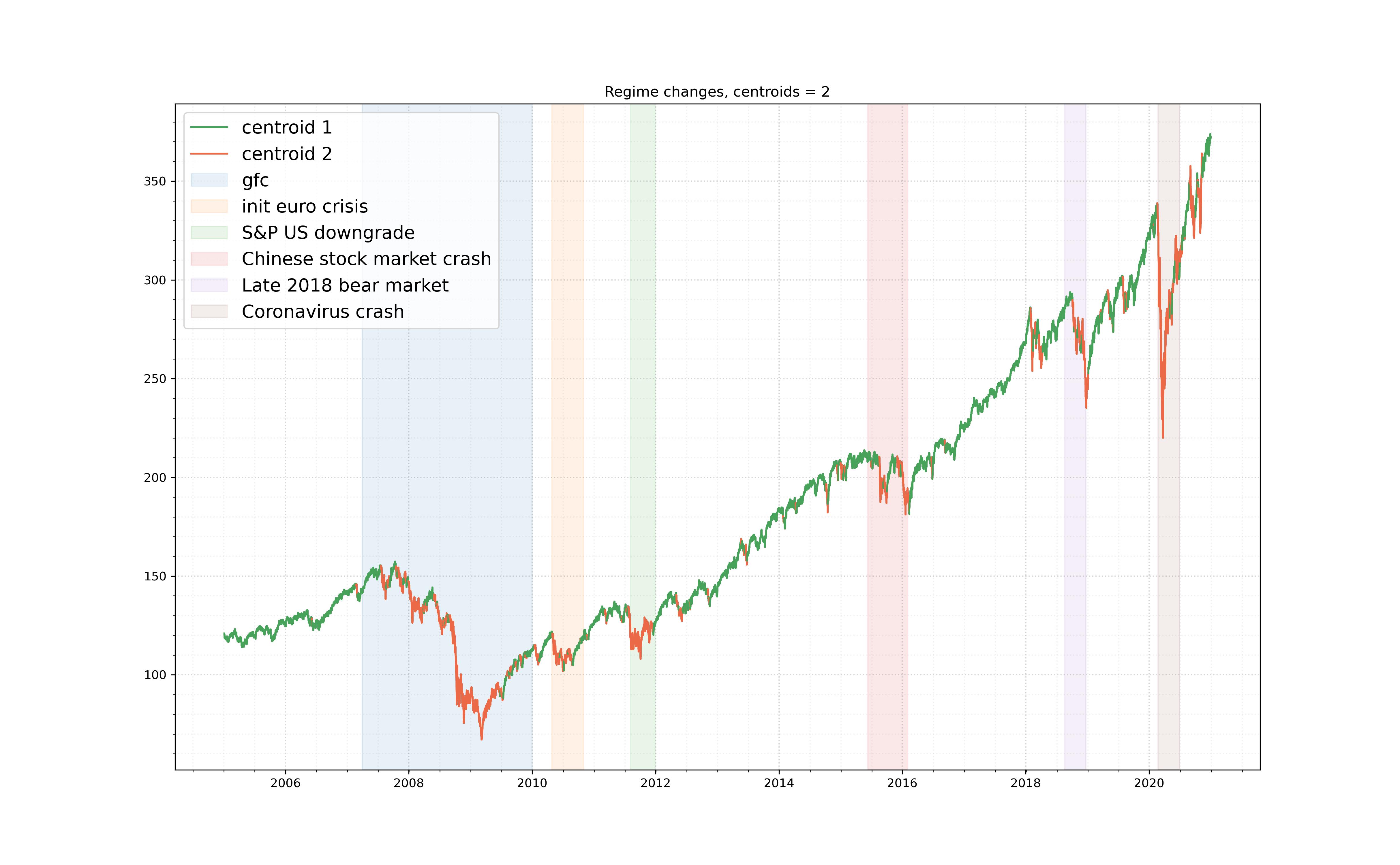}
			\caption{WK-means.}
			\label{fig:wassersteinhistorical}
		\end{subfigure}
		\begin{subfigure}{0.5\linewidth}
			\centering
			\includegraphics[width=\textwidth]{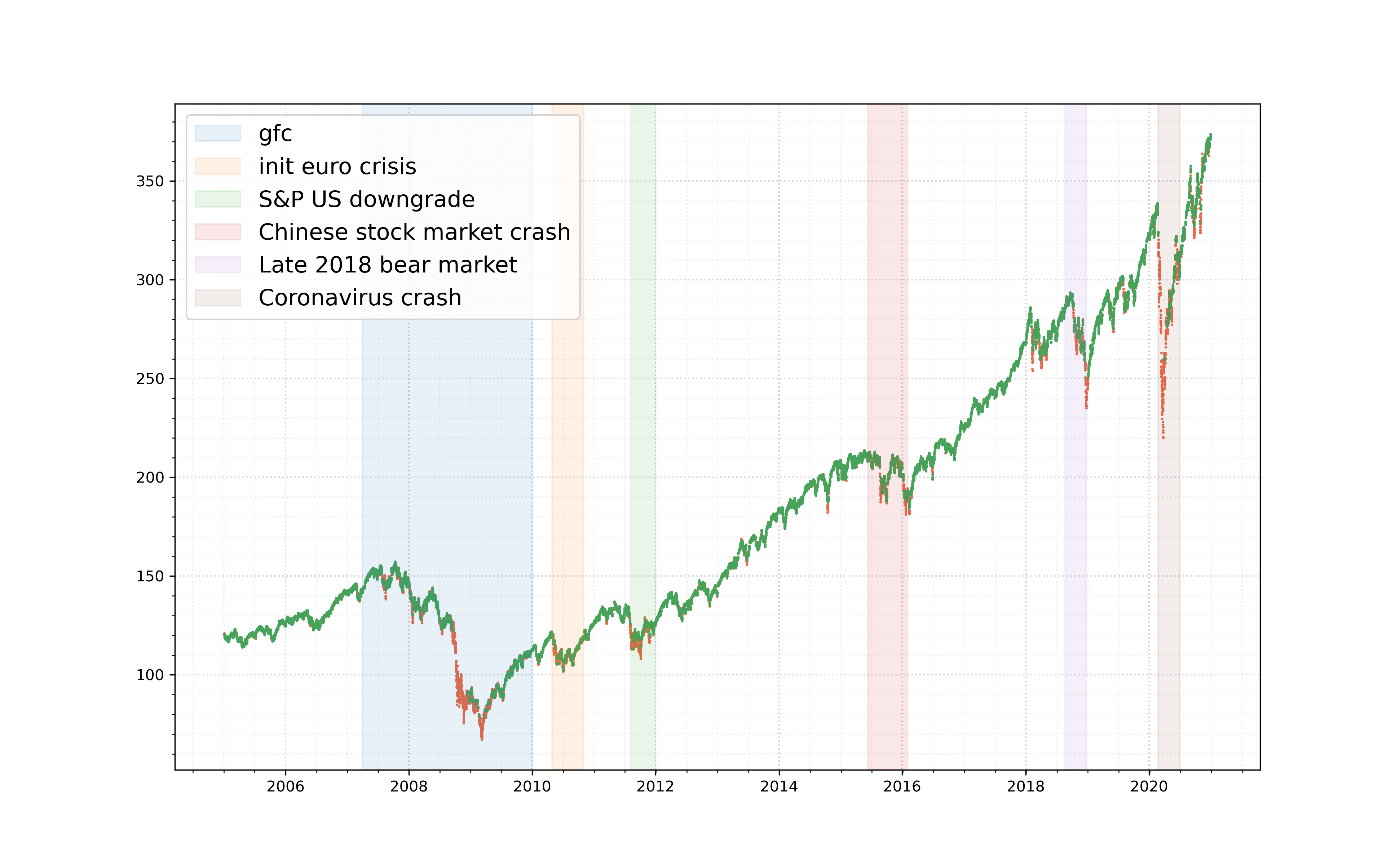}
			\caption{Hidden Markov model.}
			\label{fig:hmmhistorical}
		\end{subfigure}
		\caption{Historical cluster colouring on SPY price path.}
		\label{fig:spyhistorical}
	\end{figure}
	
	Figures \ref{fig:momentshistorical} and \ref{fig:wassersteinhistorical} show that both algorithms are able to separately classify periods of returns associated to the global financial crisis and the more recent market instability due to the coronavirus pandemic. However, only the WK-means algorithm is able to distinguish between more subtle periods of stock market volatility: the beginning of the Eurozone/Greek debt crisis in 2010, the S\&P US credit-rating downgrade in 2011, and the 2015/16 Chinese stock market crash were tagged as periods of regime change only by the WK-means algorithm. Note that we can include an example of a Hidden Markov model run with this plot, in Figure \ref{fig:hmmhistorical}. We note that the results from this approach seem to sit somewhere between the Wasserstein and moment algorithms.

	\subsubsection{Validation methods}\label{subsubsec:validation}
	In order to compare clusterings obtained from both algorithms, we use the marginal MMD test introduced in Section \ref{subsec:MMD}. The two methods of evaluation we consider are applied both between and within clusters $\mathcal{C} = \{\mathcal{C}_l\}_{l=1,2}$. Moreover, both methods of evaluation are similar, in that they involve bootstrapping the distribution of $\mathrm{MMD}^2_b$ between two sets of samples.
	
	More generally, the definition of an optimal clustering over a set of data $X$ is not well-defined, and in the case of financial data, this is certainly true. Heuristically, we would like individual clusters to contain objects that are similar to each other whilst being distinct from objects in other clusters. We note that there do already exist several indexes used to evaluate the result of a given $k$-means clustering, which we recall here.
	\begin{definition}[Davies-Bouldin index, \cite{4766909}]\label{def:daviesbouldinindex}
		Suppose $(V, \norm{\cdot}_V)$ is a normed vector space. Let $\left\{\left(\overline{\mathsf{x}}_l, \mathcal{C}_l\right)\right\}_{l=1}^k$ be $k$ centroids and clusters over $X \in \mathcal{S}(\mathcal{X})$, obtained by applying the $k$-means algorithm characterised by the functions
		\begin{align*}
			\varphi: \mathcal{X}&\to V, \\
			a: 2^V &\to V, \text{ and} \\
			l: V^K \times V^k &\to [0, +\infty).
		\end{align*}
		Suppose $d: V \times V \to [0, +\infty)$ is the metric induced by the norm on $V$. Let
		\begin{equation*}
			d_l = \frac{1}{|\mathcal{C}_l|}\sum_{\mathsf{x} \in \mathcal{C}_l} d(\mathsf{x}, \overline{\mathsf{x}}_l)
		\end{equation*}
		be the average distance of cluster elements $\mathsf{x} \in \mathcal{C}_l$ to the central element $\overline{\mathsf{x}}_l$ for $l=1,\dots,k$. Then, the \emph{Davies-Bouldin index} is given by
		\begin{equation}\label{eqn:daviesbouldinindex}
			DB\left(\left\{\left(\overline{\mathsf{x}}_l, \mathcal{C}_l\right)\right\}_{l=1}^k \right) = \frac{1}{k}\sum_{i=1}^k \max_{j\ne i} \frac{d_i + d_j}{d(\overline{\mathsf{x}}_i, \overline{\mathsf{x}}_j)}.
		\end{equation}
		Lower values of (\ref{eqn:daviesbouldinindex}) are indicative of a better clustering. 
	\end{definition}
	\begin{definition}[Dunn index, \cite{doi:10.1080/01969727408546059}]\label{def:dunnindex}
		With the same notation as Definition \ref{def:daviesbouldinindex}, define
		\begin{equation*}
			\underline{d}_{ij} = \min_{\mathsf{x} \in \mathcal{C}_i, \mathsf{y} \in \mathcal{C}_j} d(\mathsf{x}, \mathsf{y})
		\end{equation*}
		to be the smallest distance between elements of each cluster. Also define
		\begin{equation*}
			\overline{d}_{l} = \max_{\mathsf{x}, \mathsf{y} \in \mathcal{C}_l} d(\mathsf{x}, \mathsf{y})
		\end{equation*}
		to be the largest intra-cluster distance between all clusters $\{\mathcal{C}_l\}_{1\le l \le k}$. Then, the \emph{Dunn index} is given by
		\begin{equation}\label{eqn:dunnindex}
			D\left(\left\{\left(\overline{\mathsf{x}}_l, \mathcal{C}_l\right)\right\}_{l=1}^k \right) = \frac{\min_{1\le i, j\le k}\underline{d}_{ij}}{\max_{1\le l \le k}\overline{d}_l}.
		\end{equation}
		Larger values of (\ref{def:dunnindex}) are indicative of a better clustering.
	\end{definition}
	\begin{definition}[Silhouette coefficient, \cite{ROUSSEEUW198753}]\label{def:silhouette}
		With the same notation as Definition \ref{def:daviesbouldinindex}, define 
		\begin{equation*}
			b_i = \min_{i \ne j} \frac{1}{|\mathcal{C}_j|}\sum_{\mathsf{y} \in \mathcal{C}_j}d(\mathsf{x}_i, \mathsf{y}),
		\end{equation*}
		and
		\begin{equation*}
			a_i = \frac{1}{|\mathcal{C}_i|}\sum_{\mathsf{y} \in \mathcal{C}_i} d(\mathsf{x}_i, \mathsf{y}).
		\end{equation*}
		for any $\mathsf{x}_i \in \mathcal{C}_l$, $l=1,\dots,k$. Then, the \emph{Silhouette coefficient} of the point $\mathsf{x}_i$ is given by
		\begin{equation}\label{eqn:silhouette}
			S(i) = \frac{b_i - a_i}{\max(a_i, b_i)}.
		\end{equation}		
		From (\ref{eqn:silhouette}), we can see that $-1 \le S(i) \le 1$. Higher values of $S(i)$ mean that the point $\mathsf{x}_i$ was appropriately allocated to cluster $\mathcal{C}_l$. 
	\end{definition}
	\begin{remark}\label{rmk:avergaesilhouettecoefficient}
		It is often computationally expensive to calculate the (\ref{eqn:silhouette}) for every point $\mathsf{x}$. Thus, we often use an estimate from fewer samples. 
		
		For $0 < \alpha \le 1$, define $\lambda_l = \lfloor \alpha \, |\mathcal{C}_l| \rfloor$ and let $(n_k^l)_{k=1}^{\lambda_l}$ be an increasing sub-sequence of $\{1, \dots, |\mathcal{C}_l|\}$, for $l= 1,\dots,k$. Then, the $\alpha$-\emph{average Silhouette coefficient} $\overline{S}_\alpha$ is given by
		\begin{equation}\label{eqn:averagesilhouette}
			\overline{S}_\alpha = \frac{1}{k}\sum_{l=1}^k \left(\frac{1}{\lambda_l}\sum_{k=1}^{\lambda_l} S(n_k^l) \right).
		\end{equation}
	\end{remark}
	
	The indexes from Definitions \ref{def:daviesbouldinindex}, \ref{def:dunnindex} and \ref{def:silhouette} are often used to evaluate clusters derived from a standard $k$-means algorithm for different values of $k$. We will see that using them to compare clusterings between the MK- and WK-means approaches (for the same value of $k$) does not capture how appropriate clusterings are in reference to the MRCP. Firstly, such indexes are not agnostic to the choice of $(V, \norm{\cdot}_V)$ and are thus not comparable between algorithms. Secondly, a more appropriate validation method for the MRCP would be between regimes $\mu \in \mathcal{S}\left(\mathcal{P}_p(\real)\right)$, as opposed to elements of the clustering space $\varphi(\mu) \in V$. Thus, as an integrable probability metric, the MMD is a more suitable choice to be used to evaluate the appropriateness of either clustering algorithm. Nevertheless we will report the values of these metrics for completeness.
	
	For our between-cluster evaluation, we proceed as follows. Given sets $\mathcal{C}_1, \mathcal{C}_2$ obtained via either the moment- or WK-means clustering algorithm, draw $n \in \mathbb{N}$ pairwise samples $(\mu_i, \nu_i) \in \mathcal{C}_1\times \mathcal{C}_2$ for $i=1,\dots,n$. We represent each empirical measure $\mu_i, \nu_i \in \mathcal{P}_p(\real)$ by its corresponding vector of log-returns $\mathsf{x}_i, \mathsf{y}_i \in \mathcal{S}(\real)$. We then evaluate the test statistic (\ref{eqn:biasedmmd}) where we choose $k: \real^d \times \real^d \to [0, +\infty)$ to be the Gaussian kernel (\ref{eqn:gaussiankernel}) with $\sigma = 0.1$. We then compare the associated distribution of the MMD between the two histograms generated from the moment- and WK-means methods by reporting the similarity score from Definition \ref{def:selfsimilarity}.
	
	Within-cluster evaluation is performed much in the same way as the between-cluster case: for either algorithm, and for each cluster $\mathcal{C}_l$, $l=1,2$ we draw $n\in \mathbb{N}$ pairwise samples $(\mu^1_i, \mu^2_i) \in \mathcal{C}_l \times \mathcal{C}_l$ and evaluate the biased MMD (\ref{eqn:biasedmmd}). We report the similarity score associated to the empirical distribution of each within-cluster MMD and plot the resulting histograms.

	\subsubsection{Cluster validation via the marginal MMD}\label{subsubsec:mmdvalidation}
	
	Recall from Section \ref{sec:intro} that we stated that a clustering algorithm was successful if the \emph{self-similarity} and \emph{distinctness} of derived clusters were appropriately traded off against each other. In this section, we give the scores of each clustering algorithm on the SPY data with references to the indexes introduced in Definitions \ref{def:daviesbouldinindex}, \ref{def:dunnindex}, and \ref{def:silhouette}. In particular, we report the average silhouette coefficient $\overline{S}_\alpha$ with $\alpha = 0.2$. Results are given in Table \ref{tab:indexes}.
	
	\begin{table}	
		\begin{center}
			\begin{tabular}{cccc}
				\toprule
				\textbf{Algorithm} & Davies-Bouldin & Dunn & $\overline{S}_\alpha$ \\ \midrule
				Wasserstein & $1.1075$ & $7.3 \times 10^{-3}$  & $0.5093$ \\ \addlinespace
				Moment & $0.8604$ & $9.2\times 10^{-3}$ & $0.8008$ \\ \midrule
			\end{tabular}
		\end{center}
		\caption{Scores for MK- and WK-means algorithms using traditional $k$-means index evaluation methods, typical run.}
		\label{tab:indexes}
	\end{table}
	
	As noted in Section \ref{subsubsec:validation}, scores associated to the first two indexes are not invariant under the choice of $(V, \norm{\cdot}_V)$ and thus do not represent a like-for-like comparison. Yet we note that the average Silhouette coefficient $\overline{S}_\alpha$ remains higher for the MK-means method than the WK-means, implying that (under the more traditional method of cluster validation) regimes clustered via the former belong to more appropriate clusters than the latter. 
	
	As outlined in Section \ref{subsubsec:validation}, we applied our within- and between-cluster validation via the marginal MMD from Definition \ref{def:maximummeandiscrepancy} by sampling $n=1000000$ times from each cluster $\mathcal{C}_1, \mathcal{C}_2$ obtained from either method, and calculating the biased MMD (\ref{eqn:biased-mmd-general}). We order samples from clusters in ascending order to ensure like-for-like comparison between sample elements. Figure \ref{fig:between-cluster-mmd} shows two empirical distributions of the biased MMD between elements in the two clusters formed from the WK- and MK-means method.
	
	\begin{figure}[h!]
		\centering
		\includegraphics[width=0.6\textwidth]{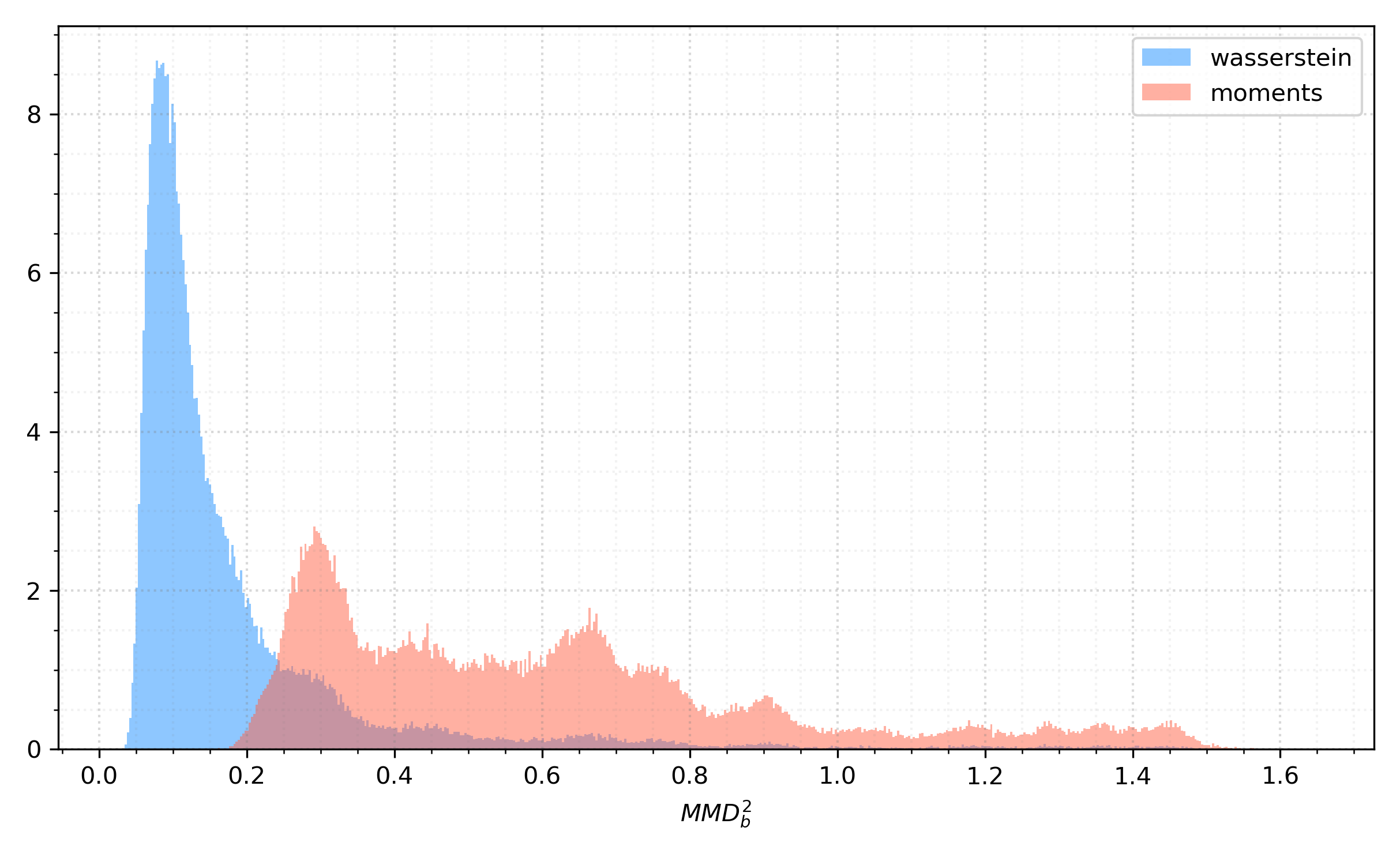}
		\caption{Histograms of between-cluster MMD approximation, Wasserstein vs moments method.}
		\label{fig:between-cluster-mmd}
	\end{figure}
	
	Figure \ref{fig:between-cluster-mmd}, shows that the clusters obtained via the WK-means method are significantly more similar to each other than those obtained from the MK-means method. However, one cannot then conclude that the latter method provides a better clustering result if the disparity within groups is due to one cluster being composed primarily of outlier elements. Figure \ref{fig:within-cluster-mmd} gives the empirical distribution of the within-cluster MMD for each algorithm. These histograms were derived by calculating the biased MMD test statistic via the sub-sampling technique from Definition \ref{def:selfsimilarity}, with $n=100000$.
	
	\begin{figure}[h!]
		\centering
		\includegraphics[width=0.6\textwidth]{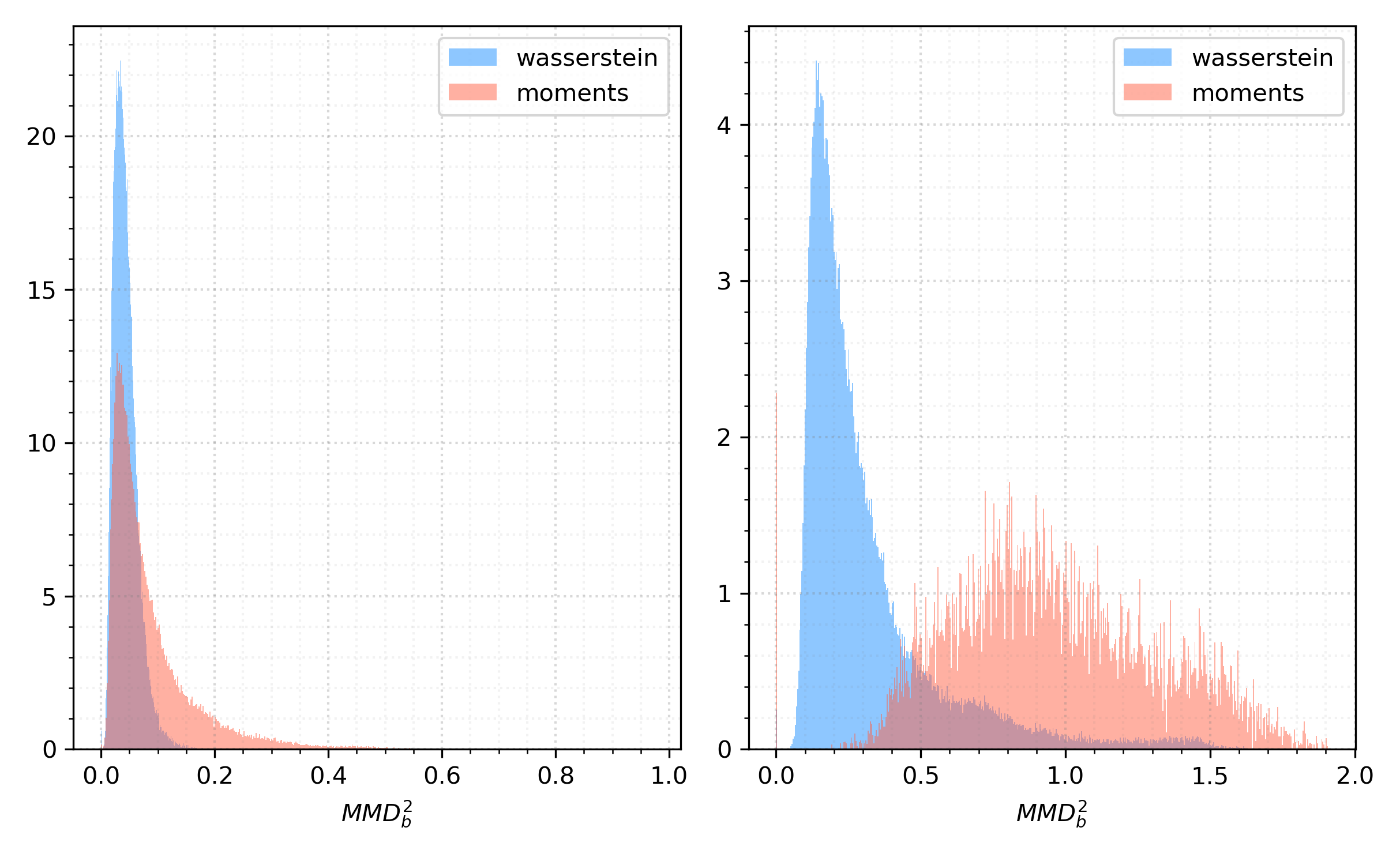}
		\caption{Histograms of within-cluster MMD approximation, Wasserstein vs moments method.}
		\label{fig:within-cluster-mmd}
	\end{figure}
	
	Both Figure \ref{fig:within-cluster-mmd} and the self-similarity scores in Table \ref{tab:selfsim} show that the clusters obtained via the WK-means algorithm are significantly more self-similar than those obtained from the MK-means algorithm. 
	
	\begin{table}	
		\begin{center}
			\begin{tabular}{ccc}
				\toprule
				\textbf{Algorithm} & $\mathcal{C}_1$ & $\mathcal{C}_2$ \\ \midrule
				Wasserstein & $0.0395$ & $0.2304$  \\ \addlinespace
				Moment & $0.0631$ & $1.1961$ \\ \midrule
			\end{tabular}
		\end{center}
		\caption{Self-similarity scores, WK- and MK-means algorithms.}
		\label{tab:selfsim}
	\end{table}
		
	\subsection{Validation on synthetic data}\label{subsec:syndata}
	
	Evaluating a given clustering algorithm on real market data is difficult for multiple reasons. One is that it is not possible to infer the underlying probabilistic structure associated to the stream of log-returns $r^S$ that a clustering algorithm is run over, and thus one cannot say with any certainty what constitutes a ``correct'' clustering. A corollary to this is that it is impossible to know exactly at what point a regime change occurs when studying real market data.
	
	Therefore, we evaluated both clustering algorithms on synthetic market data, where we specify beforehand at what times regime changes occur. Because we knew both the underlying probabilistic structure and the regime change periods \emph{a priori}, we could further evaluate both how accurately either algorithm is classifying sequences of returns into regimes, and how closely the centroids $\{\overline{\mu}_l\}_{l=1,2}$ of each cluster correspond to the true distributions $\{\mesp_{l}\}_{l=1,2}$ associated to the synthetic data. 
	
	The methodology is as follows. For a given time interval $[0, T]$ with $T \in \mathbb{N}$, we define a mesh so that each time increment roughly represents one market hour. That is, with $n:= 252\times7$, we set
	\begin{equation*}
		\Delta = \left\{\left[\dfrac{i-1}{n}, \dfrac{i}{n}\right] : i=1,2,\dots, nT \right\}.
	\end{equation*}
	Next, we define the number of regime changes $r \in \mathbb{N}$ we wish to observe. We specify their starting points and lengths by $(s_i, l_i) \in \mathbb{N}\times \mathbb{N}$ for $i=1,\dots, r$, with
	\begin{equation*}
		0\le s_0 < s_r + l_r \le nT,
	\end{equation*}
	and
	\begin{equation*}
		s_i + l_i + 2 < s_{i+1}, \qquad \text{for }i=1,\dots,r-1.
	\end{equation*}
	Each $l_i$ can be a constant or a random variable. We thus obtain the set of disjoint intervals 
	\begin{equation}\label{eqn:regimechanges}
		R = \{[s_i, s_i + l_i]: i=1,\dots,r \}
	\end{equation}
	and their associated complements $N=\Delta \setminus R$ which partition the interval $\Delta$ into two sets. Intervals in $R$ will correspond to times where we observe a regime change in our synthetic data, which will start at $s_i$ and end at $s_i+l_i$ for $i=1,\dots,r$.
	
	Once we have run a classification algorithm over a synthetic price path, we consider three measures of accuracy: total accuracy, accuracy during the standard regime (regime-off) and accuracy during the regime change (regime-on). This is calculated in the following way: for $i=1,\dots,N-1$, associate to each log-return $r^S_i$ the empirical measures $M_i = \{\mu_{j(i)}, \dots, \mu_{j(i)+v}\}$ it was a member of. With our chosen hyperparameters and $k=2$, one has that $v \in [1, 5]$ and $j \in \mathbb{N}$ is the first measure that $r^S_i$ is a member of. Note that if the overlap hyperparameter $h_2 = 0$ then $v=1$. We then calculate which cluster each $\mu \in M_i$ is associated to, which gives us our predicted labels $\overline{y}^i = \{\overline{k}_1, \dots, \overline{k}_v\}$. We then aggregate these labels into the row vector
	\begin{equation*}
		\overline{Y}^i = \left(\sum_{j=1}^v \chi_{\{x=l\}}(\overline{k}_j) \right)_{l=1}^k \qquad \text{for }i=1,\dots,N-1,
	\end{equation*}
	where $k = 2$ is the number of clusters. In what follows we assume the assignment $\overline{k}=1$ corresponds to the standard regime and $\overline{k}=2$ the regime change. We then have the following definitions. 
	\begin{definition}\label{def:accuracyscores}
		With the notation above, for a given vector of log-returns $r^S \in \mathcal{S}(\mathbb{R})$ and cluster assignments $\mathcal{C} = \{\mathcal{C}_l\}_{l=1}^k$, the \emph{regime-off accuracy score (ROFS)}  is given by 
		\begin{equation}
			\mathrm{ROFS}(r^S, \mathcal{C}) = \frac{\sum_{r^S_i \in N}\overline{Y}^i_1}{\sum_{r^S_i \in N}\sum_{k=1,2}\overline{Y}^i_k}.
		\end{equation}
		Similarly, the \emph{regime-on accuracy score (RONS)} is given by
		\begin{equation}
			\mathrm{RONS}(r^S, \mathcal{C}) = \frac{\sum_{r^S_i \in R}\overline{Y}^i_2}{\sum_{r^S_i \in R}\sum_{k=1,2}\overline{Y}^i_k}.
		\end{equation}
		Finally, \emph{total accuracy (TA)} is given by
		\begin{equation}
			\mathrm{TA}(r^S, \mathcal{C}) = \frac{\sum_{r^S_i \in N}\overline{Y}^i_1 + \sum_{r^S_i \in R}\overline{Y}^i_2}{\sum_{i=1}^{N-1}\sum_{k=1,2}\overline{Y}^i_k}.
		\end{equation}
	\end{definition}

	\subsubsection{Geometric Brownian motion}\label{subsubsec:gbm}
	
	In this section, we discuss how we tested both clustering algorithms on synthetic stock data which was modelled as a geometric Brownian motion. Let $\mathcal{M}(\Theta)$ be a family of models indexed by a parameter set $\Theta \subset \real^d$. Initially, we chose $\mathcal{M}(\Theta) = \mathrm{gBm}(\mu, \sigma)$. We then specified two parameter combinations $\theta_{\mathrm{bull}} = (\mu_1, \sigma_1)$ and $\theta_{\mathrm{bear}} = (\mu_2, \sigma_2)$, corresponding to two market regimes. 
	
	We then construct a geometric Brownian motion with associated parameters $\theta_{\mathrm{bear}}$ over intervals $[s_i, s_i + l_i] \in R$ for $i=1,\dots,r$, and with parameters $\theta_{\mathrm{bull}}$ elsewhere. We then run both clustering algorithms on the synthetic data and are returned the clusters with associated centroids as output. Since
	\begin{equation}
		\ln S_t \sim \mathrm{Normal}\left((\mu-\sigma^2/2)t, \sigma^2 t\right) \qquad \text{for all }t \ge 0,
	\end{equation} 
	the true measures $\{\mesp_l\}_{l=1,2}$ are given by
	\begin{equation}\label{eqn:measures}
		\mesp_l = \mathrm{Normal}((\mu_l - \sigma_l^2/2)dt, \sigma_l^2dt) \qquad \text{for } l=1,2,
	\end{equation} 
	where $dt = 1/n$ is the mesh size. Due to the Gaussianity of the distribution of the true log-returns, it suffices to check the mean and variance of the centroid measures $\mu_l$ to gauge how close they are to the true measures $\mesp_l$ for $l=1,2$.
	
	We begin by testing on geometric Brownian motion paths with
	\begin{align*}
		\theta_{\mathrm{bull}} &= (0.02, 0.2), \qquad \text{and} \\
		\theta_{\mathrm{bear}} &= (-0.02, 0.3).
	\end{align*}
	We simulate a path over $T=20$ years with $r=10$ regime changes, and randomly chose each $s_i$ for $i=1,\dots, 10$ and fixed $l_i = 0.5\times 252\times 7$. This choice corresponds to regime changes persisting for approximately half a year. Our mesh grid is thus given by
	\begin{equation*}
		\Delta = \left\{\left[\dfrac{i-1}{1764}, \dfrac{i}{1764}\right] : i=1,2,\dots, 252\times 7\times 20 \right\}.
	\end{equation*}
	When a regime change occurs, the gBm parameters shift from $\theta_{\mathrm{bull}}$ to $\theta_{\mathrm{bear}}$. Figure \ref{fig:gbmplots} shows an example of such a gBm path, with the regime change periods highlighted in red. Figure \ref{fig:gbmlogrets} gives the log-returns associated to Figure \ref{fig:gbmpath}, again with the regime changes highlighted in red.
	
	\begin{figure}[h!]
		\centering
		\begin{subfigure}{0.5\linewidth}
			\centering
			\includegraphics[width=\textwidth]{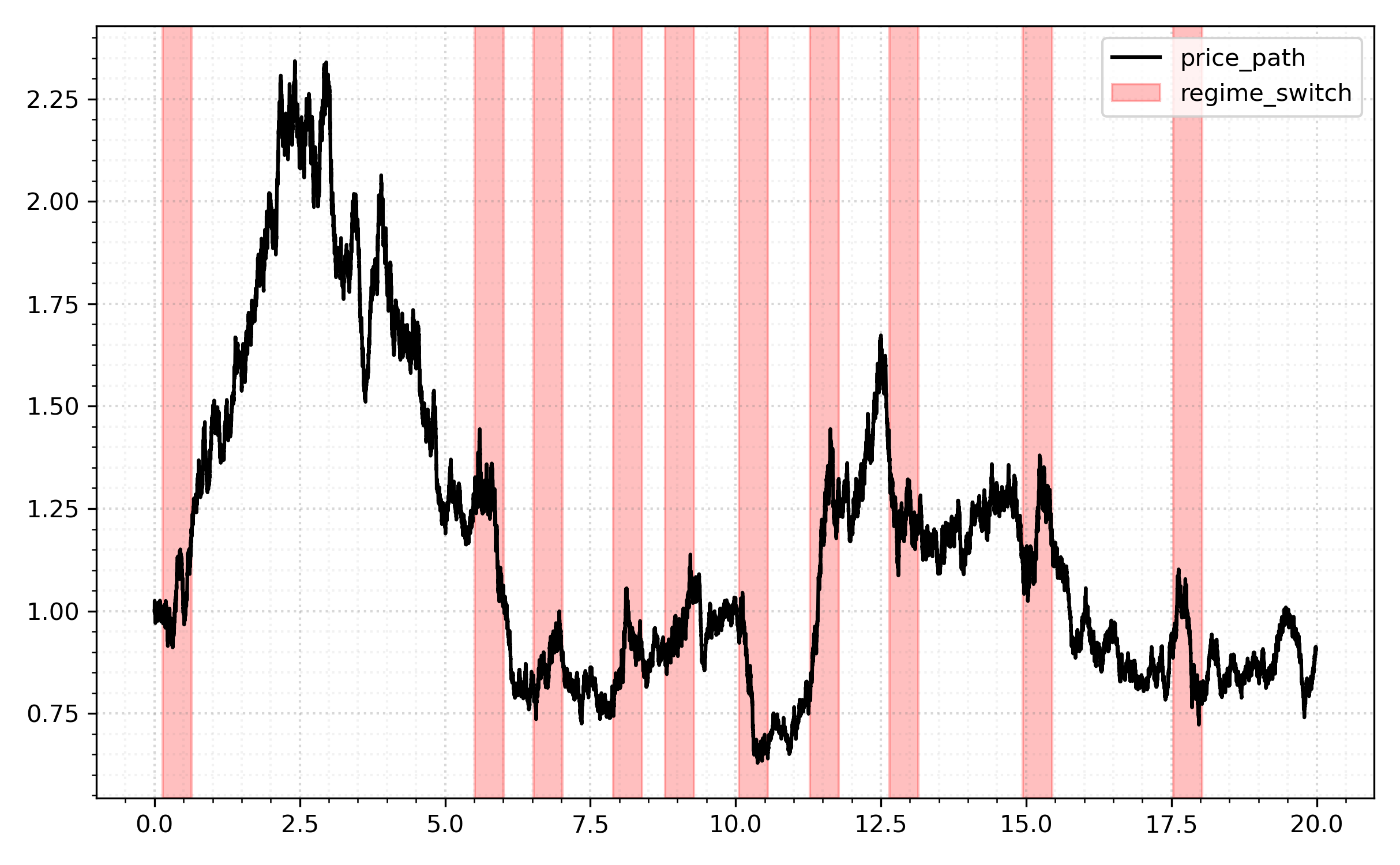}
			\caption{gBm path, regime changes highlighted.}
			\label{fig:gbmpath}
		\end{subfigure}%
		\begin{subfigure}{0.5\linewidth}
			\centering
			\includegraphics[width=\textwidth]{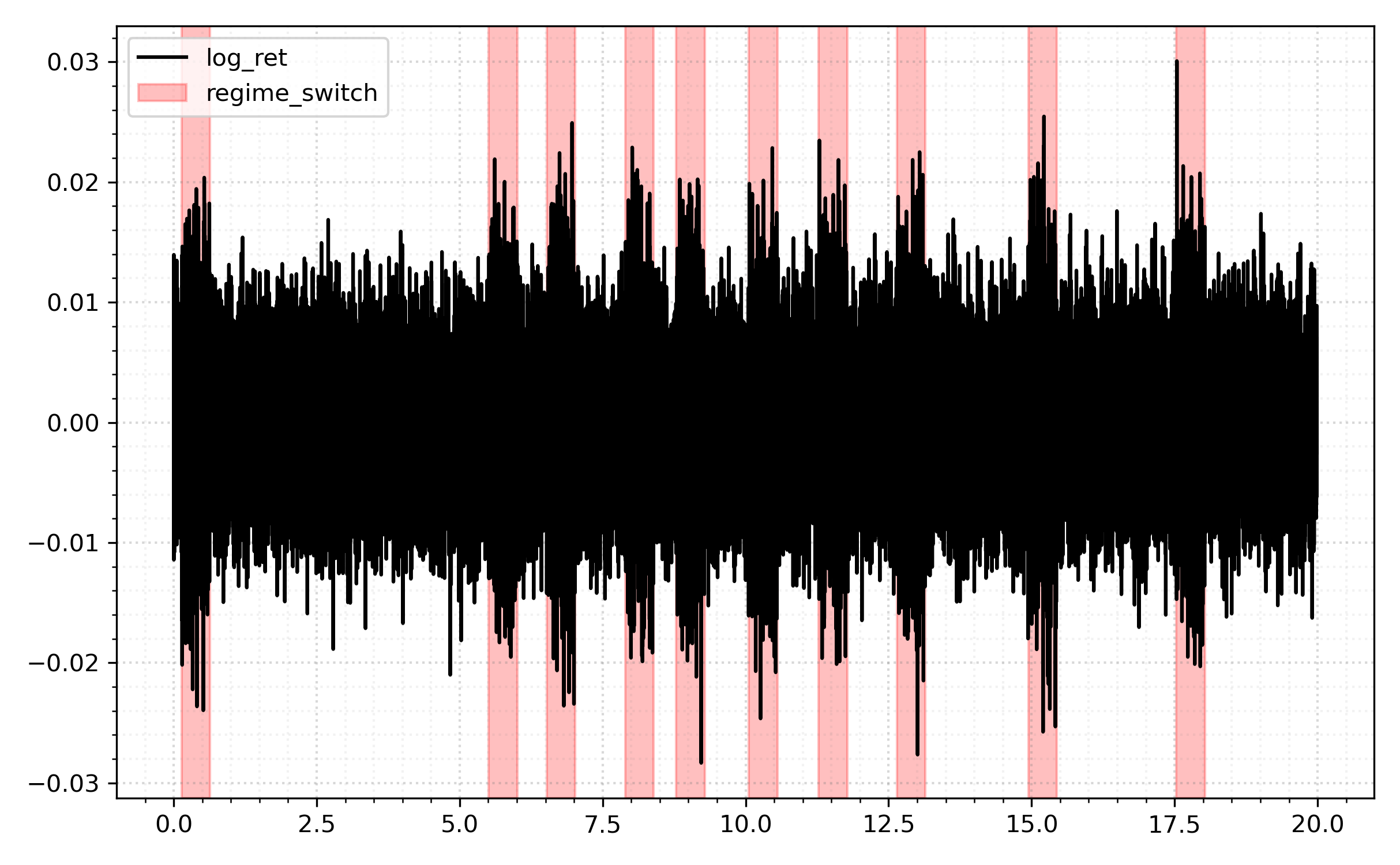}
			\caption{Associated log returns.}
			\label{fig:gbmlogrets}
		\end{subfigure}
		\caption{A synthetic geometric Brownian motion path, and associated log returns.}
		\label{fig:gbmplots}
	\end{figure}
	
	We ran all three clustering algorithms on the same simulated path data. As seen from Figure \ref{fig:meanvargbm} and Figure \ref{fig:historicalgbm}, both algorithms perform well - they are able to accurately capture both the regime changes and ends. We give a summary of the accuracy scores of each algorithm in Table \ref{table:gbmaccuracy} for a total of $n=50$ runs.
	
	\begin{table}[h]	
		\begin{center}
			\begin{tabular}{ccccc}
				\toprule
				\textbf{Algorithm} & Total & Regime-on & Regime-off & Runtime\\ \midrule
				Wasserstein & $90.60\% \pm 5.81\%$ & $\boldsymbol{87.24}\% \pm 4.11\%$  & $91.72\% \pm 6.46\%$ & $0.87s \pm 0.16s$  \\ \addlinespace
				Moment & $\boldsymbol{93.23}\% \pm 0.41\%$ & $74.83\% \pm 1.57\%$ & $\boldsymbol{99.38}\% \pm 0.1\%$ & $1.06s \pm 0.16s$\\
				\addlinespace
				HMM & $58.16\% \pm 7.11\%$ & $41.51\% \pm 7.43\%$ & $63.72\% \pm 11.94\%$ & $0.58s \pm 0.36s$\\\midrule
			\end{tabular}
		\end{center}
		\caption{Accuracy scores with $95\%$ CI, gBm synthetic path, $n=50$ runs.}
		\label{table:gbmaccuracy}
	\end{table}
	
	It is interesting to note that, even in the Gaussian case, the WK-means algorithm does a better job of picking up regime changes than the standard approach. By comparison, the HMM tends to fail to detect the changes in regime at this fixed level of difference in parameter space and thus cannot determine regime change times at this level of granularity between the two models. We provide the plots of the clustering algorithms in mean-variance space and the historical colouring plots in Figures \ref{fig:meanvargbm} and \ref{fig:historicalgbm}.
	
	\begin{figure}[h]
		\centering
		\begin{subfigure}{0.5\linewidth}
			\centering
			\includegraphics[width=\textwidth]{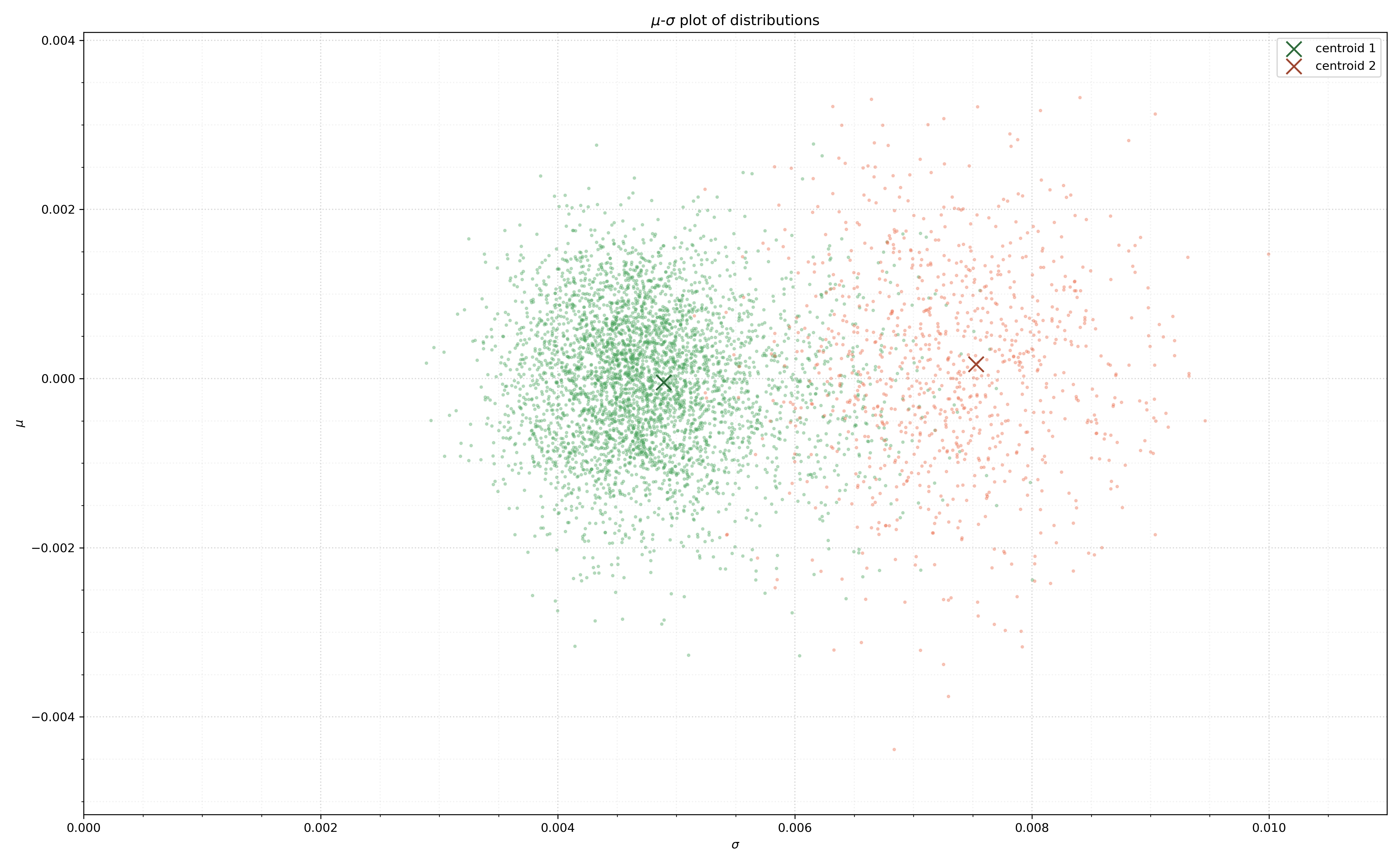}
			\caption{MK-means.}
			\label{fig:momentsgbm}
		\end{subfigure}%
		\begin{subfigure}{0.5\linewidth}
			\centering
			\includegraphics[width=\textwidth]{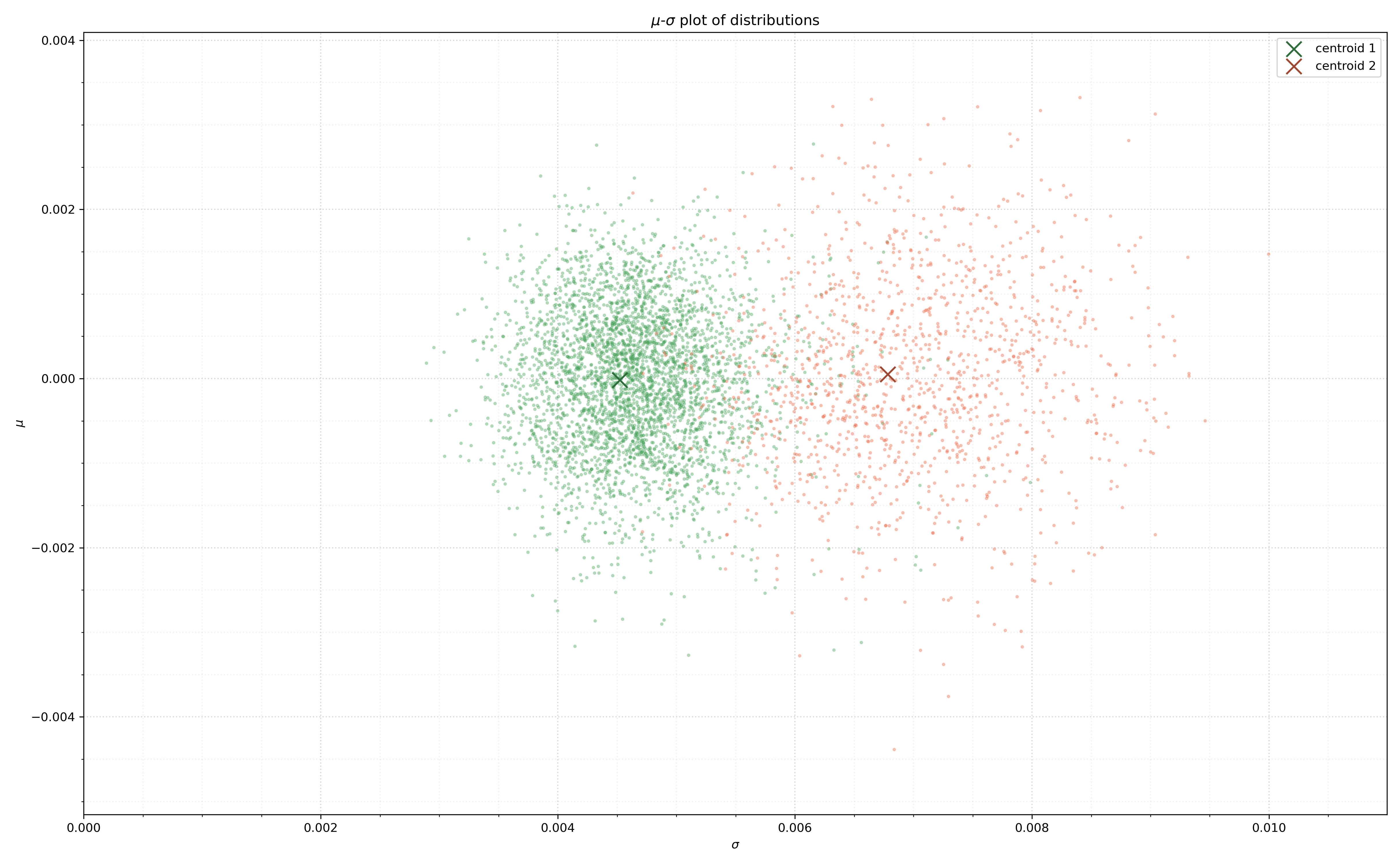}
			\caption{WK-means.}
			\label{fig:wassersteingbm}
		\end{subfigure}
		\caption{Plots of clusters in mean-variance space, synthetic gBm, example run.}
		\label{fig:meanvargbm}
	\end{figure}	
	\begin{figure}[h]
		\centering
		\begin{subfigure}{0.5\linewidth}
			\centering
			\includegraphics[width=\textwidth]{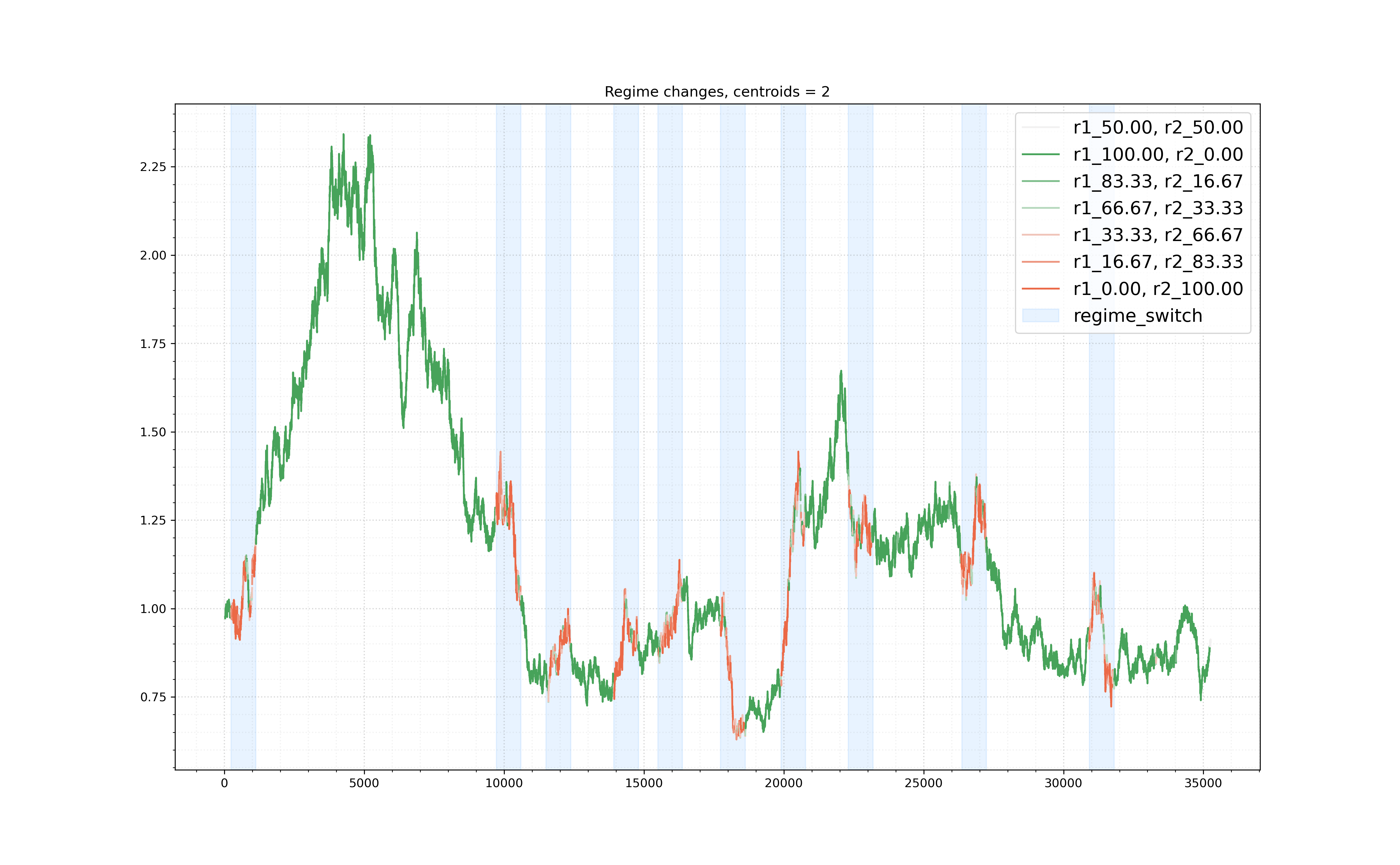}
			\caption{MK-means.}
			\label{fig:momentshistoricalgbm}
		\end{subfigure}%
		\begin{subfigure}{0.5\linewidth}
			\centering
			\includegraphics[width=\textwidth]{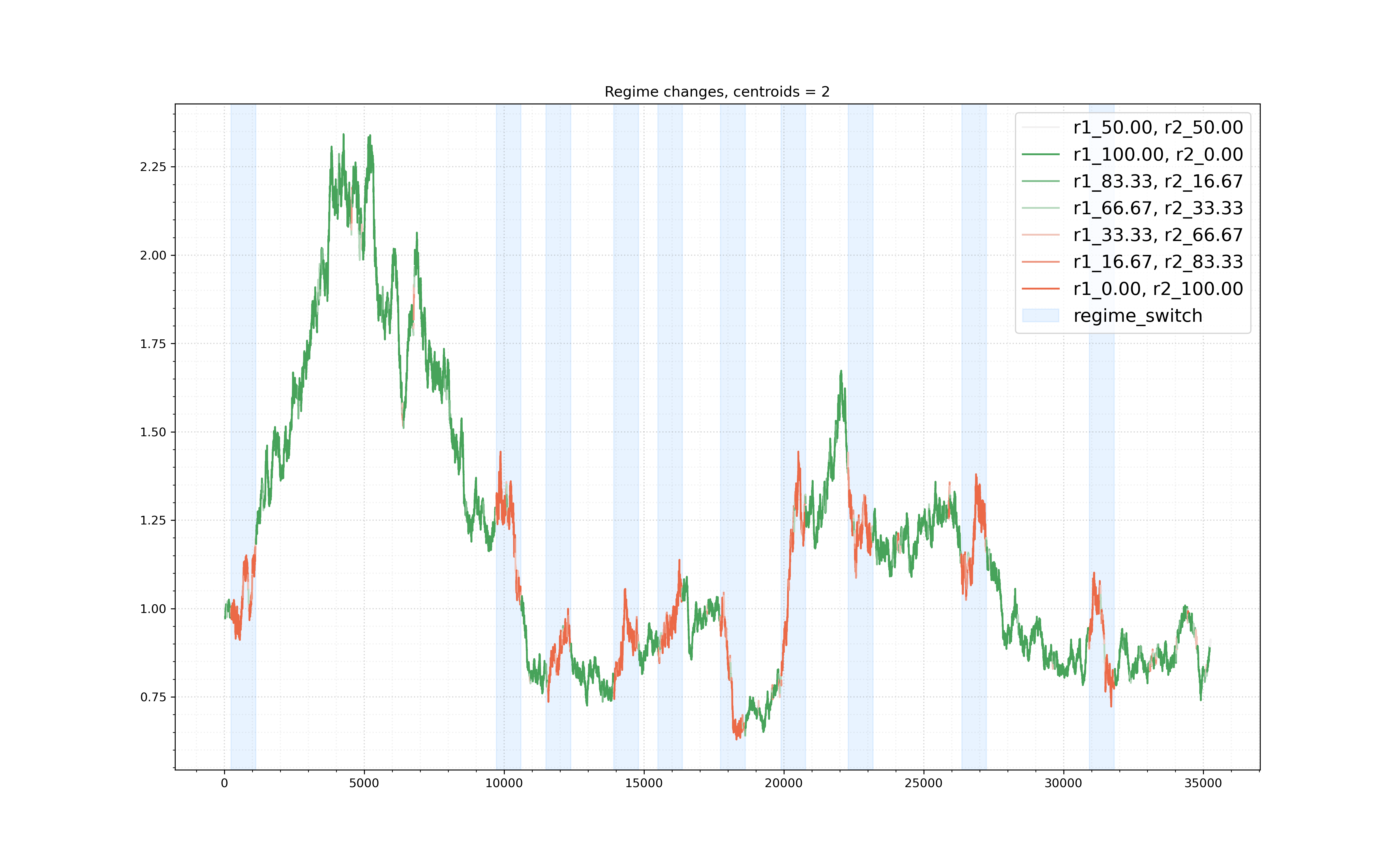}
			\caption{WK-means.}
			\label{fig:wassersteinhistoricalgbm}
		\end{subfigure}	
		\begin{subfigure}{0.5\linewidth}
			\centering
			\includegraphics[width=\textwidth]{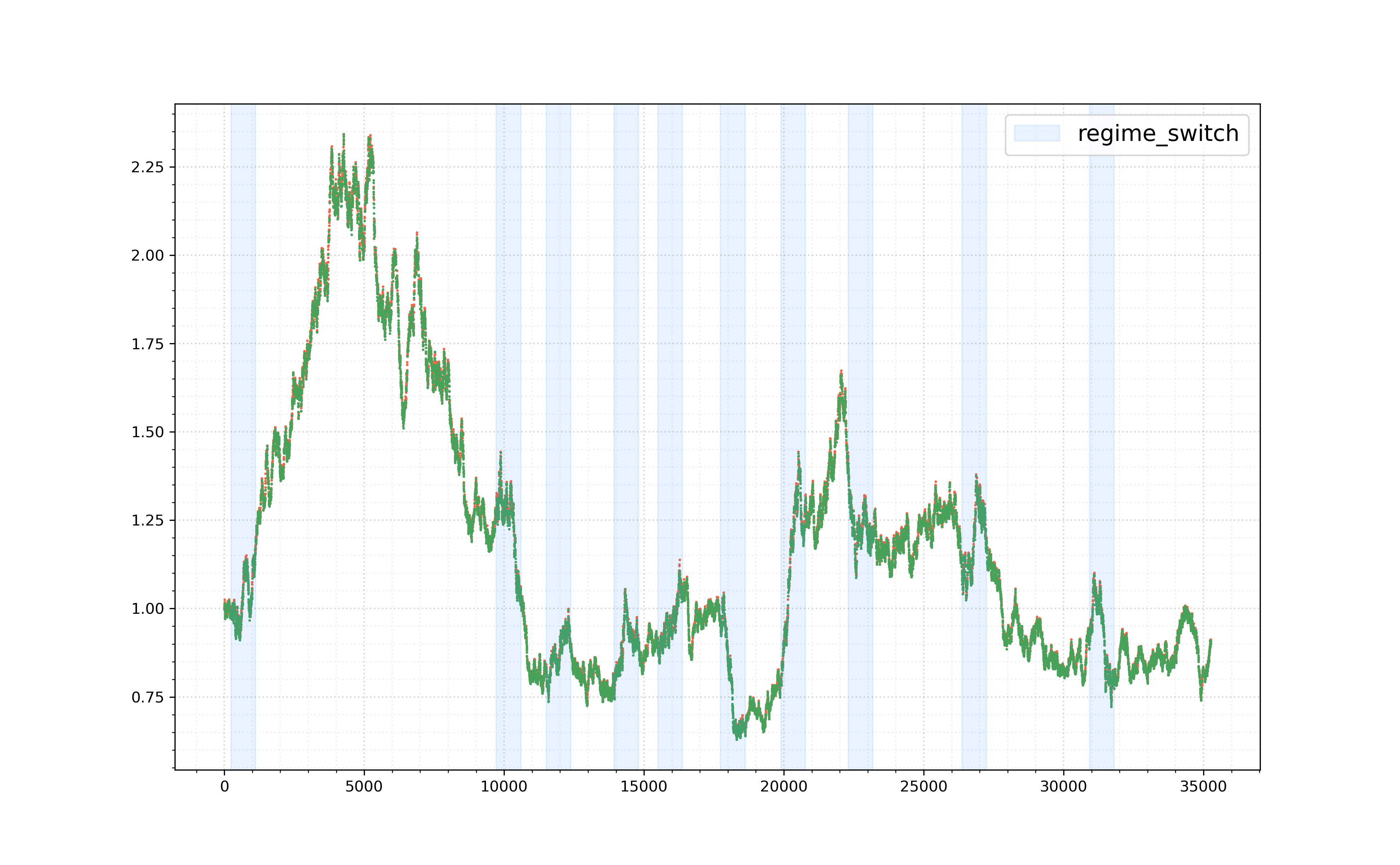}
			\caption{Hidden Markov model.}
			\label{fig:hmmhistoricalgbm}
		\end{subfigure}
		\caption{Historical cluster colouring plot, synthetic gBm, example run.}
		\label{fig:historicalgbm}
	\end{figure}
	
	We conclude this section by comparing the centroids obtained from either algorithm to the true measures. In this example, the distribution of the log-returns corresponding to either regime is given by 
	\begin{align*}
		\mesp_{\theta_\mathrm{bull}} &= \mathrm{Normal}(-1.97\times 10^{-21}, 2.27\times 10^{-05}),\qquad\text{and} \\
		\mesp_{\theta_\mathrm{bear}} &= \mathrm{Normal}(-3.68\times 10^{-05}, 5.1\times 10^{-05}).
	\end{align*}
	Since the distribution of log-returns in this model are Gaussian, we study the mean and variance of our obtained centroids and compare these to the true vales. 
	
	Figure \ref{fig:gbmapprox} summarizes our findings. In the first row, we display the histogram of the partition means $\{\ex[\mu_i]\}_{1\le i\le M}$ along with solid lines representing the true means $\ex[\mesp_{\theta_{\mathrm{bull}}}]$ and $\ex[\mesp_{\theta_{\mathrm{bear}}}]$. In Figures \ref{fig:momentmeangbm} and \ref{fig:wkmeangbm}, the dashed lines represent the bull and bear centroid means corresponding to the MK-means and WK-means algorithms respectively. In the second row, we repeat the same procedure with the variances $\{\mathrm{Var}(\mu_i)\}_{1\le i\le M}$, their true values, and in Figures \ref{fig:momentvargbm} and \ref{fig:wkvargbm} the centroid variances corresponding to the MK-means and WK-means algorithms respectively. 
	
	\begin{figure}[h!]
		\centering
		\begin{subfigure}{0.5\linewidth}
			\centering
			\includegraphics[width=\textwidth]{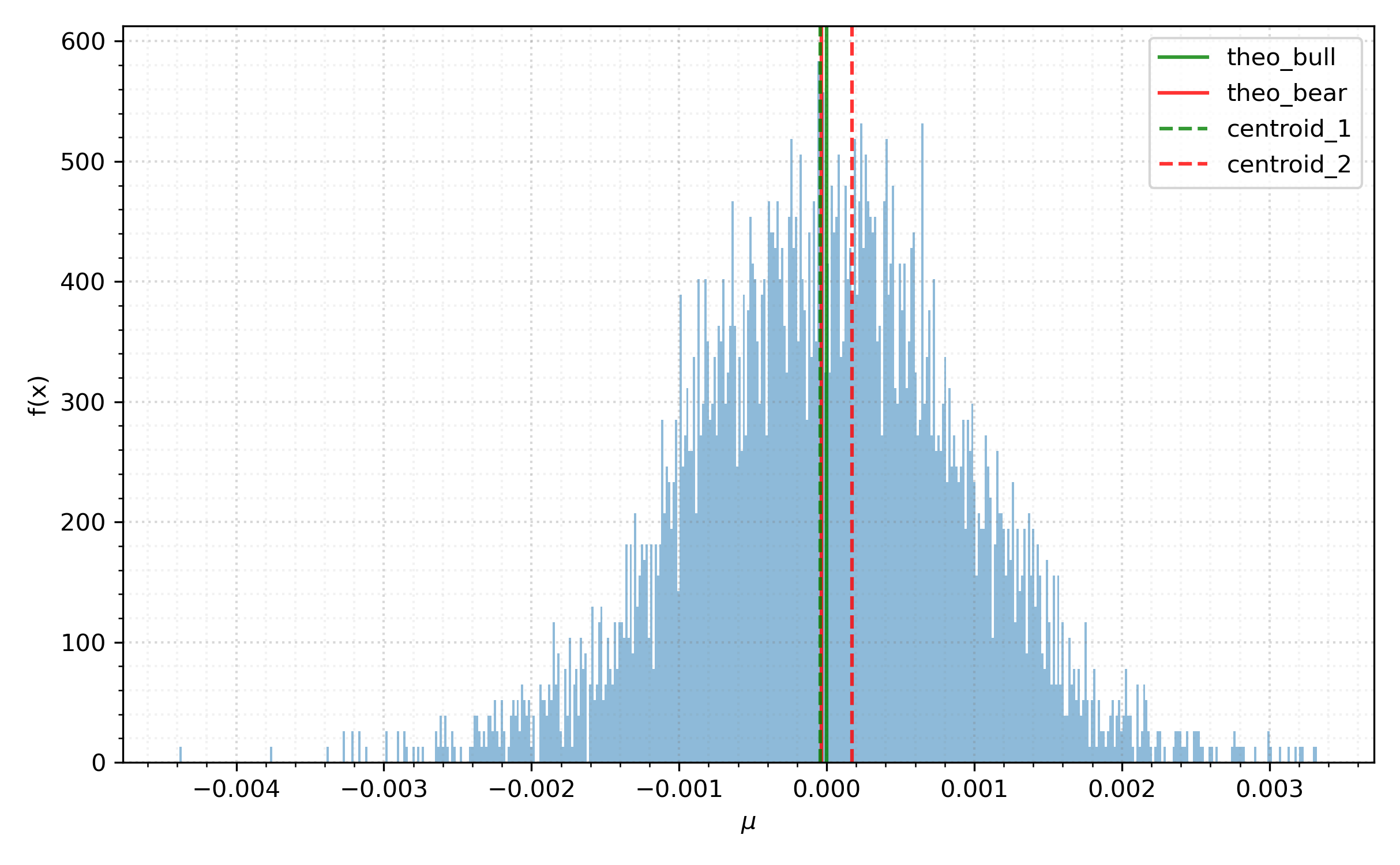}
			\caption{Distribution of $\{\ex[\mu_i]\}_{i\ge 0}$, moment.}
			\label{fig:momentmeangbm}
		\end{subfigure}%
		\begin{subfigure}{0.5\linewidth}
			\centering
			\includegraphics[width=\textwidth]{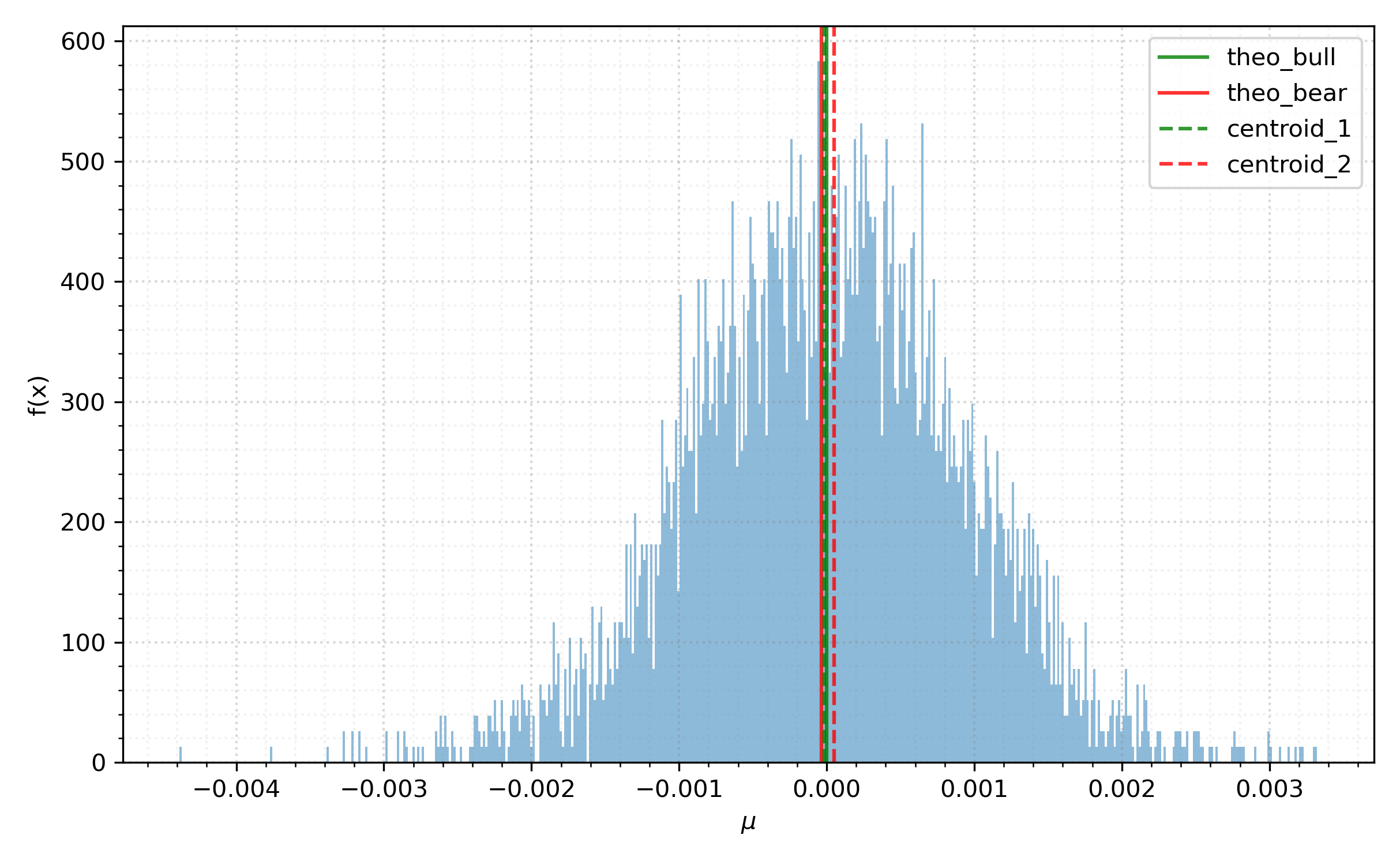}
			\caption{Distribution of $\{\ex[\mu_i]\}_{i\ge 0}$, Wasserstein.}
			\label{fig:wkmeangbm}
		\end{subfigure}
		\begin{subfigure}{0.5\linewidth}
			\centering
			\includegraphics[width=\textwidth]{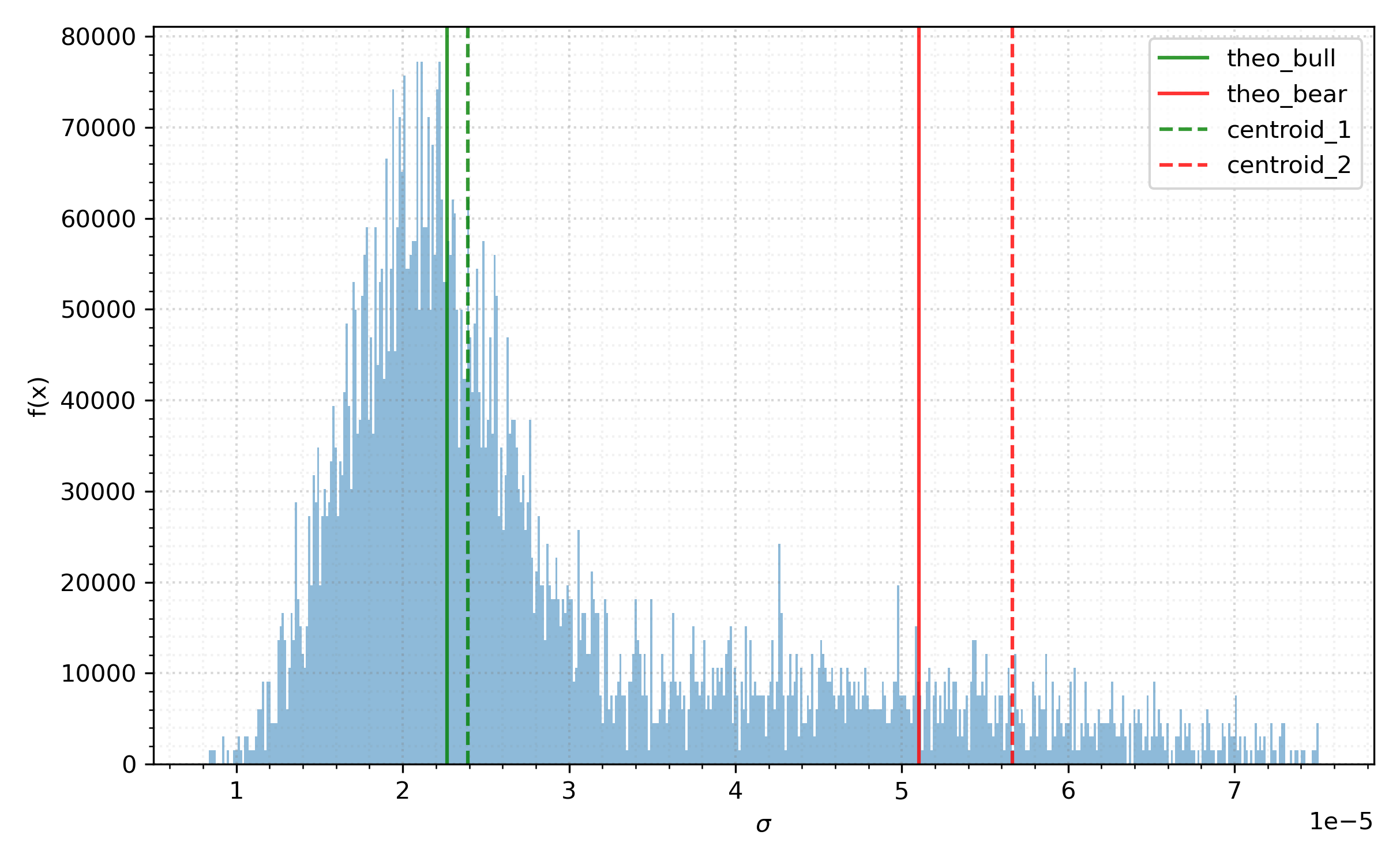}
			\caption{Distribution of $\{\mathrm{Var}(\mu_i)\}_{i\ge 0}$, moment.}
			\label{fig:momentvargbm}
		\end{subfigure}%
		\begin{subfigure}{0.5\linewidth}
			\centering
			\includegraphics[width=\textwidth]{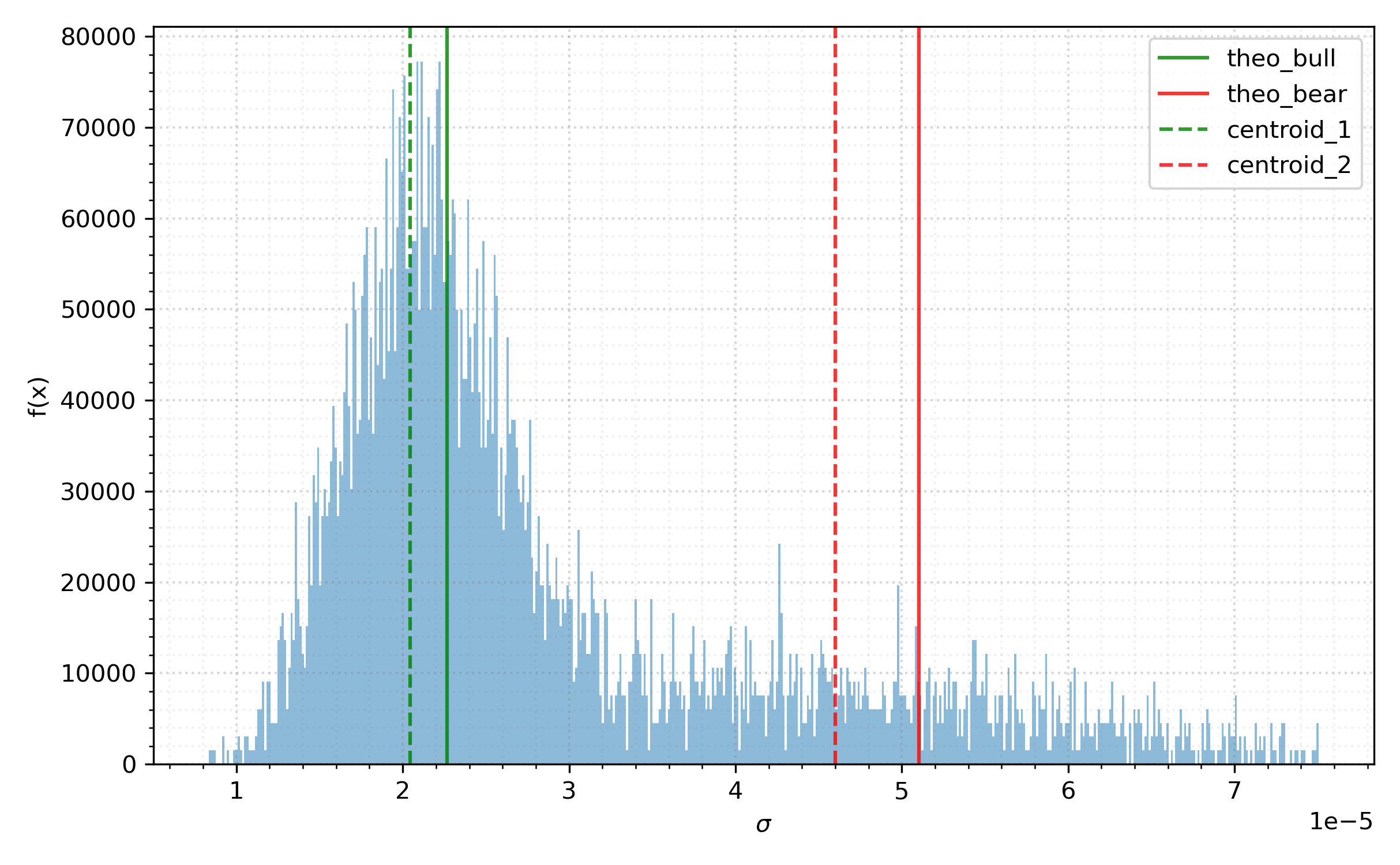}
			\caption{Distribution of $\{\mathrm{Var}(\mu_i)\}_{i\ge 0}$, Wasserstein.}
			\label{fig:wkvargbm}
		\end{subfigure}
		\caption{Approximations of mean and variance of true measures $\mesp_{i}, i =1,2$ by centroids of moments- and WK-means algorithms, gBm run.}
		\label{fig:gbmapprox}
	\end{figure}
	
	We see that the centroids derived from either algorithm perform well as estimators for the true centroids. We note that it is not altogether unsurprising that the Wasserstein algorithm does not significantly outperform the moments-based method here: since we are working under the assumption that distributions of log-returns under the ``true'' model are completely determined by their mean and variance. 
	
	\subsubsection{Merton jump diffusion}\label{subsubsec:merton}
	
	In this section, we outline a separate model used to generate synthetic data where the associated log-returns are non-Gaussian. In particular, we model stock prices by a Merton jump diffusion, which is given by the solution $S$ to the stochastic differential equation 
	\begin{equation}\label{eqn:mertondiffusion}
		dS_t = \mu S_t dt + \sigma S_t dW_t + S_{t-}dJ_t \qquad \text{for }t\ge 0, 	
	\end{equation}
	where
	\begin{equation*}
		J_t = \sum_{j=1}^{N_t}V_j-1. 
	\end{equation*}
	Here, $N_t \sim \mathrm{Po}(\lambda t)$ is a Poisson random variable, and $\ln(1+V_j) \sim \mathrm{Normal}(\gamma, \delta^2)$. Our model space $\mathcal{M}(\Theta)$ is given by
	\begin{equation}\label{eqn:modelmerton}
		\mathcal{M}(\Theta) = \mathrm{MJD}(\mu, \sigma, \lambda, \gamma, \delta) \qquad \text{for } \Theta \subset \mathbb{R}^5.
	\end{equation}
	The solution to (\ref{eqn:mertondiffusion}) is given by
	\begin{equation}\label{eqn:mertonpath}
		S_t = F(0, t)\mathcal{E}(\sigma W_t)\prod_{j=1}^{N_t}V_j,
	\end{equation}
	where $F(0, t) = S_0\exp(\mu t)$, and $\mathcal{E}: \mathbb{R}\to [0, +\infty)$ is the Dol{\'e}ans-Dade stochastic exponential. Let $R^\mathrm{M}_t = \ln(S_{t+dt})-\ln(S_t)$ be the log-return associated to a realisation of (\ref{eqn:mertonpath}) at a time $t\in[0, T]$ on a mesh with grid size $dt$. Then, \cite{synowiec} gives that
	\begin{align}
		\ex[R^\mathrm{M}_t] &= \big((\mu - \sigma^2/2) + \lambda \gamma\big)dt, \qquad \text{and } \label{eqn:mertonmean} \\
		\mathrm{Var}(R^\mathrm{M}_t) &= (\sigma^2 + \lambda(\delta^2 + \gamma^2))dt, \label{eqn:mertonvar}
	\end{align}
	and we will use these quantities to check the suitability of the centroids from either algorithm to the true measures.
	
	To test either algorithm on synthetic data as generated from (\ref{eqn:modelmerton}), we apply the same methodology as outlined in Section \ref{subsec:syndata}. That is, we define two sets of parameters $\theta_{\mathrm{bear}}, \theta_{\mathrm{bull}}$ and a partition $\Delta$ with regime changes $[s_i, s_i + l_i] \in R$ for $i=1,\dots, r$, where $R$ is given by (\ref{eqn:regimechanges}). We then run both clustering algorithms over a Merton jump diffusion with parameters $\theta_{\mathrm{bear}}$ over intervals in $R$ and parameters $\theta_{\mathrm{bull}}$ elsewhere. In regime switch dynamics are given by the parameter choices
	\begin{align*}
		\theta_{\mathrm{bull}} &= (0.05, 0.2, 5, 0.02, 0.0125), \qquad \text{and } \\
		\theta_{\mathrm{bear}} &= (-0.05, 0.4, 10, -0.04, 0.1).
	\end{align*}
	We again set $r=10$ and $l_i = 252\times 7 \times 0.5$ for $i=1,\dots, 10$. When a regime change occurs we shift the parameters of the Merton jump diffusion from $\theta_{\mathrm{bull}}$ to $\theta_{\mathrm{bear}}$, and revert them back when the regime change ends. Figure \ref{fig:mertonplots} gives an example path and the associated log-returns, with the periods of regime change highlighted in red.
	
	\begin{figure}[h!]
		\centering
		\begin{subfigure}{0.5\linewidth}
			\centering
			\includegraphics[width=\textwidth]{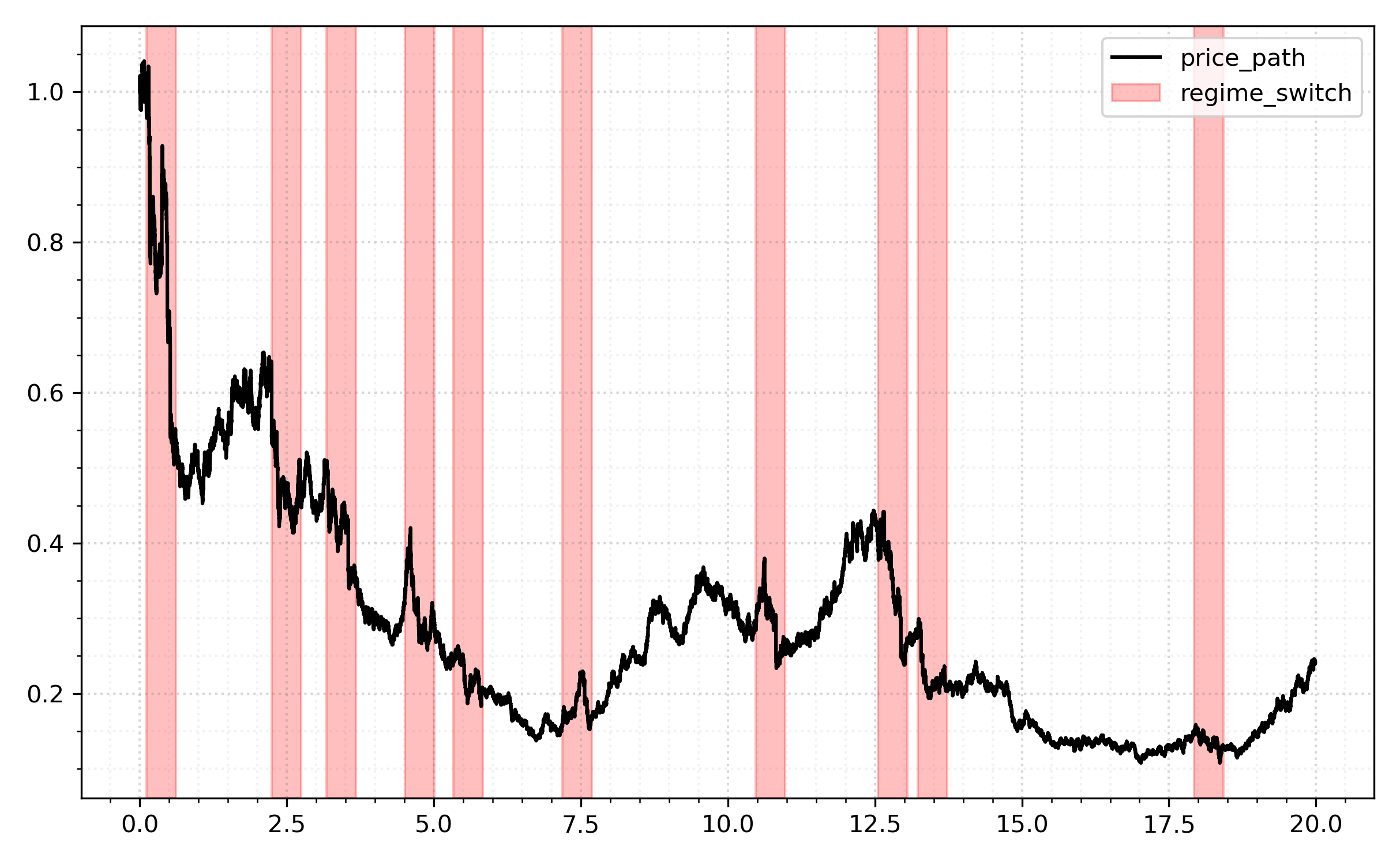}
			\caption{Merton jump diffusion path.}
			\label{fig:mertonpath}
		\end{subfigure}%
		\begin{subfigure}{0.5\linewidth}
			\centering
			\includegraphics[width=\textwidth]{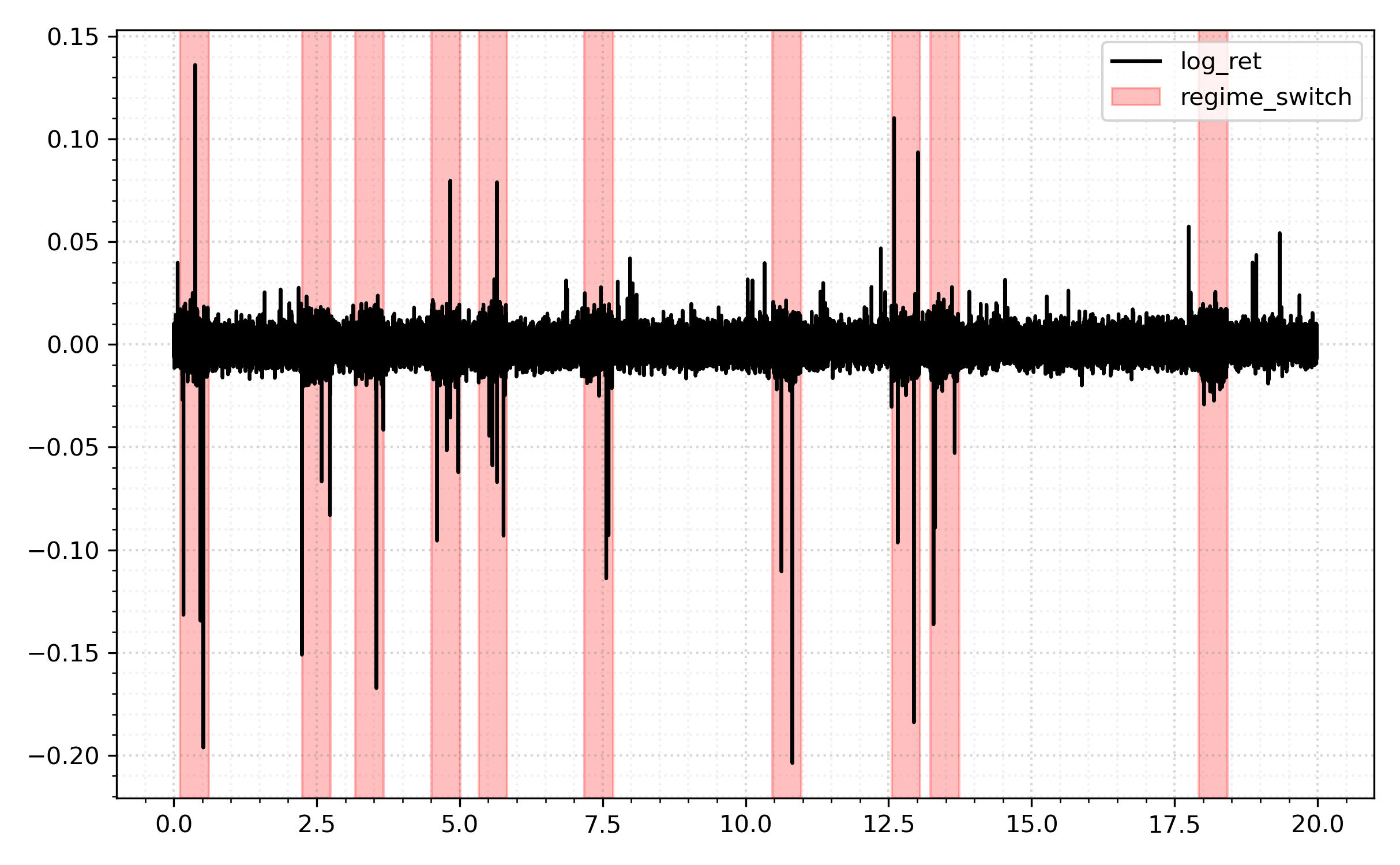}
			\caption{Log returns.}
			\label{fig:mertonlogrets}
		\end{subfigure}
		\caption{An iteration of a Merton jump diffusion path $\omega \to S^{\mathrm{J}}_t$ and the associated log-returns, regime changes highlighted.}
		\label{fig:mertonplots}
	\end{figure}
	
	For the example path presented in Figure \ref{fig:mertonpath}, we present plots from all three algorithms where applicable. Figure \ref{fig:meanvarmerton} gives the projection of the derived clusters for the moment- and WK-means algorithm. Here, one can see that the MK-means approach fails to discern between the two market regimes as it is not robust enough to adjust for outlier return series in the bear regime. By comparison, the Wasserstein approach is relatively robust to these outliers, and is able to correctly identify the two different regimes with their associated distributions. 
	
	\begin{figure}[h]
		\centering
		\begin{subfigure}{0.5\linewidth}
			\centering
			\includegraphics[width=\textwidth]{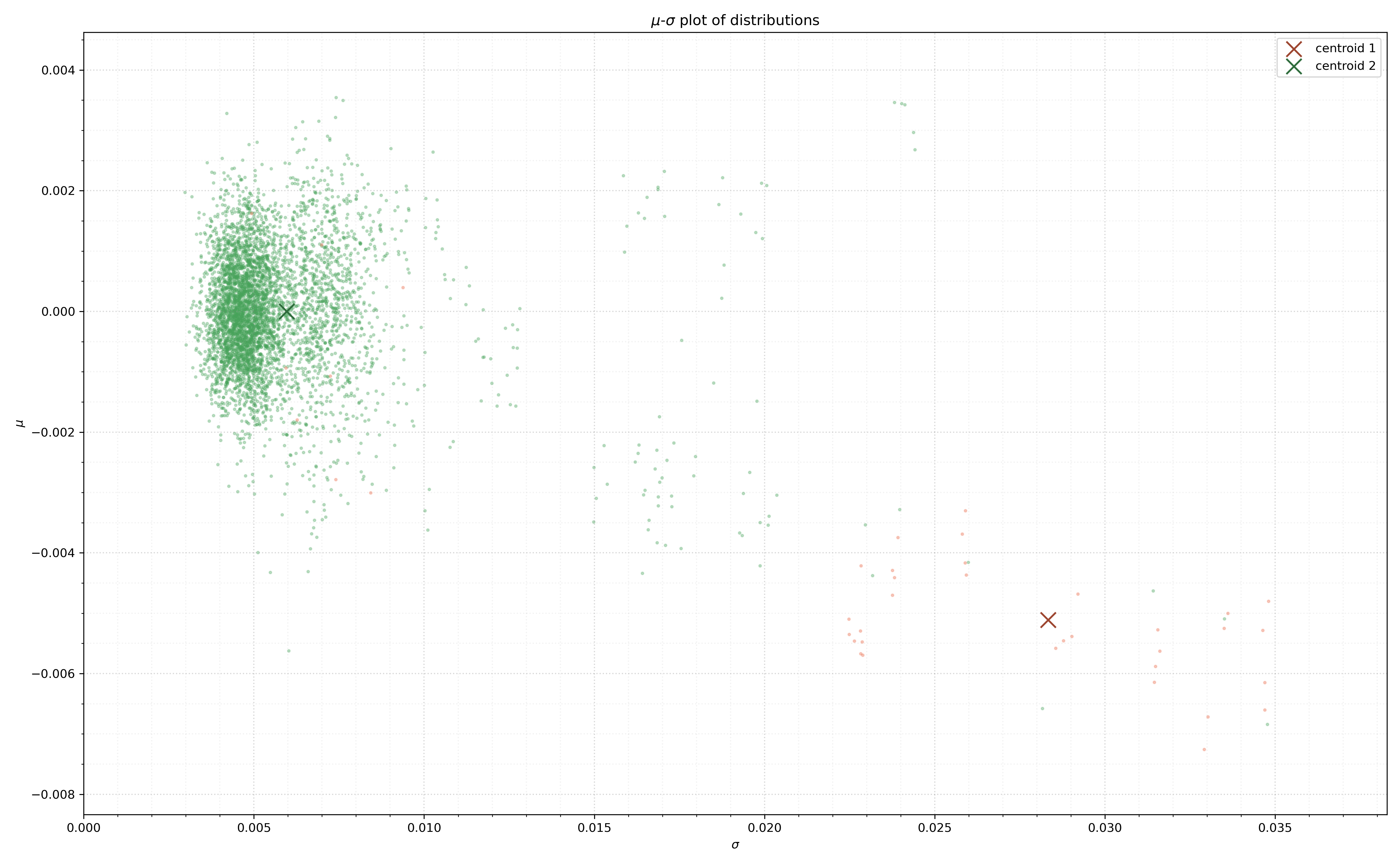}
			\caption{MK-means.}
			\label{fig:momentmertonmeanvar}
		\end{subfigure}%
		\begin{subfigure}{0.5\linewidth}
			\centering
			\includegraphics[width=\textwidth]{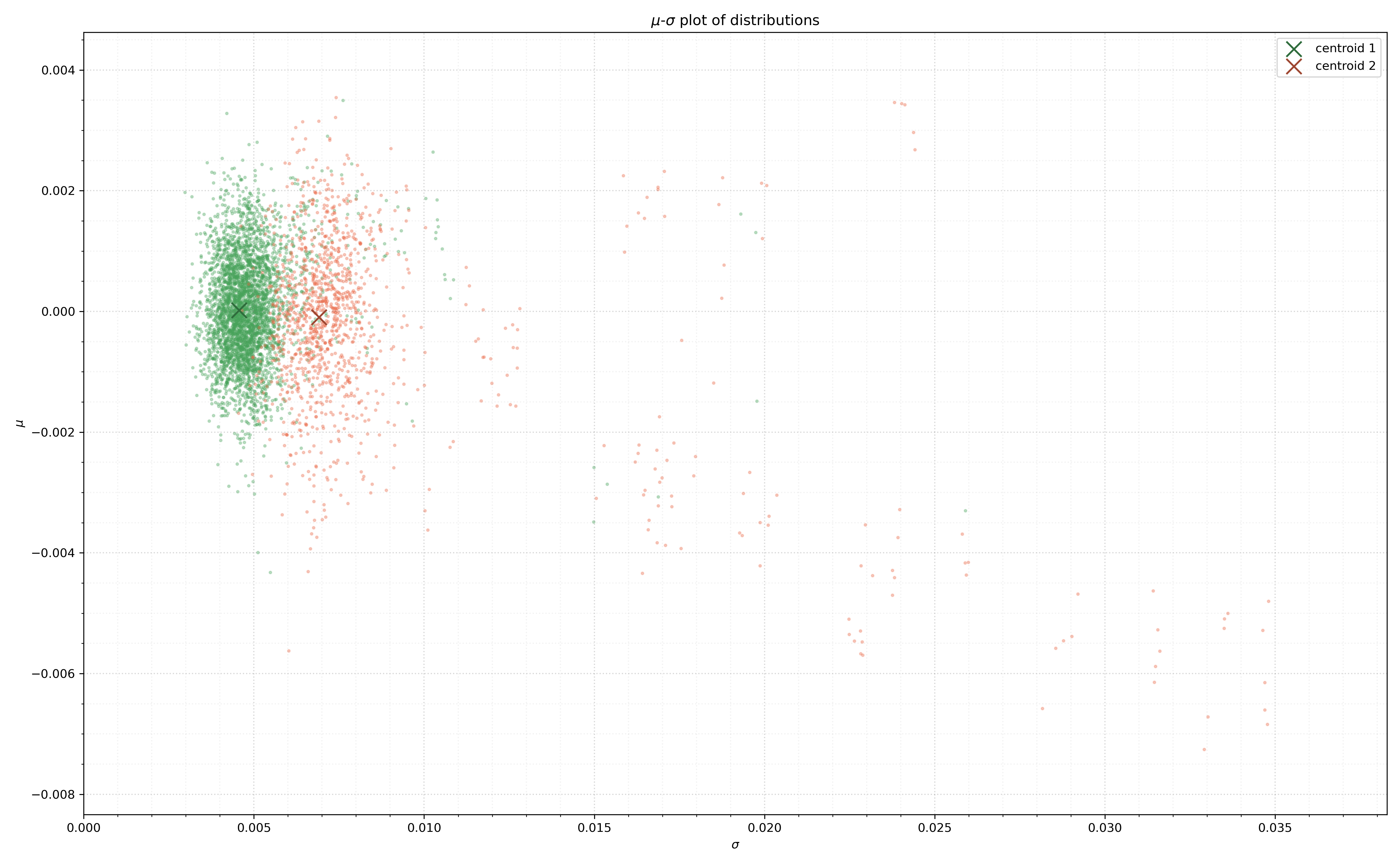}
			\caption{WK-means.}
			\label{fig:wassersteinmerton}
		\end{subfigure}
		\caption{Outputs of clustering algorithms in mean-variance space, Merton jump diffusion, example run.}
		\label{fig:meanvarmerton}
	\end{figure}
	
	Figure \ref{fig:wkmertonskewkurt} shows the WK-means clusters in skew-kurtosis space. We see here that the algorithm correctly identifies several distributions exhibiting positive skew as belonging to the bull regime. Furthermore, the algorithm is also able to discern between these distributions and those belonging to the bear regime, which are significantly more positively skewed.
	
	\begin{figure}[h]
		\centering
		\includegraphics[width=0.6\textwidth]{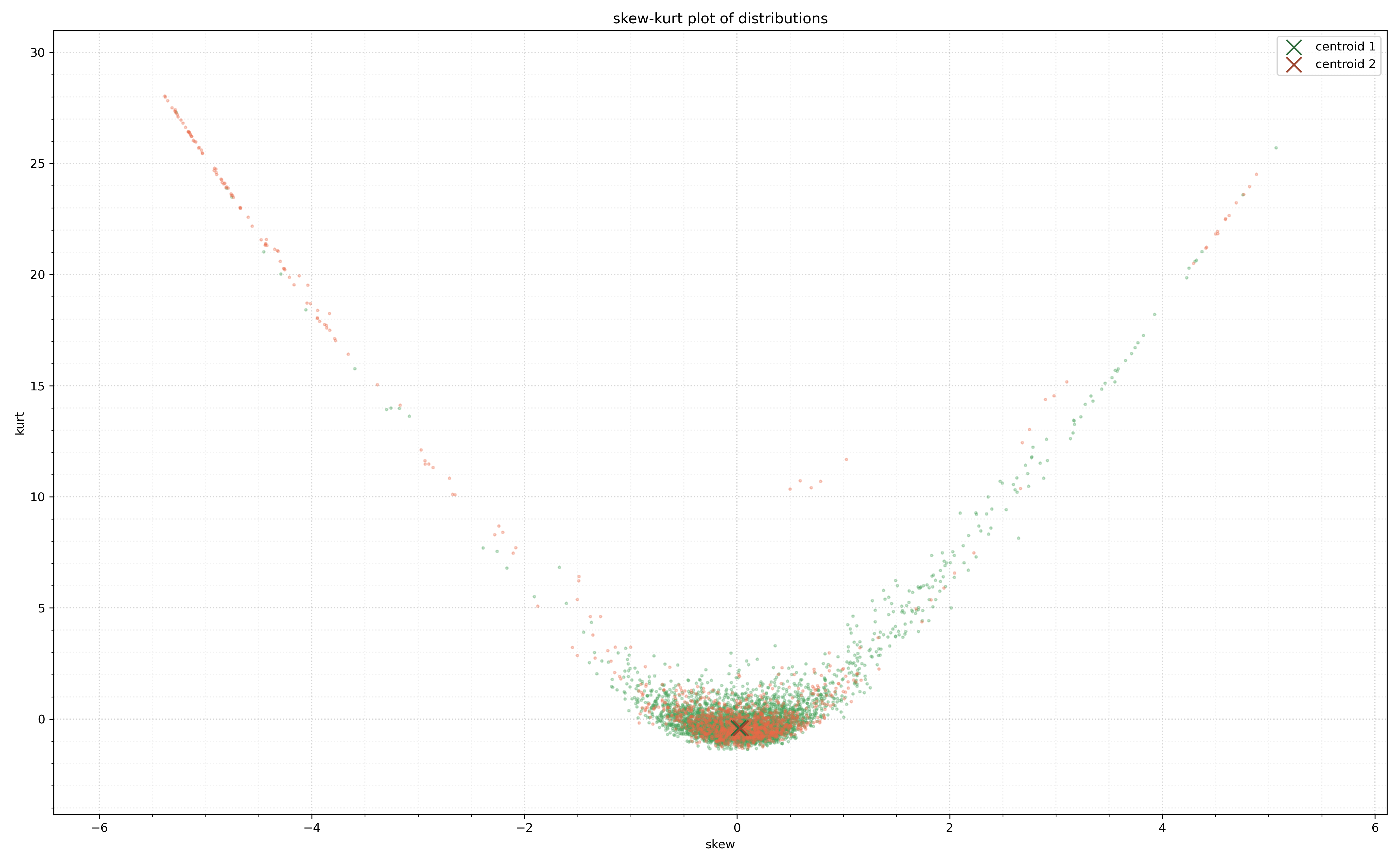}
		\caption{Clustered empirical Merton distributions in skew-kurtosis space, WK-means, example run.}
		\label{fig:wkmertonskewkurt}
	\end{figure}
	
	We summarize the results for $n=50$ runs with the parameters $(\theta_{\mathrm{bull}}, \theta_{\mathrm{bear}})$. It is clear that both the MK-means and HMM approach are unable to discern between periods of regime change and normalcy. As expected, the Wasserstein approach does not fare as well as the Gaussian case, although the difference is negligible. Note that the high accuracy in the regime-on case for the HMM is due to the fact that it never identifies the alternate regime and places all points into the first category. For the path in Figure \ref{fig:mertonpath}, we give the historical colouring plot associated to a run of all three algorithms in Figure \ref{fig:histmerton}, which highlights the numerical results obtained from the table. 
	
	\begin{table}[h]	
		\begin{center}
			\begin{tabular}{ccccc}
				\toprule
				\textbf{Algorithm} & Total & Regime-on & Regime-off & Runtime\\ \midrule
				Wasserstein & $\boldsymbol{91.28}\% \pm 4.08\%$ & $\boldsymbol{86.87}\% \pm 3.1\%$  & $92.76\% \pm 4.43\%$ & $1.11s \pm 0.25s$  \\ \addlinespace
				Moment & $66.64\% \pm 3.42\%$ & $27.25\% \pm 8.73\%$ & $79.79\% \pm 7.40\%$ & $1.71s \pm 0.28s$\\
				\addlinespace
				HMM & $75.05\% \pm 0.01\%$ & $0.66\% \pm 0.04\%$ & $\boldsymbol{99.87}\% \pm 0.01\%$ & $0.66s \pm 0.04s$\\\midrule
			\end{tabular}
		\end{center}
		\caption{Accuracy scores with $95\%$ CI, Merton synthetic path, $n=50$ runs.}
		\label{table:mertonaccuracy}
	\end{table}
	
	\begin{figure}[h]
		\centering
		\begin{subfigure}{0.5\linewidth}
			\centering
			\includegraphics[width=\textwidth]{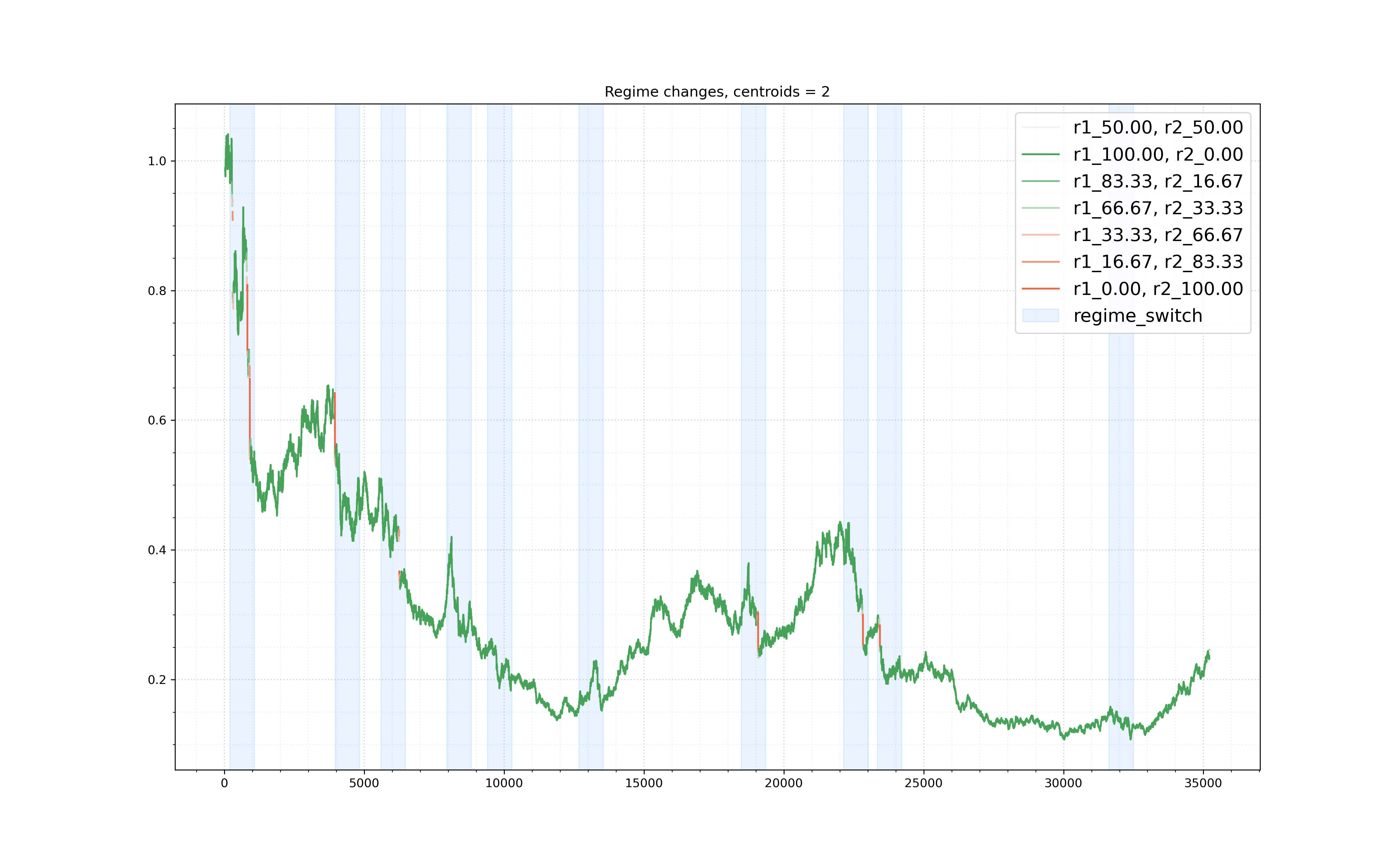}
			\caption{MK-means.}
			\label{fig:momentshistoricalmerton}
		\end{subfigure}%
		\begin{subfigure}{0.5\linewidth}
			\centering
			\includegraphics[width=\textwidth]{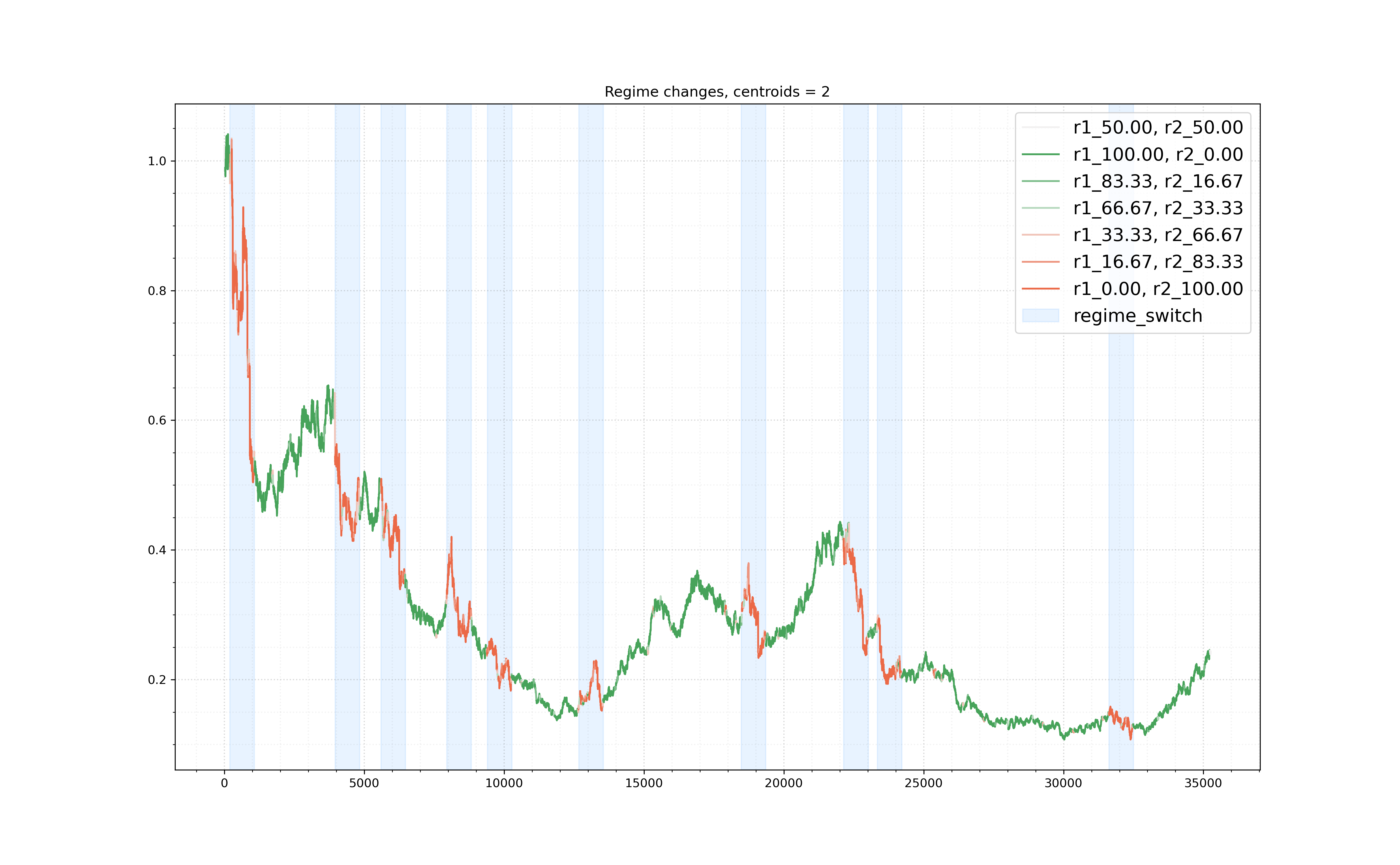}
			\caption{WK-means.}
			\label{fig:wassersteinhistoricalmerton}
		\end{subfigure}
		\begin{subfigure}{0.5\linewidth}
			\centering
			\includegraphics[width=\textwidth]{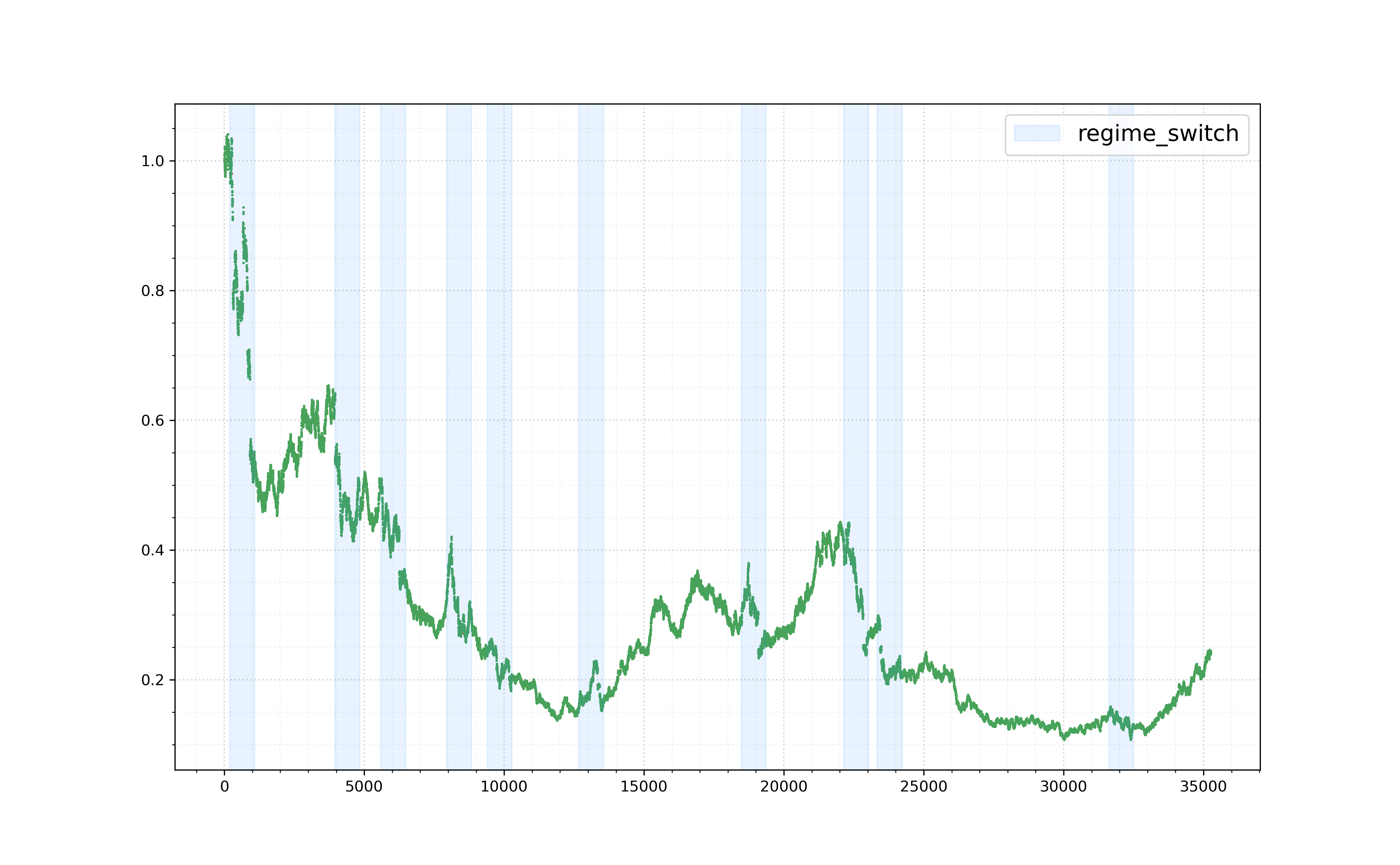}
			\caption{Hidden Markov model.}
			\label{fig:hmmhistoricalmerton}
		\end{subfigure}
		\caption{Historical cluster colouring on synthetic Merton jump diffusion price path, example runs.}
		\label{fig:histmerton}
	\end{figure}
	
	As we did with the geometric Brownian motion example, we conclude the results section with a comparison between the mean and variance of the centroids obtained from either algorithm, and those associated to the true distributions as given by equations (\ref{eqn:mertonmean}) and (\ref{eqn:mertonvar}). As expected, from Figures \ref{fig:momentmeanmerton} and \ref{fig:momentvarmerton} it is clear that the MK-means centroids do a poor job of approximating the true measures associated to the bull and bear regimes. We compare this to Figures \ref{fig:wkmeanmerton} and \ref{fig:wkvarmerton}, where the mean and variance of the WK-means centroids are much closer to the theoretical counterparts. 
	
	\begin{figure}[h!]
		\centering
		\begin{subfigure}{0.5\linewidth}
			\centering
			\includegraphics[width=\textwidth]{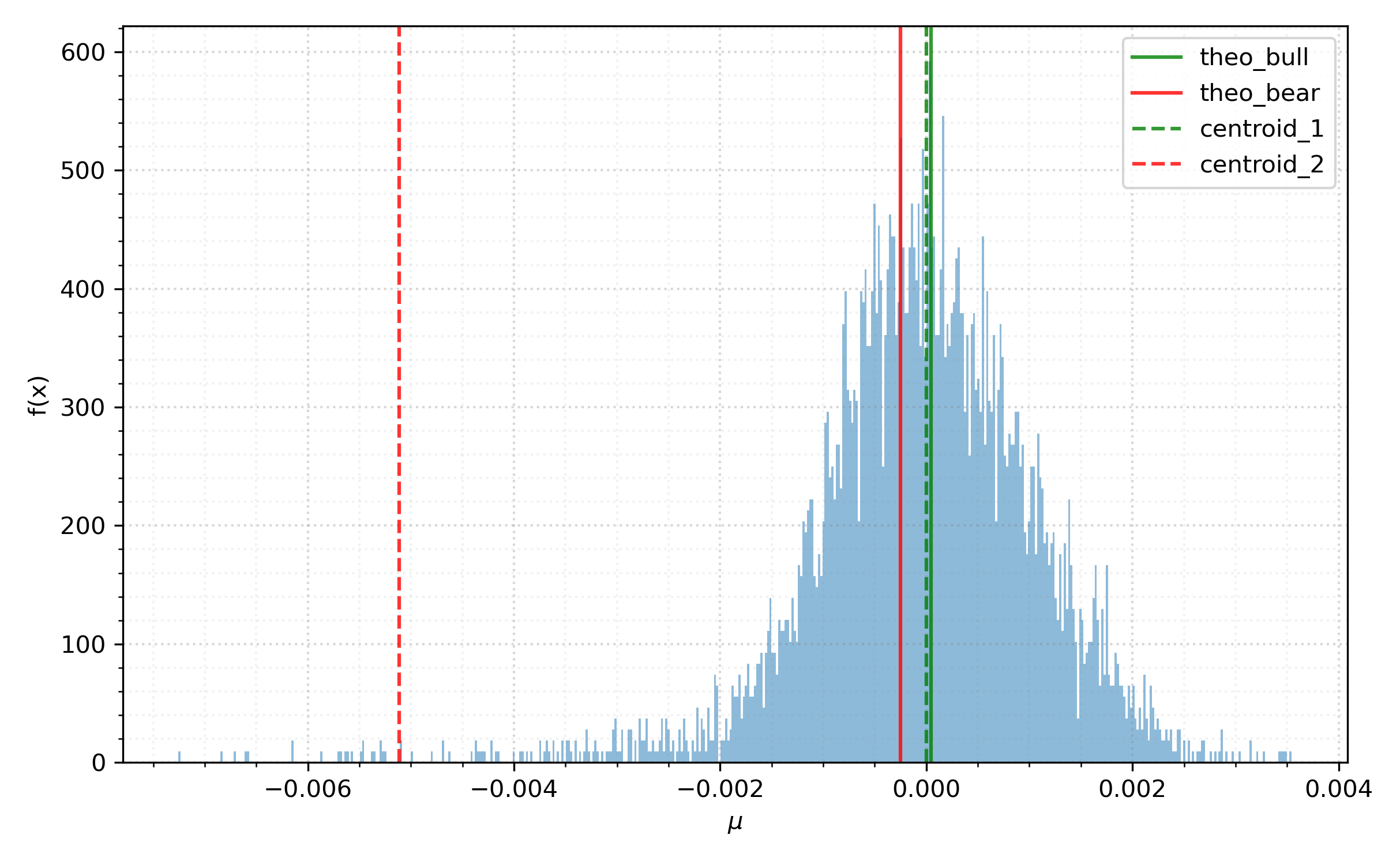}
			\caption{Distribution of $\{\ex[\mu_i]\}_{i\ge 0}$, MK-means.}
			\label{fig:momentmeanmerton}
		\end{subfigure}%
		\begin{subfigure}{0.5\linewidth}
			\centering
			\includegraphics[width=\textwidth]{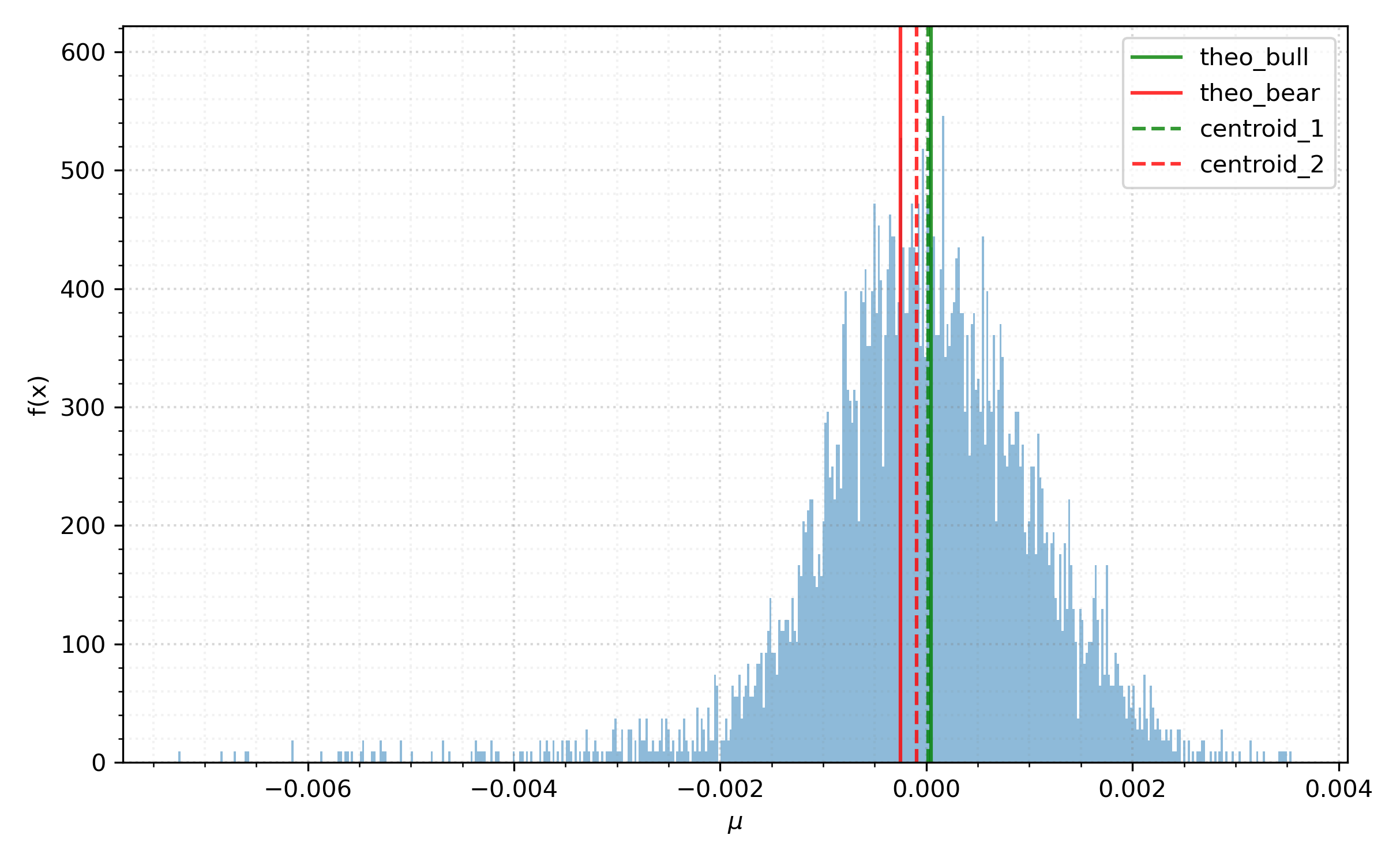}
			\caption{Distribution of $\{\ex[\mu_i]\}_{i\ge 0}$, WK-means.}
			\label{fig:wkmeanmerton}
		\end{subfigure}
		\begin{subfigure}{0.5\linewidth}
			\centering
			\includegraphics[width=\textwidth]{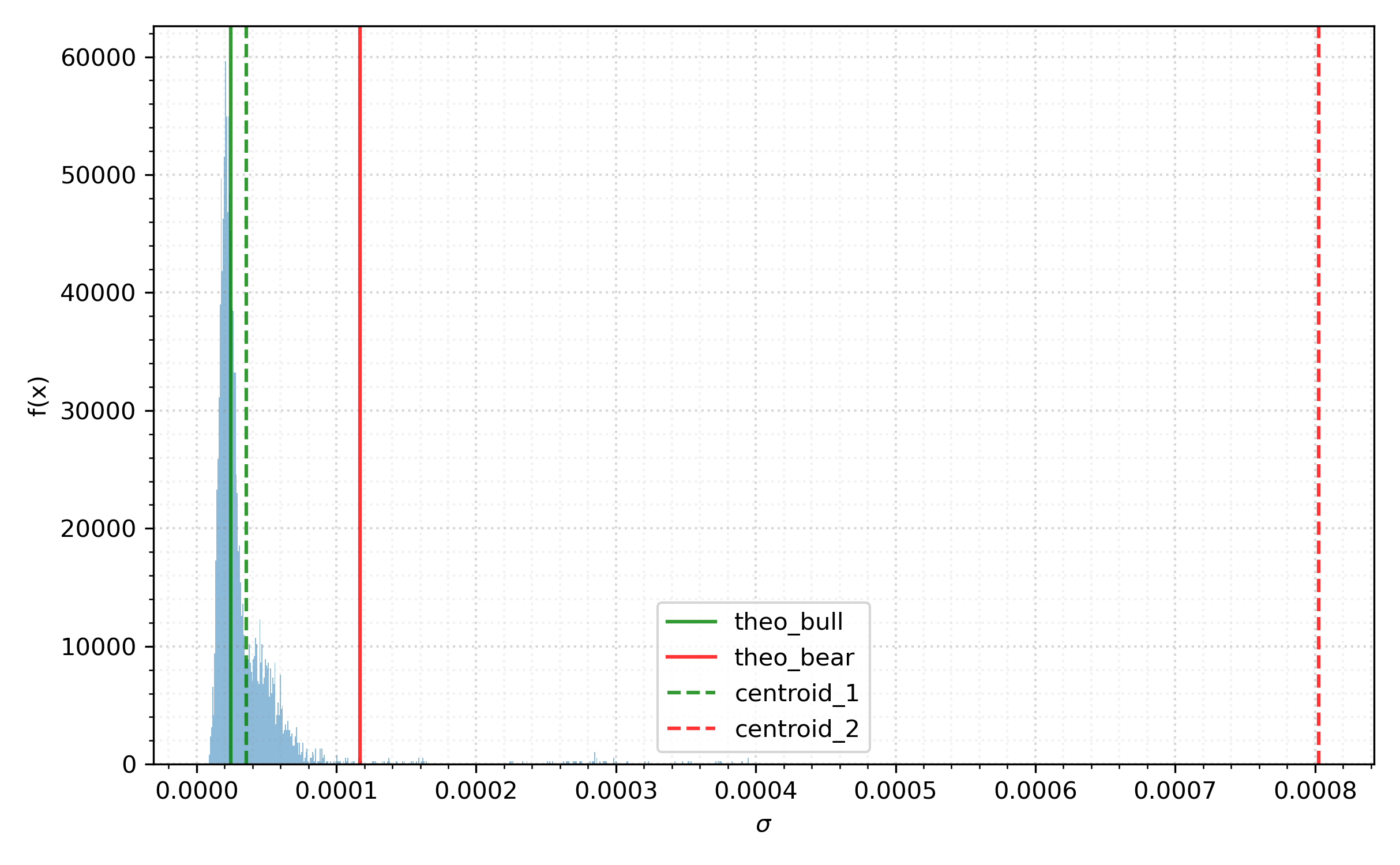}
			\caption{Distribution of $\{\mathrm{Var}(\mu_i)\}_{i\ge 0}$, MK-means.}
			\label{fig:momentvarmerton}
		\end{subfigure}%
		\begin{subfigure}{0.5\linewidth}
			\centering
			\includegraphics[width=\textwidth]{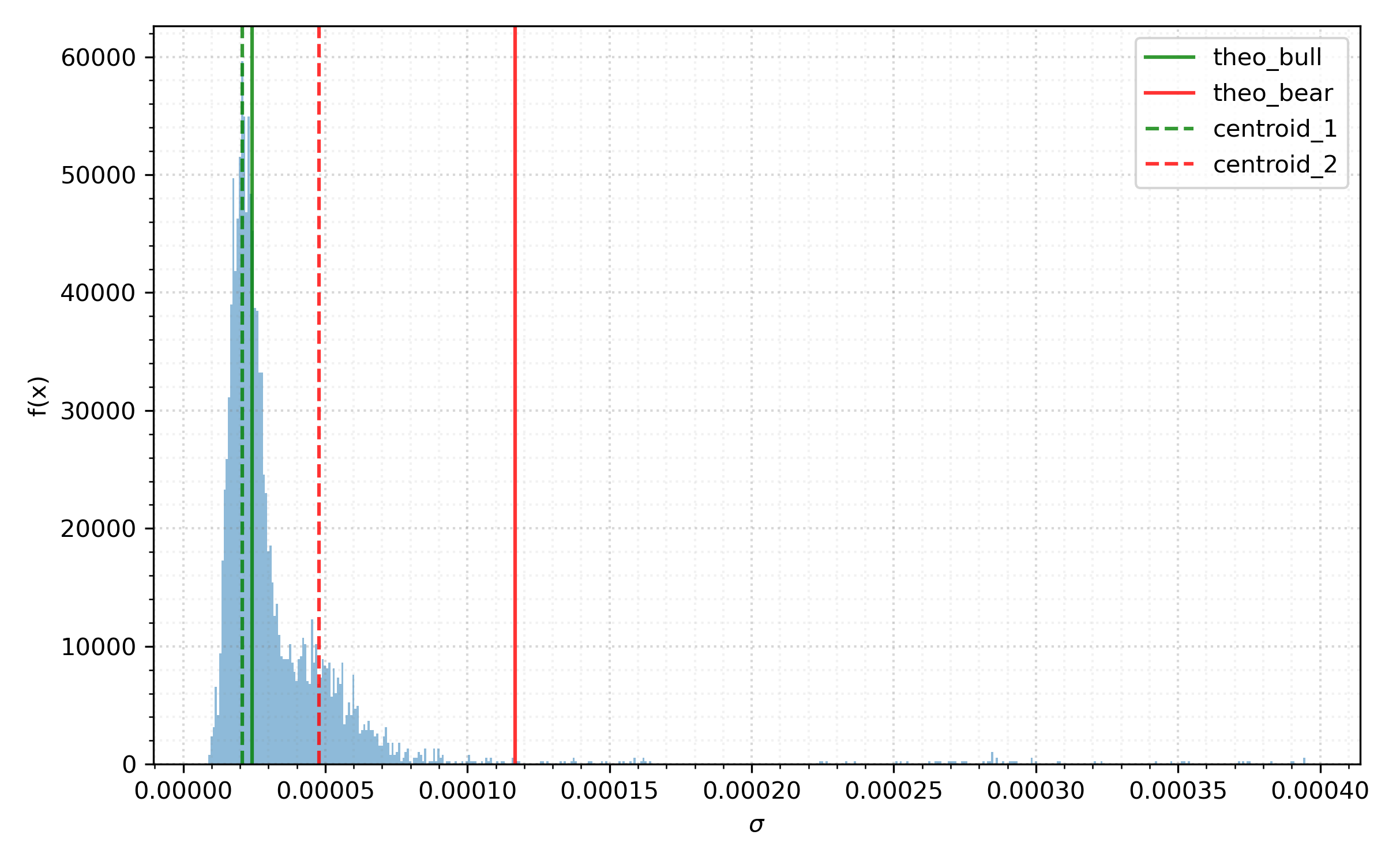}
			\caption{Distribution of $\{\mathrm{Var}(\mu_i)\}_{i\ge 0}$, Wasserstein.}
			\label{fig:wkvarmerton}
		\end{subfigure}
		\caption{Approximations of mean and variance of true measures by centroids of moments- and WK-means algorithms.}
		\label{fig:mertonapprox}
	\end{figure}

	\subsection{Selection of hyperparameters}\label{subsec:hyperparameters}
	
	In this section we give a brief discussion regarding how the choice of hyperparameters $(h_1, h_2)$ affects the results of the WK-means algorithm. We begin with a discussion to the first hyperparameter $h_1$, which corresponds to the number of returns that form each empirical distribution within the clustering algorithm. Classically one would like to take as large a value of $h_1$ as is feasibly possible in order to best approximate empirically the true data-generating measure. In the regime clustering context, however, this is not always ideal: choosing $h_1$ to be too large can mean that regime changes are not captured, or (in a live data setting) the detection of such changes are lagged. However, certainly if $h_1$ is chosen too small spurious classifications dominated by noise will be made. Thus we believe that the choice of window length hyperparameter is more an art than a science and strongly depends on the application in mind. 
	
	Regarding the overlap hyperparameter $h_2$: heuristically, a larger value of overlap parameter (relative to $h_1$) can be thought of in two ways. Mechanically it is a way of increasing the number of samples used in the clustering algorithm, which may be necessary in a low-data environment. Heuristically, by increasing the clustering set with measures that are very similar to each other, one is indirectly making a statement about how representative the observed sequence of log-returns (and thus clustering set measures) relative to what one might deem ``standard'' conditions. This phenomena can be seen in the simple case when one clusters on S\&P 500 data before and after the GFC: for $k=2$ and with $h_2 = 0$, one would expect that the outlier centroid $\overline{\mu}^2$ moves significantly faster to its new position than $\overline{\mu}^1$, whereas if $h_2$ is closer to $h_1$, one expects the centroids to not initially change as much during the onset of the GFC, since new observations are less constituent relative to the corpus of measures preceding them. 
	
	We note however that in general the overlap hyperparameter does not have too large an effect on centroids obtained (and, thus, clusters) assuming that one is not operating in too low a data environment, and $h_1$ is suitably chosen. We present the results of clustering on SPY for the hyperparameter choice $h^1 = (35, 28)$, the choice we made in Section \ref{subsec:realdata}, and $h^2 = (35, 0)$ in Figure \ref{fig:hyperparameters}. Here, one can see that the obtained centroids from either algorithm do not change drastically in spite of the lower data density.
		
	\begin{figure}[h!]
		\centering
		\begin{subfigure}{0.5\linewidth}
			\centering
			\includegraphics[width=\textwidth]{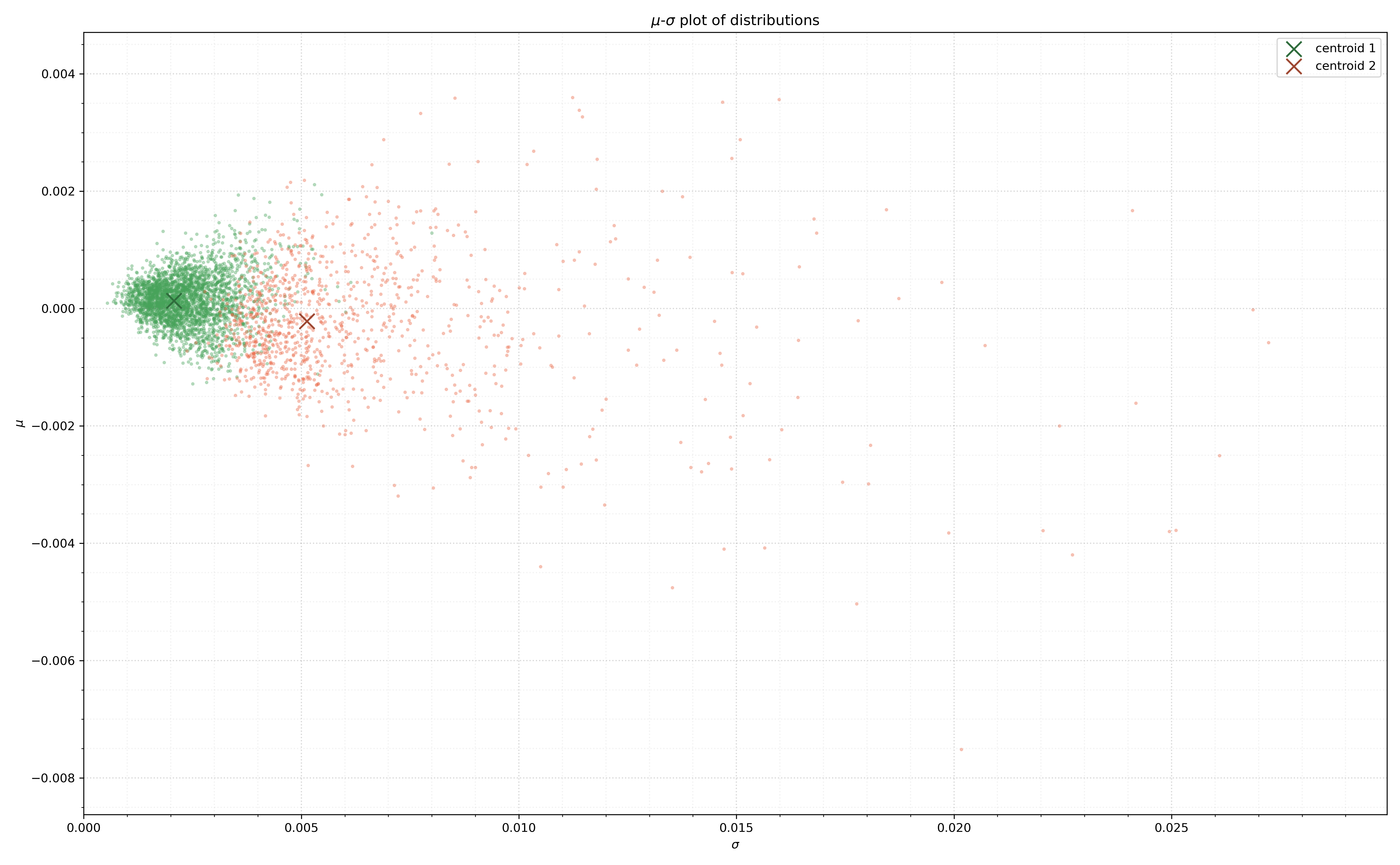}
			\caption{Mean-variance plot of clusters and centroids, $h^1=(35, 28)$.}
			\label{fig:h1meanvar}
		\end{subfigure}%
		\begin{subfigure}{0.5\linewidth}
			\centering
			\includegraphics[width=\textwidth]{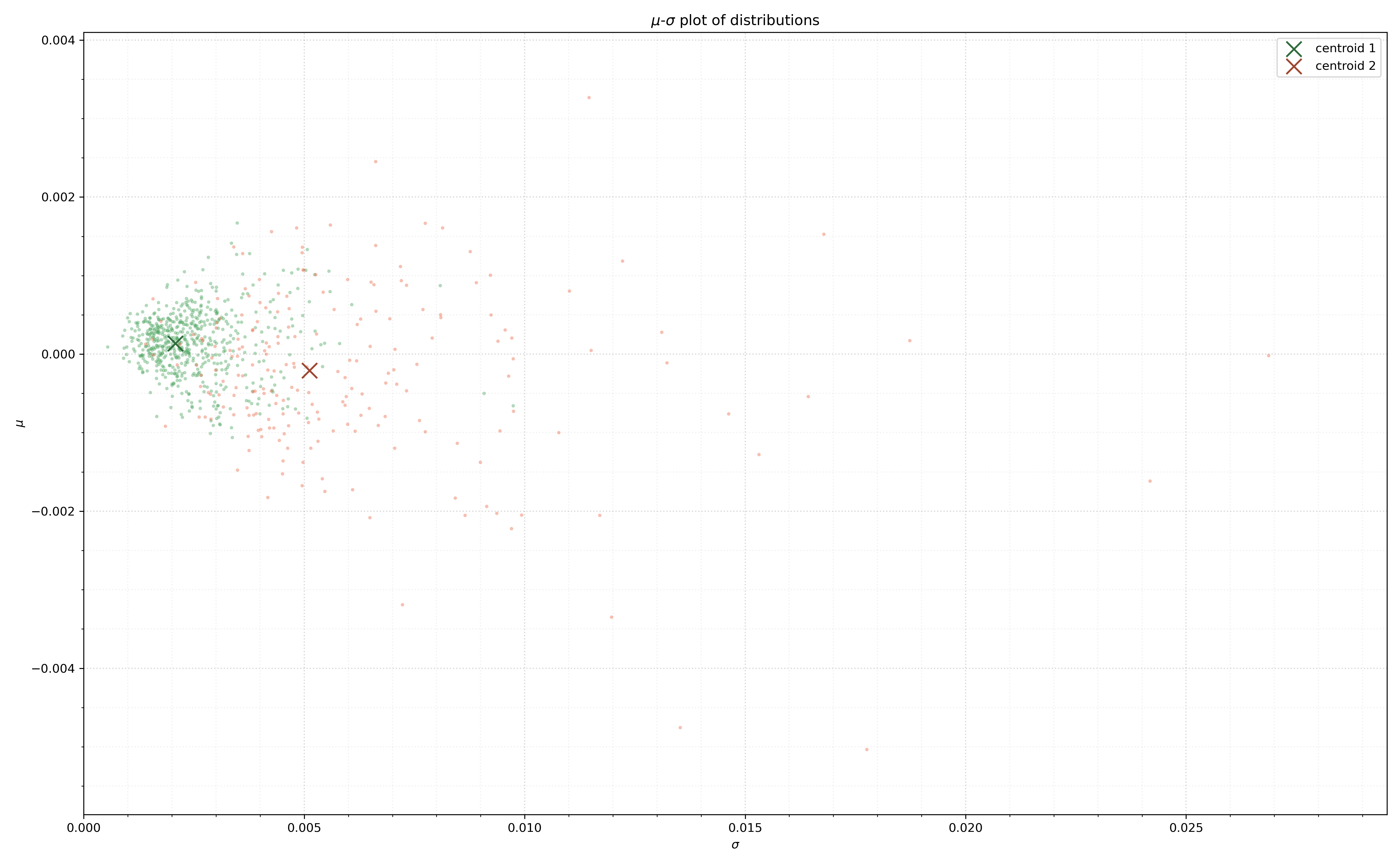}
			\caption{Mean-variance plot of clusters and centroids, $h^2=(35, 0)$.}
			\label{fig:h2meanvear}
		\end{subfigure}
		\begin{subfigure}{0.5\linewidth}
			\centering
			\includegraphics[width=\textwidth]{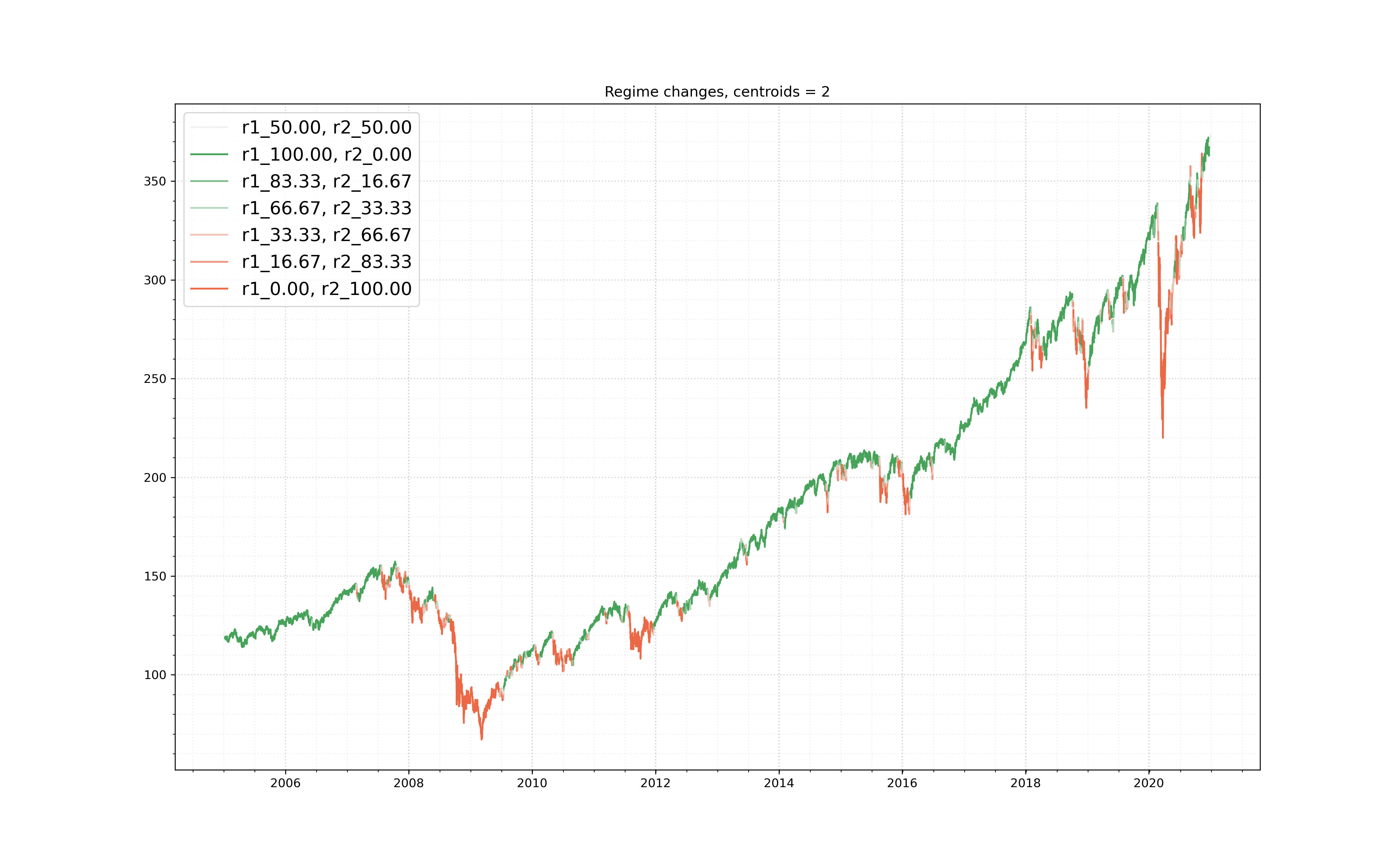}
			\caption{Historical backtest plot, $h^1=(35, 28)$.}
			\label{fig:h1backtest}
		\end{subfigure}%
		\begin{subfigure}{0.5\linewidth}
			\centering
			\includegraphics[width=\textwidth]{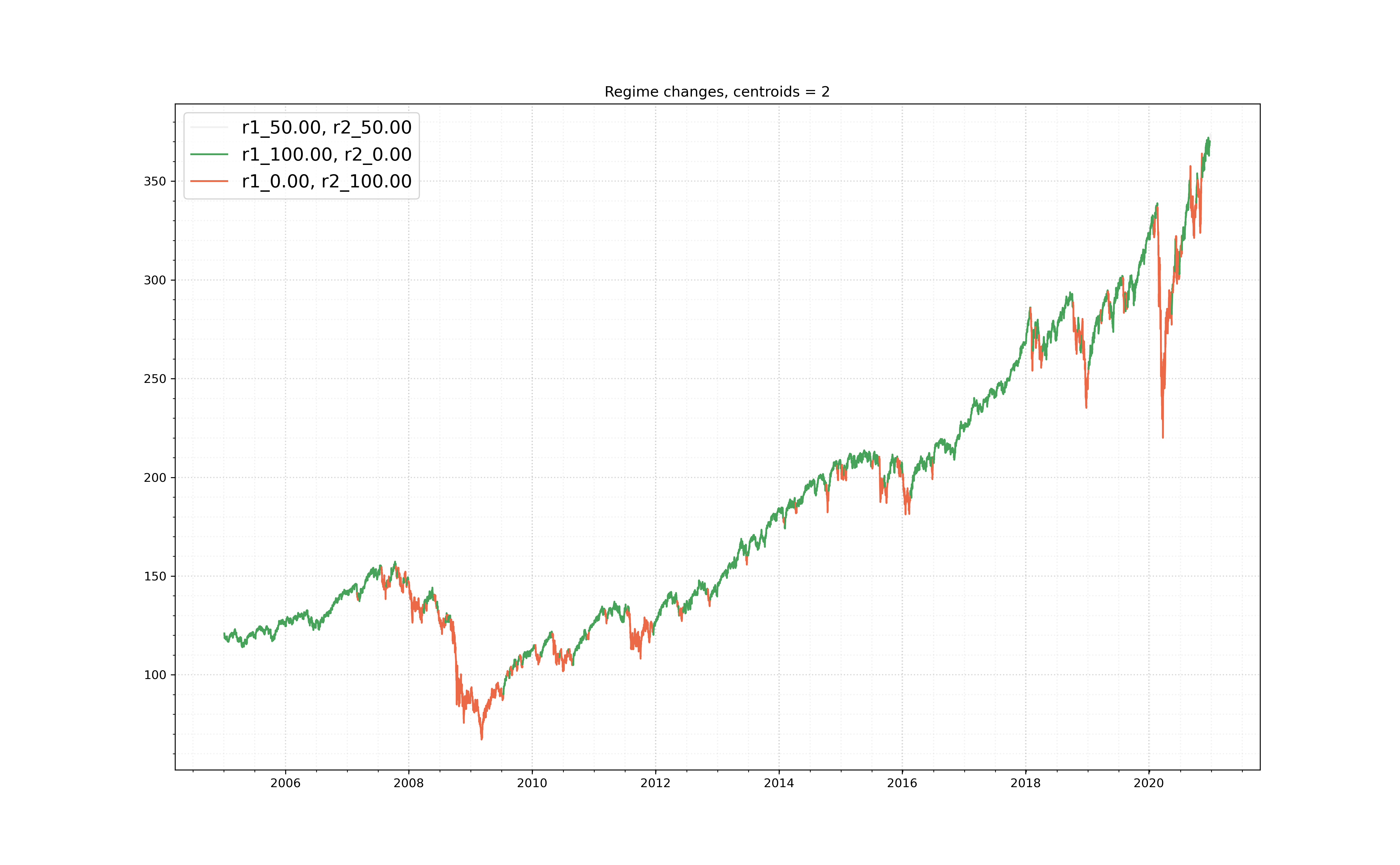}
			\caption{Historical backtest plot, $h^2=(35, 0)$.}
			\label{fig:h2backtest}
		\end{subfigure}
		\caption{WK-means results with $h^1$ and $h^2$ hyperparameter choices, SPY price path.}
		\label{fig:hyperparameters}
	\end{figure}

	Indeed a simple Kolmogorov-Smirnov two-sample test between the first centroids $(\overline{\mu}^{h_1}_1, \overline{\mu}^{h_2}_1)$ returns a test statistic score of $0.02857$ with an associated $p$-value of $1.0$, and the second set of centroids $(\overline{\mu}^{h_1}_2, \overline{\mu}^{h_2}_2)$ yields the same score and $p$-value. 
	
	Finally, we test the effect of changing the window length parameter $h_1$. As stated in the beginning of the section, there exist many incorrect choices for $h_1$ in practice (too small, or too large), but reasonable choices will tend to give robust results. To show this, we chose a sequence of window lengths $H_1 = (7 + 7i)_{i=1}^10$ and set the overlap parameter $h_2 = \lfloor 3h_1/4\rfloor $. We then ran the regime-switching experiment from Section \ref{subsec:syndata} with the same switching dynamics as in Section \ref{subsubsec:merton}, repeating the experiment $30$ times for each hyperparameter choice. In particular, regime changes in this setting lasted for half a year, which amounts to 882 time steps given our year mesh, which means that only an unreasonably large rolling window would miss the given regime change. Figure \ref{fig:h1tests} gives the various accuracy results (total, regime-on, and regime-off) for the window length parameters given by $H_1$. Past a certain threshold, the window length size does not matter and the model's accuracy converges. 
	
	\begin{figure}[h]
		\centering
		\includegraphics[width=0.8\textwidth]{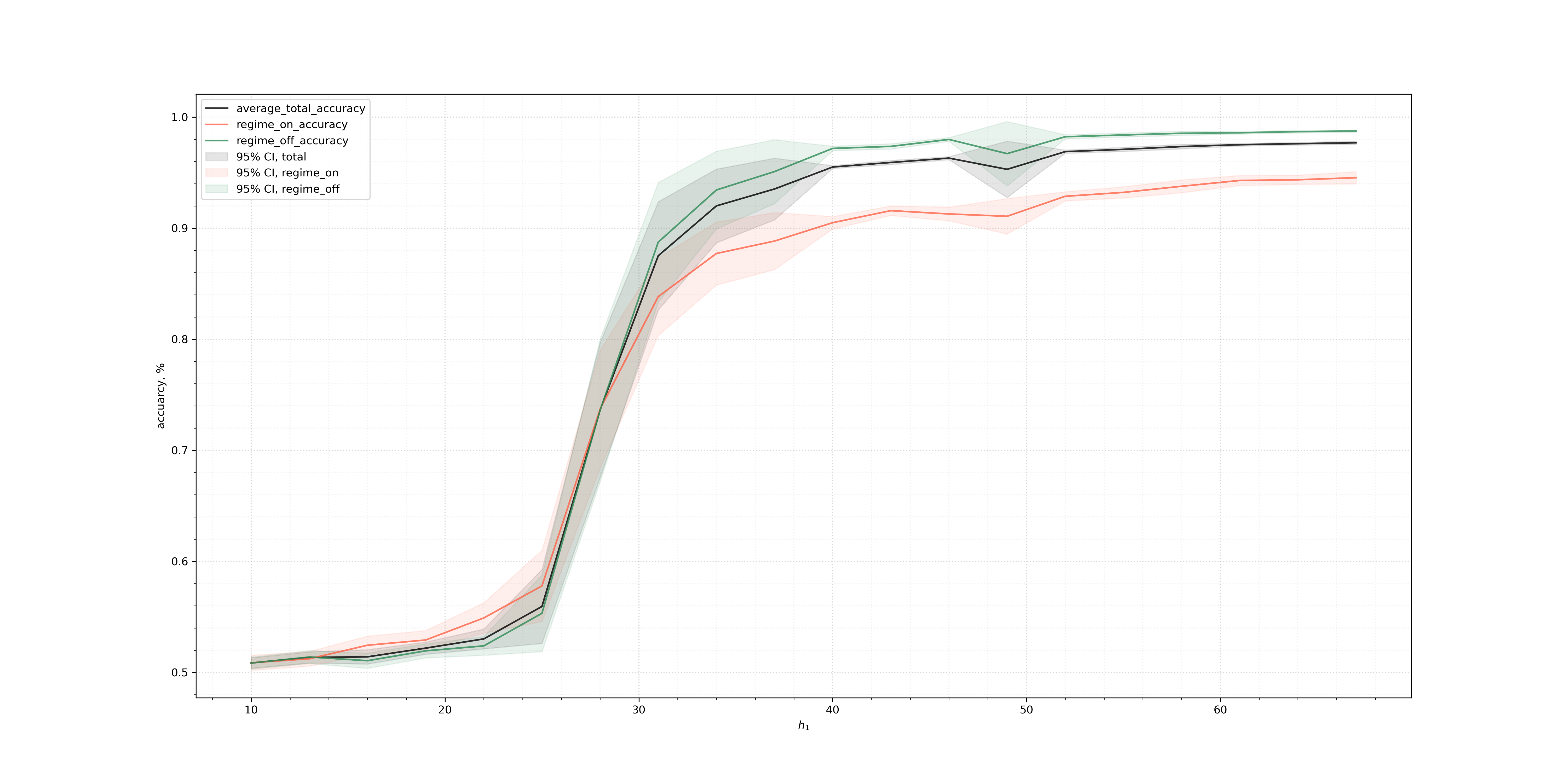}
		\caption{The effect of increasing $h_1$ on accuracy scores.}
		\label{fig:h1tests}
	\end{figure} 

	\section{Conclusion}
	
	In this paper, we have shown that a slight modification to the $k$-means algorithm does an excellent job at classifying partitions of market returns into regimes. We have verified this on both real data by comparing to known periods of market instability, and on synthetic data where we explicitly determine when the regime changes occur. We have compared this approach to a more standard moments-based algorithm which did not perform as well when returns were non-Gaussian, and a more classical approach using a hidden Markov model, which failed to accurately discern between regimes in the synthetic case. We also showed that clusters obtained via the Wasserstein approach are more self-similar than those derived from the moment-based method.
	
	Future research would include employing other clustering algorithms rather than the standard k-means approach; for instance, fuzzy and hierarchical clustering methods. Further study into the robustness of the derived clusters under the choice of hyperparameters (that is, partitioning of the underlying time series) would also be relevant to understand how stable derived clusters are. We also note that there exist methodologies for determining the optimal number of clusters $k$ to be used. Finally, a more analytic and rigorous study of the weak convergence of derived centroids to the true measures (in the synthetic data case) would also be of interest.

	\newpage
	
	\appendix
	
	\section{The $k$-means algorithm}\label{appendix:A}
	
	In this section, we outline some of the notation used in the paper, along with some standard results regarding the classical $k$-means algorithm.
		
	Recall that $X \in \mathcal{S}(V)$ is a stream of data over a normed vector space $(V, \norm{\cdot}_V)$. 
	\begin{definition}[Set of clusterings over $X$]
		We write
		\begin{equation*}
			\mathcal{C}(X) = \left\{\{\mathcal{C}_i\}_{0\le i\ge n} : \mathcal{C}_i \cap \mathcal{C}_j = \emptyset, \bigcup_{i=1}^n\mathcal{C}_i = X, n \in \mathbb{N} \right\}
		\end{equation*}
		to be the set of all possible (disjoint) clusterings over $X$.
	\end{definition}
	
	The $k$-means algorithm returns an element $\mathcal{C}^* \in \mathcal{C}(X)$ which is locally optimal with respect to the induced metric $d: V \times V \to [0, +\infty)$ on $V$. 
	
	Before continuing with a more detailed explanation of the $k$-means algorithm, we introduce the following definitions.
	
	\begin{definition}[Within-cluster variation]\label{def:withinclustervariation}
		Let $k \in \mathbb{N}$ and let $X \in \mathcal{S}(V)$ be a steam of data over a normed vector space $V$. Suppose $\mathcal{C}\subset \mathcal{C}(X)$ are disjoint clusters over $X$. Associate to each $\mathcal{C}_l$ its centroid $\overline{\mathsf{x}}_l$ for $l=1,\dots, k$. Then, for a given $\mathcal{C}_l$, the \emph{within-cluster variation} is defined as
		\begin{equation}\label{eqn:within-cluster}
			\mathrm{WC}(\mathcal{C}_l) = \sum_{\mathsf{x} \in \mathcal{C}_l}\norm{\mathsf{x} - \overline{\mathsf{x}}_l}^2_V \qquad \text{for }l=1,\dots, k.
		\end{equation}  
	\end{definition}
	\begin{definition}[Total-cluster variation]\label{def:totalclustervariation}
		With the notation of Definition \ref{def:withinclustervariation}, define 
		\begin{equation}\label{eqn:total-cluster}
			\mathrm{TC}(\mathcal{C}) = \sum_{i=1}^k \mathrm{WC}(\mathcal{C}_l)
		\end{equation}
		to be the \emph{total-cluster variation} corresponding to a clustering $\mathcal{C} \in \mathcal{C}(X)$ on the normed vector space $(V, \norm{\cdot}_V)$.
	\end{definition}
	
	Recall that $\{\mathcal{C}^n_l\} \in\mathcal{C}(X)$ denotes an intermediate disjoint clustering of the set $X$ at step $n \in \mathbb{N}$. The $k$-means algorithm first assigns nearest neighbours via (\ref{eqn:kmeansupdate}) with respect to $V$. The next step in the algorithm is to update the centroids via an aggregation function $\alpha: 2^V \to V$ which gives a central element $\overline{\mathsf{x}}_l \in V$ from the set of nearest neighbours $\mathcal{C}_l$ for $l=1,\dots, k$. In the classical $k$-means on $\real^d$, this function is given by
	\begin{equation}\label{eqn:kmeansaggregation}
		\alpha(\mathcal{C}_l) = \Bigg( \frac{1}{|\mathcal{C}_l|}\sum_{\mathsf{x} \in \mathcal{C}_l} x_j \Bigg)_{1\le j\le d}.
	\end{equation}
	Here, $|\mathcal{C}|$ denotes the cardinality of the set $\mathcal{C}$. Note that other choices of $(V, \norm{\cdot}_V)$ necessitate different aggregation methods depending on the structure of $V$.
	
	Continuing, the algorithm then updates the centroids via the function $\alpha$:
	\begin{equation*}
		\overline{\mathsf{x}}^n_l = a(\mathcal{C}^n_l) \qquad \text{for }l=1,\dots,k.
	\end{equation*}
	The new centroids $\overline{\mathsf{x}}^n$ are then compared to $\overline{\mathsf{x}}^{n-1}$ via the following stopping rule.
	\begin{definition}[$k$-means stopping rule]
		Suppose $(V, \norm{\cdot}_V)$ is a normed vector space. For fixed $k \in \mathbb{N}$, consider a loss function $l: V^{k} \times V^k \to \real_+$ given by
		\begin{equation}\label{eqn:loss}
			l(x, y) = \sum_{i=1}^k \norm{x_i - y_i}_{V},
		\end{equation}
		For a tolerance level $\varepsilon > 0$, the \emph{stopping rule} corresponding to the standard $k$-means algorithm is given by 
		\begin{equation}\label{eqn:stoppingrulekmeans}
			l(\overline{\mathsf{x}}^{n-1}, \overline{\mathsf{x}}^n) < \varepsilon,
		\end{equation}
		where $n \in \mathbb{N}$ denotes the step of the algorithm, and we take $V = \real^d$.
	\end{definition}

	\begin{remark}
		Various algorithms which attempt to find the optimal number of clusters $k$ that should be used to separate data $X$ consider (\ref{eqn:within-cluster}) as the loss to be minimised over all possible clusterings. We will not cover these algorithms in this paper (see, for instance \cite{ray1999determination}), and their applications to the MRCP are topics for future research.
	\end{remark}

	A given run of the $k$-means algorithm is characterised by the 2-tuple of centroids and nearest neighbour assignments $\left(\overline{\mathsf{x}}_l, \mathcal{C}_l\right)_{1\le l\le k}$. We conclude this section with the following proposition.
	
	\begin{proposition}
		The $k$-means algorithm converges in finitely many steps to a local minima. 
	\end{proposition}
	\begin{proof}
		Finiteness is guaranteed since the number of possible partitions of datum $X$ is at most $k^N$. Thus, the function $\mathrm{TC}: \mathcal{C}(X) \to [0, +\infty]$ necessarily achieves a global minimum. Therefore, the sequence $(\mathrm{TC}(\mathcal{C}^n))_{n\ge 1}$ is non-increasing (by definition of the $k$-means update step (\ref{eqn:kmeansupdate})) and bounded from below, which guarantees convergence to a local minima. 
	\end{proof}	
	
	Finally, we summarise Section \ref{subsec:kmeans} and Appendix \ref{appendix:A} with Algorithm \ref{standardkmeansalgo}. 
	
	\begin{algorithm}[h]
		\SetAlgoLined
		\KwResult{$k$ centroids}
		\textbf{initialise} centroids by sampling $k$ times from $X$\;
		\While{loss\_function $>$ tolerance}{
			\ForEach{$x_i$}{
				\textbf{assign} closest centroid wrt Euclidean distance;
			}
			\textbf{update} centroids\;
			\textbf{calculate} loss\_function\; 
		}
		\caption{Standard $k$-means algorithm}
		\label{standardkmeansalgo}
	\end{algorithm}	

	\section{The maximum mean discrepancy}\label{appendix:B}
	
	In this section, we include extra details regarding the derivation of the MMD from Definition \ref{def:maximummeandiscrepancy}. In particular we show how one can show the MMD can be employed as a metric on the space of probability measures. 
	
	\begin{definition}[Reproducing kernel Hilbert space, \cite{aronszajn1950theory}, Section 1.1]\label{def:reproducingkernelhilbertspace}
		Suppose $\mathcal{X}$ is a non-empty set, and let $(\mathcal{H}, \langle \cdot, \cdot \rangle_\mathcal{H})$ be a Hilbert space of functions $f: \mathcal{X} \to \real$. We call a positive definite function $k: \mathcal{X} \times \mathcal{X} \to \real$ a \emph{reproducing kernel} of $\mathcal{H}$ if 
		\begin{enumerate}[label=(\roman*)]
			\item For all $x \in \mathcal{X}$, we have that $k(\cdot, x) \in \mathcal{H}$, and
			\item For all $x \in \mathcal{X}$ and $f \in \mathcal{H}$, one has that 
			\begin{equation}
				f(x) = \langle f(\cdot), k(\cdot, x) \rangle_\mathcal{H},
			\end{equation}
			referred to as the \emph{reproducing property}.
		\end{enumerate}
		We call the Hilbert space $\mathcal{H}$ associated to $k$ a \emph{reproducing kernel Hilbert space} (RKHS). 
	\end{definition}
	We can associate to each RKHS $\mathcal{H}$ the \emph{canonical feature map} given by $\phi(x) = k(\cdot, x)$. We thus have that
	\begin{equation*}
		k(x,y) = \langle k(\cdot, x), k(\cdot, y) \rangle_\mathcal{H} = \langle \phi(x), \phi(y) \rangle_\mathcal{H} \qquad \text{for all }x, y \in \mathcal{X},
	\end{equation*}

	by the reproducing property from Definition \ref{def:reproducingkernelhilbertspace}. Directly from the definition of a RKHS $\mathcal{H}$, we have the following equivalent definition. 
	\begin{definition}[\cite{berlinet2011reproducing}, Theorem 1]
		Suppose $\mathcal{H}$ is a Hilbert space. Define $\delta_x: \mathcal{H} \to \real$ to be the evaluation map. Then, $\mathcal{H}$ is a RKHS if and only if $\delta_x$ is continuous. 
	\end{definition}
	\begin{proof}
		Suppose that $\mathcal{H}$ is a RKHS. Denote by $\mathcal{H}'$ as the dual of $\mathcal{H}$. One has that 
		\begin{align}
			|\delta_x(f)| = |f(x)| &= |\langle f, k(\cdot, x) \rangle_\mathcal{H}| \nonumber \\ 
			&\le \norm{f}_\mathcal{H} \langle k(\cdot, x), k(\cdot, x) \rangle_\mathcal{H}^{1/2} \nonumber\\
			&= \sqrt{k(x,x)}\norm{f}_\mathcal{H}, \label{eqn:operatornorm}
		\end{align}
		so in particular the linear operator $\delta_x$ is bounded with operator norm equal to $\sqrt{k(x,x)}$, which is well-defined by positive definiteness of $k$. Since the upper bound in (\ref{eqn:operatornorm}) is achieved by $f = k(\cdot, x)$, $\delta_x$ is bounded with operator norm $\norm{\delta_x}_{\mathcal{H}'} = \sqrt{k(x,x)}$. Since $\delta_x$ is bounded, it is continuous. 
		
		Now suppose that $\delta_x$ is bounded, so $\delta_x \in \mathcal{H}'$. By the Riesz representation theorem, there exists an element $f_{\delta_x}\in \mathcal{H}$ such that $\delta_x(f) = \langle f, f_{\delta_x}\rangle_\mathcal{H}$. Define $k(x, x') := f_{\delta_x}(x')$. We then have that $\delta_x(f) = f(x) = \langle f, k(\cdot, x) \rangle_\mathcal{H}$ and $k(\cdot, x) \in \mathcal{H}$ by construction. Thus, properties (1) and (2) are satisfied in Definition \ref{def:reproducingkernelhilbertspace}, so $\mathcal{H}$ is a RKHS. 
	\end{proof}

	In what follows, we will choose our function class $\mathcal{F}$ from Definition \ref{def:maximummeandiscrepancy} to be the unit ball in a RKHS $\mathcal{H}$ with associated reproducing kernel 
	\begin{equation}\label{eqn:gaussiankernel}
		k(x,y) = \exp \big((2\sigma)^{-2}\norm{x-y}_{\real^d}^2\big) \qquad \text{for }\sigma > 0,
	\end{equation}
	called the \emph{Gaussian kernel}. Importantly, such a RKHS has the following property, as shown in Steinwart \cite{steinwart2001influence}.
	\begin{definition}[Steinwart \cite{steinwart2001influence}, Definition 4]\label{def:universalkernel}
		Let $(\mathcal{X}, d)$ be a compact metric space. Suppose $\mathcal{H}$ is a RKHS with associated reproducing kernel $k: \mathcal{X} \times \mathcal{X} \to \real$. We call $\mathcal{H}$ \emph{universal} if
		\begin{enumerate}[label=(\roman*)]
			\item $k(\cdot, \cdot)$ is continuous, and 
			\item $\mathcal{H}$ is dense in $C_b(\mathcal{X})$, the space of bounded continuous functions on $\mathcal{X}$, with respect to the supremum norm $\norm{\cdot}_\infty$.
		\end{enumerate}
	\end{definition}	
	This choice of RKHS ensures that the MMD is metric on the space of Borel probability measures, which allows us to conclude the following. 
	\begin{theorem}[\cite{gretton2012kernel}, Theorem 5]\label{theroem:mmdmetric}
		Let $\mathcal{F}$ be the unit ball of a universal RKHS $\mathcal{H}$ comprised of $\real$-functions on a compact space $\mathcal{X}$. Suppose $\mu, \nu \in \mathcal{P}(\mathcal{X})$ are Borel. Then $\mathrm{MMD}[\mathcal{F}, \mu, \nu] = 0$ if and only if $\mu = \nu$.
	\end{theorem}
	\begin{proof}
		For $\mu \in \mathcal{P}(\mathcal{X})$, consider the linear functional $T_\mu: \mathcal{F} \to \real$ given by $T_\mu(f) = \ex_\mu[f]$. We have that 
		\begin{equation*}
			|T_\mu(f)| = |\ex_\mu[f(x)]| \le \ex_\mu[|f(x)|] = \ex_\mu[|\langle f, k(\cdot, x) \rangle_\mathcal{H}|] \le \ex_\mu[\sqrt{k(x,x)}]\norm{f}_\mathcal{H},
		\end{equation*}
		so in particular $T_\mu$ is continuous if $k(\cdot, \cdot)$ is measurable and $\ex_\mu[\sqrt{k(x,x)}] < +\infty$. By the Riesz representation theorem, there exists a $m_\mu \in \mathcal{H}$ such that $T_\mu(f) = \langle f, m_\mu \rangle_\mathcal{H}$. In particular, 
		\begin{equation*}
		 m_\mu(x) = \langle m_\mu, k(\cdot, x) \rangle_\mathcal{H} = \ex_\mu[k(\cdot, x)] = \ex_\mu[\phi(x)].
		\end{equation*}
		We call $m_\mu$ the \emph{mean embedding} of $\mu$ in $\mathcal{H}$. From (\ref{eqn:mmd_general}), we have that
		\begin{align}
			\mathrm{MMD}^2[\mathcal{F}, \mu, \nu] &= \sup_{f \in \mathcal{F}}\Big(\ex_\mu[f(x)] - \ex_{\nu}[f(y)] \Big)^2 \nonumber \\
			&= \sup_{f \in \mathcal{F}}\Big(\langle f, m_\mu \rangle_\mathcal{H} - \langle f, m_\nu\rangle_\mathcal{H}\Big)^2 \nonumber \\
			&= \sup_{\norm{f}_\mathcal{H} \le 1} \Big(\langle m_\mu - m_\nu, f \rangle_\mathcal{H}\Big)^2 = \norm{m_\mu - m_\nu}^2_\mathcal{H}. \label{eqn:mmddualrep}
		\end{align}
		Here, we have used the fact that $\mathcal{F}$ is a unit ball in $\mathcal{H}$. Suppose that $\mu = \nu$. By (\ref{eqn:mmddualrep}), this implies $\mathrm{MMD}[\mathcal{F}, \mu, \nu] = 0$. 
		
		Now suppose that $\mu = \nu$. By universality of $\mathcal{H}$, for any $\varepsilon > 0$ and $f \in C_b(\mathcal{X})$ there exists a $g \in \mathcal{H}$ such that
		\begin{equation}\label{eqn:denserkhs}
			\norm{f-g}_\infty < \frac{\varepsilon}{2}.
		\end{equation}
		Then, we have that
		\begin{equation}\label{eqn:mmdresult}
			\big|\ex_\mu[f] - \ex_\nu[f]\big| \le \big|\ex_\mu[f] - \ex_\mu[g]\big| + \big|\ex_\mu[g] - \ex_\nu[g]\big| + \big|\ex_\nu[g] - \ex_\nu[f]\big| < \frac{\varepsilon}{2} + 0 + \frac{\varepsilon}{2} = \varepsilon,
		\end{equation}
		where we have used the fact that for measures $\mu$ and $\nu$, (\ref{eqn:denserkhs}) gives that
		\begin{equation*}
			\big|\ex[f] - \ex[g]\big| \le \ex\big[\big|f(x) - g(x)\big|\big] \le \norm{f-g}_\infty,
		\end{equation*}
		and
		\begin{equation*}
			|\ex_\mu[g] - \ex_\nu[g]| = |\langle g, m_\mu -m_\nu \rangle_\mathcal{H}| = 0
		\end{equation*} 
	
		since we assumed $\mathrm{MMD}[\mathcal{F}, \mu, \nu] = 0$. Thus $\mu=\nu$ as (\ref{eqn:mmdresult}) holds for all $f \in C_b(\mathcal{X})$.
	\end{proof}
	\begin{remark}
		Suppose $(X, \Sigma)$ is a general measurable space (and not necessarily compact). Recalling Definition (\ref{def:characteristic}), if a kernel $k$ associated to the RKHS $\mathcal{H}$ is \emph{characteristic}, so the mapping
		\begin{equation*}
			\mathcal{P}(\mathcal{X}) \ni \mu \mapsto \mathbb{E}_{X \sim \mu}[k(\cdot, X)] \in \mathbb{R}
		\end{equation*}
		is injective,  then one can conclude Theorem \ref{theroem:mmdmetric} via equation (\ref{eqn:mmddualrep}).
	\end{remark}
	
	Using the definition of the mean embedding, the fact that $m_\mu(t) = \ex_{x\sim \mu}[k(t,x)]$,
	and the reproducing property of $\mathcal{F}$, we can write (\ref{eqn:mmddualrep}) as 
	\begin{equation}\label{eqn:mmdrkhs}
		\mathrm{MMD}^2[\mathcal{F}, \mu, \nu] = \ex_{x, x' \sim \mu}[k(x, x')] - 2\ex_{x\sim \mu, y\sim \nu}[k(x, y)] + \ex_{y, y' \sim \nu}[k(y,y')]
	\end{equation}
	since, for example,
	\begin{equation*}
		\langle m_\mu, m_\mu \rangle_\mathcal{H} = \ex_{x\sim\mu}[m_\mu(x)] = \ex_{x,x' \sim \mu}[k(x,x')].
	\end{equation*}
	Given samples $x= (x_1, \dots, x_n)$ and $y=(y_1, \dots, y_m)$, a biased empirical estimate of (\ref{eqn:mmdrkhs}) is given by
	\begin{equation}\label{eqn:biasedmmd}
		\mathrm{MMD}_b[\mathcal{F}, x, y] = \Bigg[\frac{1}{n^2}\sum_{i, j = 1}^n k(x_i, x_j) - \frac{2}{mn}\sum_{i, j = 1}^{m, n}k(x_i, y_j) + \frac{1}{m^2}\sum_{i, j= 1}^m k(y_i, y_j) \Bigg]^{\tfrac{1}{2}}.
	\end{equation}
	We will use the test statistic (\ref{eqn:biasedmmd}) to evaluate the success of a given clustering algorithm. 

	\section{The Wasserstein distance}
	
	In this section, we include proofs of results regarding the Wasserstein distance. 
	
	\begin{proposition}[Wasserstein barycenter, empirical measures]\label{prop:wassersteinbarycenter}
		Suppose that $\{\mu_i\}_{1\le i\le M}$ are a family of empirical probability measures, each with $N$ atoms $\alpha^1_i, \dots, \alpha^N_i$ for $i=1,\dots, M$. Let 
		\begin{equation*}
			a_j = \mathrm{Median}(\alpha^j_1, \dots, \alpha^j_M) \qquad \text{for } j=1,\dots, N.
		\end{equation*}
		Then, the cumulative distribution function of the Wasserstein barycenter $\overline{\mu} \in \mathcal{P}_p(\real)$ over $\{\mu_i\}_{1\le i\le M}$ with respect to the 1-Wasserstein distance is given by
		\begin{equation}\label{eqn:wassbaryapp}
			\overline{\mu}\left((-\infty, x]\right) = \frac{1}{N}\sum_{i=1}^N \chi_{a_i \le x}(x).
		\end{equation}
		Moreover, $\overline{\mu}$ is not necessarily unique.
	\end{proposition}
	\begin{remark}[$p > 1$]
		When $p > 1$, the proof follows in a similar manner. One will arrive at
		\begin{equation*}
			a_j = \text{Mean}(\alpha^j_1, \dots, \alpha^j_M).
		\end{equation*}
	\end{remark}
	\begin{proof}
		Assume that $N=1$, so each measure $\mu_i$ is comprised of only one atom $\alpha_i$ for $i=1,\dots, M$. WLOG we can also assume that the sequence $(\alpha_i)_{i=1}^M$ is non-decreasing. By convexity of the function $\phi_a(x) = |x-a|$ for $a \in \mathbb{R}$, the Wasserstein barycenter will also have $N=1$ atoms. Then, by (\ref{eqn:wassdistatoms}) the problem of finding the barycenter $\overline{\mu}$ is equivalent to the optimisation
		\begin{equation}\label{eqn:wassbaryoptimisation}
			\inf_{\nu \in \mathcal{P}_1(\real)} \sum_{i=1}^M W_1(\mu_i, \nu)
			= \inf_{a \in \mathbb{R}} \sum_{i=1}^M |a-\alpha_i| = \inf_{a\in \mathbb{R}} \sum_{i=1}^M \phi_{\alpha_i}(a).
		\end{equation}
		The minimiser $a^* \in \mathbb{R}$ to the right-hand side of (\ref{eqn:wassbaryoptimisation}) is obtained by solving $df/dx(x) = 0$ over $\mathbb{R}$, where 
		\begin{equation*}
			f(x) = |x-\alpha_1| + \dots + |x-\alpha_M| = \phi_{\alpha_1}(x) + \dots + \phi_{\alpha_M}(x).
		\end{equation*}
		Since
		\begin{equation*}
			\frac{d \phi_{\alpha_i}}{dx}(x) = \text{sgn}(x-\alpha_i) \qquad \text{for }i=1,\dots, M,
		\end{equation*}
		we have that
		\begin{equation}\label{eqn:wassbaryatom}
			a^* = \arginf_{a \in \mathbb{R}} \sum_{i=1}^M |a-\alpha_i| = \text{Median}(\alpha_1, \dots, \alpha_M).
		\end{equation}
		In particular, if $M \text{ mod } 1 = 0$, then $a^* \in [\alpha_{M/2}, \alpha_{M/2 + 1})$. If $M \text{ mod } 2 = 1$, then the (unique) optimizer is given by $a^* = \alpha_{K}$ where $K=\lfloor M/2 \rfloor + 1$. Setting $a=a^*$ gives (\ref{eqn:wassbaryapp}). 
		
		If $N > 1$, then the problem of finding the Wasserstein barycenter is equivalent to
		\begin{equation*}
			\inf_{\nu \in \mathcal{P}_1(\real))} \sum_{i=1}^M \mathcal{W}_1(\mu_i, \nu) = \inf_{(a_1, \dots, a_N) \in \mathbb{R}^N} \sum_{i=1}^M \sum_{j=1}^N |a_j - \alpha^j_i|.
		\end{equation*}
		Interchanging the order of summation, we see that
		\begin{equation*}
			\inf_{(a_1, \dots, a_N) \in \mathbb{R}^N} \sum_{i=1}^M \sum_{j=1}^N |a_i - \alpha^j_i| = \sum_{j=1}^N \left( \inf_{a_j \in \mathbb{R}} \sum_{i=1}^M |a_j - \alpha_i^j| \right).
		\end{equation*}
		By applying (\ref{eqn:wassbaryatom}) to each summation over $M$, we obtain the desired result (\ref{eqn:wassbaryapp}).
	\end{proof}

	\newpage
	
	\bibliographystyle{halpha-abbrv}
	\bibliography{references}

\end{document}